%% file: main.tex
\pdfoutput=1
\documentclass[runningheads]{llncs}
\usepackage{appendix}
\usepackage[T1]{fontenc}
\usepackage{graphicx}

\usepackage{mathtools}
\usepackage{amsfonts}
\usepackage{amssymb}
\usepackage{stmaryrd}
\usepackage{csquotes}
\usepackage{proof}
\usepackage{tabularx}

\usepackage{caption} 
\usepackage{subcaption} %
\usepackage{xspace} %
\usepackage{wasysym} %
\usepackage{pifont} %
\usepackage{nicefrac} %
\usepackage{tensor} 
\usepackage{tikz}
\usepackage{pgfplots} 
\usetikzlibrary{external,patterns,arrows,backgrounds,calc,shapes,shadows,decorations.pathmorphing,decorations.pathreplacing,automata,shapes.multipart,positioning,shapes.geometric,fit,circuits,trees,shapes.gates.logic.US,fit,decorations.markings}
\usepackage[plainpages=false]{hyperref} 
\hypersetup{
    colorlinks,
    linkcolor={red!50!black},
    citecolor={blue!70!black},
    urlcolor={blue!80!black}
}
\usepackage[capitalise, noabbrev]{cleveref} 
\usepackage[disable]{todonotes}

\newcommand{\genericcomment}[3]{\todo[color=#1,size=\scriptsize,fancyline,author=#2]{#3}\xspace}

\newcommand{\kb}[1]{\genericcomment{violet!35}{KB}{#1}\xspace}

\newcommand{\tw}[1]{\genericcomment{blue!35}{TW}{#1}\xspace}

\usepackage{adjustbox}
\usepackage{thm-restate}

\input{macros}

\begin{document}
\title{J-P: MDP. FP. PP.\thanks{\setlength{\leftskip}{0em}%
		This research was partially supported by the ERC AdG project 787914 FRAPPANT, the European Union’s Horizon 2020 research and innovation programme under the Marie Skłodowska-Curie grant agreement No 101008233 (MISSION), the DFG Research Training Group 2236 UnRAVeL, and the DFF project AuRoRa.}}
\subtitle{
Characterizing Total Expected Rewards in Markov Decision Processes as Least Fixed Points\\
with an Application to\\
Operational Semantics of Probabilistic Programs (Technical Report)
}
%
%
\author{%
	Kevin Batz\inst{1}\orcidID{0000-0001-8705-2564} \and
	Benjamin Lucien Kaminski\inst{2,3}\orcidID{0000-0001-8705-2564} \and
	Christoph Matheja\inst{4}\orcidID{0000-0001-9151-0441} \and
	Tobias Winkler\inst{1}\orcidID{0000-0003-1084-6408}%
}
\authorrunning{K. Batz et al.}
%
\institute{%
	RWTH Aachen University, Germany\\
	\email{\{kevin.batz,tobias.winkler\}@cs.rwth-aachen.de} \and
	Saarland University, Saarland Informatics Campus, Germany\\
	\email{kaminski@cs.uni-saarland.de} \and 
	University College London, United Kingdom \and
	Technical University of Denmark, Kgs. Lyngby, Denmark\\
	\email{chmat@dtu.dk}}

%
\maketitle              
\begin{abstract}
Markov decision processes (MDPs) with rewards are a widespread and well-studied model for systems that make both probabilistic and nondeterministic choices.
A fundamental result about MDPs is that their minimal and maximal expected rewards satisfy Bellmann's optimality equations.
For various classes of MDPs -- notably finite-state MDPs, positive bounded models, and negative models -- expected rewards are known to be the least solution of those equations.
However, these classes of MDPs are too restrictive for probabilistic program verification. 
In particular, they assume that all rewards are finite.
This is already not the case for the expected runtime of a simple probabilisitic program modeling a 1-dimensional random walk.

In this paper, we develop a generalized least fixed point characterization of expected rewards in MDPs without those restrictions.
Furthermore, we demonstrate how said characterization can be leveraged to prove weakest-preexpectation-style calculi sound with respect to an operational MDP model.
\end{abstract}

\section{Introduction}
\input{intro}
\label{sec:intro}

\section{Markov Decision Processes}
\label{sec:prelim:mdp}
\input{mdps}

\section{Expected Rewards}
\label{sec:expected_rewards_def}
\input{expected_rewards_def}

\section{Reachability Probabilities}
\label{sec:reach_prob}
\input{rewards_properties}

\section{Basic Fixed Point Theory}
\label{sec:prelim:domain_theory}
\input{prelims}

\section{Expected Rewards via Least Fixed Points}
\label{sec:rewards_lfp}
\input{rewards_lfp}

\section{Existence of Optimal Schedulers for Expected Reward}
\label{sec:exrew_schedulers}
\input{exrew_schedulers}

\section{Probabilistic Programs}
\label{sec:pp_tick}
\input{probabilistic_programs_with_tick}

\section{Weakest Preexpectation Calculi}
\label{sec:wp_tick}
%
%
\input{wp_tick}

\section{Conclusion}

We have generalized well-established least fixed point characterizations of expected rewards of possibly infinite-state MDPs. In particular, our characterizations apply to infinite (expected) rewards, which arise naturally when reasoning about expected outcomes or more general expected values modeled by probabilistic programs. We applied those characterizations to obtain cleaned-up soundness proofs for weakest preexpectation calculi that Joost-Pieter has studied over the past decade.

\bibliographystyle{splncs04}

\bibliography{literature}

\begin{appendix}
\tw{I have enabled allowdisplaybreak for the appendix}
\allowdisplaybreaks
\input{appendix}
\end{appendix}

\end{document}

%% file: macros.tex
\newcommand{\triangleqed}{{}~{}\hfill{}$\triangleleft$}

\newcommand{\keyword}[1]{#1}
\newcommand{\keywordpl}[2]{#1}
\newcommand{\highlight}[1]{#1}

\newcommand{\setfont}[1]{\mathsf{#1}}

\newcommand{\black}[1]{\textcolor{black}{#1}}
\newcommand{\gray}[1]{\textcolor{gray}{#1}}
\newcommand{\lightgray}[1]{\textcolor{lightgray}{#1}}

\newcommand*{\powerset}[1]{\mathcal{P}\left(#1\right)}

\newcommand*{\setcomp}[2]{\left\{ #1 ~\middle|~ #2 \right\}}

\DeclareMathOperator*{\argmin}{argmin}
\DeclareMathOperator*{\mopt}{\mathsf{op}}

\newcommand{\mylambda}[1]{\ensuremath{\lambda #1.\,}}

\newcommand*{\size}[1]{|#1|}

\newcommand{\morespace}[1]{~{}#1{}~}
\newcommand{\qmorespace}[1]{\quad{}#1{}\quad}
\newcommand{\qqmorespace}[1]{\qquad{}#1{}\qquad}

\newcommand{\tforall}[1]{\text{for all}~#1}

\newcommand{\qimplies}{\quad\text{implies}\quad}
\newcommand{\qqimplies}{\qquad\text{implies}\qquad}

\newcommand{\qiff}{\qmorespace{\text{iff}}}
\newcommand{\qqiff}{\qqmorespace{\text{iff}}}
\newcommand{\qand}{\qmorespace{\text{and}}}
\newcommand{\qqand}{\qqmorespace{\text{and}}}

\newcommand{\Nats}{\mathbb{N}}
\newcommand{\nati}{\ensuremath{i}}
\newcommand{\nata}{\ensuremath{n}}
\newcommand{\natb}{\ensuremath{m}}

\newcommand{\Rats}{\mathbb{Q}}

\newcommand{\Reals}{\mathbb{R}}
\newcommand{\PosRats}{\Rats_{\geq 0}}
\newcommand{\PosRatsInf}{\PosRats^{\infty}}

\newcommand{\PosReals}{\Reals_{\geq 0}}
\newcommand{\PosRealsInf}{\PosReals^{\infty}}
\newcommand{\reala}{\ensuremath{\alpha}}

\newcommand{\proba}{\ensuremath{p}}

\newcommand{\eeq}{\morespace{=}}
\newcommand{\iin}{\morespace{\in}}
\newcommand{\qeq}{\eeq}
\newcommand{\lleq}{~{}\leq{}~}
\newcommand{\ggeq}{~{}\geq{}~}

\newcommand{\domainset}{\ensuremath{D}}
\newcommand{\domaina}{\ensuremath{d}}
\newcommand{\domainb}{\ensuremath{e}}
\newcommand{\domainc}{\ensuremath{c}}

\newcommand{\domainleq}{\ensuremath{\sqsubseteq}}
\newcommand{\ddomainleq}{\ensuremath{~{}\sqsubseteq{}~}}
\newcommand{\domaintuple}{(\domainset,\, {\domainleq})}

\newcommand{\domainsubset}{\ensuremath{S}}
\newcommand{\domainop}{\ensuremath{\Phi}}
\newcommand{\supremum}{\ensuremath{\bigsqcup}}
\newcommand{\infimum}{\ensuremath{\bigsqcap}}
\newcommand{\domainleast}{\ensuremath{\bot}}
\newcommand{\domaingreatest}{\ensuremath{\top}}
\newcommand{\lfp}{\ensuremath{\mathsf{lfp}~}}

\newcommand{\mdp}{\mathcal{M}}
\newcommand{\msa}{\ensuremath{s}}
\newcommand{\msb}{\ensuremath{t}}

\newcommand{\ms}{\ensuremath{\mathsf{S}}}
\newcommand{\mact}{\ensuremath{\mathsf{Act}}}
\newcommand{\macta}{\ensuremath{\mathfrak{a}}}

\newcommand{\mprob}{\ensuremath{P}}
\newcommand{\msuccs}[1]{\ensuremath{\mathsf{Succs}_{#1}}}
\newcommand{\mpath}{\ensuremath{\pi}}
\newcommand{\mpaths}{\ensuremath{\mathsf{Paths}}}
\newcommand{\mpathsb}[1]{\ensuremath{\mathsf{Paths}^{\leq #1}}}
\newcommand{\mpathseq}[1]{\ensuremath{\mathsf{Paths}^{= #1}}}
\newcommand{\mpathprob}[1]{\ensuremath{\mathsf{Prob}_{#1}}}
\newcommand{\mcpathprob}[1]{\ensuremath{\mathsf{Prob}}}
\newcommand{\msched}{\mathfrak{S}}
\newcommand{\mscheds}{\ensuremath{\mathsf{Scheds}}}
\newcommand{\mmscheds}{\ensuremath{\mathsf{MLScheds}}}
\newcommand{\mtarget}{\ensuremath{\mathsf{T}}}
\newcommand{\mrew}{\ensuremath{\mathsf{rew}}}
\newcommand{\mval}{\ensuremath{v}}
\newcommand{\mvals}{\ensuremath{\mathbb{V}}}
\newcommand{\mvalsset}{\ensuremath{V}}
\newcommand{\mleq}{\preceq}
\newcommand{\mmleq}{~{}\mleq{}~}
\newcommand{\mdptuple}{(\ms, \, \mact, \, \mprob)}

\newcommand{\mrprob}[3]{\mathsf{Pr}_{#1} \left( \mdp,#2 \models \text{\large$\lozenge$} #3 \right)}

\newcommand{\mminrprob}[2]{\mathsf{MinPr} \left( \mdp,#1 \models \text{\large$\lozenge$} #2 \right)}
\newcommand{\mmaxrprob}[2]{\mathsf{MaxPr} \left( \mdp,#1 \models \text{\large$\lozenge$} #2 \right)}

\newcommand{\mopmin}[2]{\tensor*[^{\textnormal{min}}_{#1}]{\Phi}{_{{#2}}}}
\newcommand{\mcop}[2]{\tensor*[^{\textnormal{}}_{#1}]{\Phi}{_{{#2}}}}

\newcommand{\mopmins}{\mopmin{\mdp}{\mrew}}
\newcommand{\mopminsiter}[1]{\tensor*[^{\textnormal{min}}_{{\mdp}}]{\Phi}{^{#1}_{{\mrew}}}}
\newcommand{\mopmax}[2]{\tensor*[^{\textnormal{max}}_{#1}]{\Phi}{_{{#2}}}}
\newcommand{\mopmaxs}{\mopmax{\mdp}{\mrew}}
\newcommand{\mopmaxsiter}[1]{\tensor*[^{\textnormal{max}}_{{\mdp}}]{\Phi}{^{#1}_{{\mrew}}}}

\newcommand{\Vars}{\ensuremath{\setfont{Vars}}}
\newcommand{\vara}{\ensuremath{x}}
\newcommand{\varb}{\ensuremath{y}}
\newcommand{\varc}{\ensuremath{z}}

\newcommand{\Vals}{\ensuremath{\setfont{Vals}}}
\newcommand{\States}{\ensuremath{\mathsf{States}}}

\newcommand{\statea}{\ensuremath{\sigma}}
\newcommand{\stateb}{\ensuremath{\tau}}
\newcommand{\Preds}{\ensuremath{\powerset{\States}}}
\newcommand{\statesubst}[2]{\ensuremath{\left[ #1 \mapsto #2 \right]}}

\newcommand{\Axa}{\ensuremath{A}} 
\newcommand{\Bxa}{\ensuremath{B}}
\newcommand{\preda}{\ensuremath{P}} 
\newcommand{\predb}{\ensuremath{Q}} 

\newcommand{\pgcl}{\ensuremath{\setfont{pGCL}}}

\newcommand{\cc}{\ensuremath{C}}

\newcommand{\SKIP}{\ensuremath{\textnormal{\texttt{skip}}}}

\newcommand{\AssignSymbol}{\mathrel{\textnormal{\texttt{:=}}}}
\newcommand{\ASSIGN}[2]{\ensuremath{#1 \AssignSymbol #2}}

\newcommand{\KWTICK}{\ensuremath{\textnormal{\texttt{reward}}}}
\newcommand{\TICK}[1]{\ensuremath{\KWTICK\, \left(\, {#1} \,\right)}}
\newcommand{\tickrew}{r}

\newcommand{\COMPOSE}[2]{\ensuremath{{#1}{\,;}~ {#2}}}

\newcommand{\PCHOICE}[3]{\ensuremath{\left\{\, {#1} \,\right\}\mathrel{\left[\,#2\,\right]}\left\{\, {#3} \,\right\}}}

\newcommand{\ndchoicesymb}{\mathrel{\Box}}
\newcommand{\NDCHOICE}[2]{\ensuremath{\left\{\, {#1} \,\right\}\ndchoicesymb\left\{\, {#2} \,\right\}}}

\newcommand{\IFSYMBOL}{\ensuremath{\textnormal{\texttt{if}}}}

\newcommand{\ELSESYMBOL}{\ensuremath{\textnormal{\texttt{else}}}}

\newcommand{\ITE}[3]{\ensuremath{\IFSYMBOL\,\left(\, {#1} \,\right)\,\left\{\, {#2} \,\right\}\,\ELSESYMBOL\,\left\{\, {#3} \,\right\}}}

\newcommand{\WHILESYMBOL}{\ensuremath{\textnormal{\texttt{while}}}}
\newcommand{\WHILE}[1]{\ensuremath{\WHILESYMBOL \left(\, {#1} \,\right)\left\{\right.}}
\newcommand{\WHILEDO}[2]{\ensuremath{\WHILESYMBOL \left(\, {#1} \,\right)\left\{\, {#2} \,\right\}}}

\newcommand{\Confs}{\ensuremath{\mathsf{Conf}}}
\newcommand{\confa}{\mathfrak{c}}
\newcommand{\term}{\ensuremath{\Downarrow}}
\newcommand{\execrel}[4]{\ensuremath{#1 ~{}\xrightarrow{#2, #3}{}~ #4}}
\newcommand{\actl}{\ensuremath{L}}
\newcommand{\actr}{\ensuremath{R}}
\newcommand{\actn}{\ensuremath{N}}
\newcommand{\opmdp}{\mathcal{O}}

\newcommand{\opsink}{\bot}
\newcommand{\oprew}{\mrew_{\pgcl}}

\newcommand{\E}{\mathbb{E}} 
\newcommand{\Eset}{\ensuremath{E}}
\newcommand{\eleq}{\sqsubseteq}
\newcommand{\eeleq}{~{}\eleq{}~}
\newcommand{\FF}{\ensuremath{X}}
\newcommand{\FG}{\ensuremath{Y}}

\newcommand{\expsubst}[2]{\ensuremath{\left[{#1} / {#2} \right]}}

\newcommand{\wpsymbol}{\mathsf{wp}}
\newcommand{\wptrans}[1]{\wpsymbol\llbracket#1\rrbracket}
\renewcommand{\wp}[2]{\wptrans{#1}\left(#2\right)}

\newcommand{\dwpsymbol}{\mathsf{dwp}}
\newcommand{\dwptrans}[1]{\dwpsymbol\llbracket#1\rrbracket}
\newcommand{\dwp}[2]{\dwptrans{#1}\left(#2\right)}

\newcommand{\awpsymbol}{\mathsf{awp}}
\newcommand{\awptrans}[1]{\awpsymbol\llbracket#1\rrbracket}
\newcommand{\awp}[2]{\awptrans{#1}\left(#2\right)}

\newcommand{\dopsymbol}{\mathsf{dop}}
\newcommand{\doptrans}[1]{\dopsymbol\llbracket#1\rrbracket}
\newcommand{\dop}[2]{\doptrans{#1}\left(#2\right)}

\newcommand{\aopsymbol}{\mathsf{aop}}
\newcommand{\aoptrans}[1]{\aopsymbol\llbracket#1\rrbracket}
\newcommand{\aop}[2]{\aoptrans{#1}\left(#2\right)}

\newcommand{\someopsymbol}{\ensuremath{\mathcal T}}
\newcommand{\someoptrans}[1]{\someopsymbol\llbracket#1\rrbracket}
\newcommand{\someop}[2]{\someoptrans{#1}\left(#2\right)}

\newcommand{\torew}[1]{\ensuremath{\mrew_{#1}}}

\newcommand{\iverson}[1]{\ensuremath{\left[ #1 \right]}}

\newcommand{\dcharfun}[2]{\tensor*[^{\dwpsymbol}_{#1}]{\Phi}{_{{#2}}}}
\newcommand{\acharfun}[2]{\tensor*[^{\awpsymbol}_{#1}]{\Phi}{_{{#2}}}}

\newcommand{\dcharfuniter}[3]{\tensor*[^{\dwpsymbol}_{#1}]{\Phi}{_{{#2}}^{#3}}}
\newcommand{\acharfuniter}[3]{\tensor*[^{\awpsymbol}_{#1}]{\Phi}{_{{#2}}^{#3}}}

%% file: intro.tex

(Discrete) probabilistic programs are ordinary programs which can base control flow on the outcome of coin flips.
They are ubiquitous in computer science. They appear, e.g., as implementations of randomized algorithms and communication protocols, as models of physical or biological processes, and as descriptions of statistical models used in artificial intelligence; see~\cite{barthe2020foundations} for an overview.

Typical questions in the design and verification of probabilistic programs are concerned with quantifying their expected -- or average -- behavior. 
Examples include the expected number of retransmissions of a protocol, the probability that a particle reaches its destination, and the expected resource consumption of an algorithm.
In the presence of truly \emph{nondeterministic} choices that the program does \emph{not} resolve by flipping coins, a typical goal is to \emph{maximize} the expected reward earned by a system, or to \emph{minimize} its expected cost over all possible ways to resolve those nondeterministic choices.

\subsubsection*{J-P.}
At least over the past decade, Joost-Pieter has pushed the boundaries of probabilistic program verification by studying (1) the semantics of probabilistic programs with intricate features, (2) rules for reasoning about them, and (3) algorithms for automating their verification.
A reoccurring question -- indeed, almost a leitmotif -- in Joost-Pieter's probabilistic program research is:%
\begin{center}
	\enquote{\emph{What about loops?}}
\end{center}%
While quickly stated, it is a profound question that quickly unravels into concrete research questions which continue to drive Joost-Pieter, such as:%
\begin{itemize}
	\item \emph{How to verify probabilistic programs with unbounded loops or recursion?}\smallskip
	\item \emph{How can model checking benefit from deductive verification and vice versa?}\smallskip
	\item \emph{How to make probabilistic program verification a \enquote{push-button} technology?}
\end{itemize}
To address these questions, Joost-Pieter (J.-P.) worked extensively on connecting probabilistic programs (PPs) and Markov decision processes (MDPs) via fixed point (FP) theory.

\subsubsection{MDP.}
Markov decision processes~\cite{BK08,puterman} with rewards (or costs) are a standard mathematical model for systems that can make both discrete probabilistic and truly nondeterministic choices.
They are for probabilistic systems what labelled transition systems are for ordinary ones.
From a probability theory perspective, MDPs are well-understood, although the standard treatment involves -- perhaps surprisingly -- advanced machinery like cylinder sets.
Nevertheless, MDPs allow for an intuitive 
understanding of expected rewards in terms possible execution paths through an MDP.
For finite MDPs, probabilistic model checkers, such as \textsc{Storm}~\cite{DBLP:conf/cav/DehnertJK017},
can compute minimal and maximal rewards fully automatically.

MDPs are a natural model for assigning \emph{operational semantics} to probabilistic programs, just like transition systems are a natural model for assigning operational semantics to ordinary programs.
For example, the \textsc{Prism} language~\cite{DBLP:conf/cpe/KwiatkowskaNP02} -- a standard format for probabilistic model checking problems -- can be viewed as a probabilistic programming language similar to Dijkstra's guarded commands~\cite{DBLP:journals/cacm/Dijkstra75}. Since it requires that all variables have bounded domains, every \textsc{Prism} program corresponds to a finite-state MDP.

\subsubsection{PP.}
Verifying programs written in more general probabilistic programming languages than \textsc{Prism} -- think: with unbounded variable domains -- is arguably harder than ordinary program verification~\cite{hardness2,hardness1}.
In particular, those programs correspond to infinite-state MDPs~\cite{gretz_op}, just like arbitrary programs correspond to infinite transition systems.
Hence, model checking by explicit state space exploration becomes infeasible.

To reason about non-probabilistic programs, one may use \emph{deductive verification} approaches, such as Hoare logic or weakest preconditions.
For probabilistic programs, Kozen was the first to propose a dynamic logic~\cite{kozen83,kozen85} for reasoning about expected values of probabilistic programs with unbounded loops.
McIver and Morgan~\cite{DBLP:series/mcs/McIverM05} extended Kozen's approach by incorporating nondeterministic choices.
They also coined the term \emph{weakest preexpectation} calculus in reminiscence of Dijkstra's weakest preconditions.
Intuitively, the weakest preexpectation $\wp{\cc}{\FF}$ of a program $\cc$ and a random variable $\FF$ maps every initial program state $\sigma$ to the (minimal or maximal) expected value of $\FF$ after termination of~$\cc$ on~$\sigma$.  
The main advantage of such \emph{$\wpsymbol$-style} calculi is that they can be defined as \emph{symbolic} -- and even entirely \emph{syntactic}~\cite{DBLP:journals/pacmpl/BatzKKM21} -- expressions \emph{by induction on the program structure}.
They are in that sense \emph{compositional}.
For example, the weakest preexpectation of the program $\PCHOICE{\cc_1}{p}{\cc_2}$, which executes~$\cc_1$ with probability $p$ and $\cc_2$ with probability $(1-p)$, is given by 
\begin{align*}
\wp{\PCHOICE{\cc_1}{p}{\cc_2}}{\FF} \qeq p \cdot \wp{\cc_1}{\FF} + (1-p) \cdot \wp{\cc_2}{\FF}~.
\end{align*}
That is, if we already know the expected value of random variable $\FF$ after termination of the subprograms $\cc_1$ and $\cc_2$, respectively, then the expected value of $\FF$ after termination of $\PCHOICE{\cc_1}{p}{\cc_2}$ is a weighted sum of those expected values.
Hence, $\wpsymbol$-style calculi enable formal reasoning about probabilistic programs. 

Joost-Pieter's works on PP verification feature advanced $\wpsymbol$-style calculi for reasoning about, amongst others, conditioning~\cite{DBLP:journals/entcs/0001KKOGM15}, expected runtimes \cite{DBLP:journals/jacm/KaminskiKMO18,DBLP:conf/esop/KaminskiKMO16}, probabilistic pointer programs~\cite{qsl_popl}, and amortized expected runtimes~\cite{aert}.
Joost-Pieter also pushed for automating verification,  e.g. through techniques for learning loop invariants~\cite{DBLP:conf/qest/GretzKM13,DBLP:conf/tacas/BatzCJKKM23}, specialized proof rules~\cite{DBLP:conf/esop/BatzKKM18,DBLP:journals/pacmpl/McIverMKK18,DBLP:journals/pacmpl/HarkKGK20}, novel verification algorithms~\cite{DBLP:conf/cav/BatzCKKMS20,DBLP:conf/cav/BatzJKKMS20}, and by developing an automated verification infrastructure~\cite{DBLP:journals/pacmpl/SchroerBKKM23}.

Throughout those works, probabilistic programs turned out to be full of subtleties.
For example, the semantics of allowing observations inside of loops is non-trivial~\cite{DBLP:journals/entcs/0001KKOGM15}. The same holds for allowing variables with mixed-sign data domains~\cite{DBLP:conf/lics/KaminskiK17}.
Joost-Pieter thus always strongly advised that the developed%
\begin{center}
	\emph{calculi must be proven sound with respect to an \underline{intuitive} operational model}.
\end{center}%
More technically put, a $\wpsymbol$-style calculus should be accompanied by a soundness theorem stating that, for every initial program state, the value obtained from the calculus corresponds to the (minimal or maximal) expected reward of an MDP representing the program's operational semantics.\footnote{%
	There are of course other reasonable 
	semantical models of probabilistic programs that one could consider as a ``ground truth'', such as 
	\emph{probabilistic Turing machines}~\cite{santos1969probabilistic}, 
	\emph{stochastic Petri nets}~\cite{westphal1970supervisory},
	\emph{measure transformers}~\cite{DBLP:conf/focs/Kozen79}, 
	\emph{probabilistic control flow graphs}~\cite{goswami1993design}, 
	\emph{stochastic lambda calculi}~\cite{DBLP:conf/popl/RamseyP02},
	or even formal approaches where one assumes that \emph{the de facto semantics is the sampler}~\cite{DBLP:conf/lics/Dahlqvist0S23}, to name only a few.
	But -- borrowing from the Rifleman's Creed~\cite{rupertus1942}:%
	\begin{center}
		\textsl{MDPs are Joost-Pieter's operational model.\\
	There are many like it, but this one is Joost-Pieter's.}
	\end{center}%
}

When attempting to connect $\wpsymbol$-style calculi to expected rewards of suitable MDPs, one encounters the following (solved) issue: An operational semantics models how a program is executed, i.e. it moves forward through a program, whereas $\wpsymbol$-style calculi push suitable objects (usually random variables) \emph{backward} through a program.
In particular, for nondeterministic choices, a semantics based on MDPs resolves all nondeterminism upfront (through a scheduler), whereas $\wpsymbol$-style calculi resolve nondeterminism on demand.

A similar issue arises when reasoning about unbounded loops. Both approaches deal with loops by \emph{unfolding} them, i.e. they treat the following two programs as equivalent:
\begin{align*}
  \WHILEDO{\textit{guard}}{\cc} 
  \quad\equiv\quad 
  \ITE{\textit{guard}}{\COMPOSE{\cc}{\WHILEDO{\textit{guard}}{\cc}}}{\SKIP}
\end{align*}
As is standard for operational semantics~\cite{winskel}, the MDP semantics unfolds loops on demand.
By contrast, $\wpsymbol$ calculi capture all unfoldings at once through taking the least \emph{fixed point} of the above equivalence, as is standard for denotational semantics~\cite{winskel}.

\subsubsection{FP.}

To deal with the above issues, soundness proofs rely on the fact that expected rewards are the least solution of Bellman's optimality equations~\cite{Bellman1957MDPs}.
In other words, both the minimal and the maximal expected reward of an MDPs can be characterized as the least fixed point of a Bellmann operator -- a recursive definition of expected rewards that is closer to $\wpsymbol$-style calculi than the more intuitive path-based definition.

Least fixed point characterizations of expected rewards have been extensively studied. 
For example, they are considered in \cite{BK08} for finite-state MDPs. Puterman gives least fixed point characterizations of (undiscounted) expected rewards for infinite-state MDPs with non-negative rewards (cf.~\cite[Theorems 7.2.3a and 7.3.3a]{puterman}).
However, there are slight mismatches between his characterizations and the precise theorems required for proving $\wpsymbol$-style calculi sound:

The fixed point characterization of \emph{maximal} expected rewards is close to Puterman's \emph{positive bounded models}~\cite[Section 7.2]{puterman}.
But, those models assume that all rewards are finite, see \cite[Assumption 7.2.1]{puterman}. This assumption is unrealistic when using $\wpsymbol$-style calculi, where $\infty$ appears naturally in the value domain of certain random variables of interest.
Puterman also requires that \emph{expected values} are finite, which is problematic in two ways: 
First, discharging this assumption requires a technique for analyzing the expected value of a probabilistic program -- exactly what we want to use weakest preexpectations for. To prove a verification technique sound, we would thus first need another sound verification technique for probabilistic programs to discharge all assumptions.
Second, infinite expected values occur naturally when considering probabilistic programs. 
For example, a probabilistic program may terminate with probability one but still have an infinite expected runtime~\cite{DBLP:conf/esop/KaminskiKMO16}.

For \emph{minimal} expected rewards, the fixed point characterization is close to the one given by Puterman for \emph{negative models}~\cite[Section 7.3]{puterman}. 
However, Puterman again rules out infinite (expected) rewards.

\subsubsection*{Goal of this Paper.}

To avoid ad-hoc arguments in future soundness proofs, we give least fixed point characterizations of expected rewards of MDPs with assumptions suitable for probabilistic program verification, in particular allowing for infinite rewards.

\subsubsection*{Contributions and Outline.}

Our contributions can be summarized as follows:

\begin{itemize}
	\item After briefly recapitulating MDPs with assumptions suitable for probabilistic program verification in \Cref{sec:prelim:mdp}, we formalize minimal and maximal total expected rewards as sums-over-all-paths in \Cref{sec:expected_rewards_def}. \Cref{sec:reach_prob} shows how those expected rewards relate to reachability probabilities for MDPs.
	\item Our main results are the least fixed point characterizations of minimal and maximal rewards (\Cref{thm:mdp:bellman_correct}) in \Cref{sec:rewards_lfp}; the necessary background on fixed point theory is presented in \Cref{sec:prelim:domain_theory}. 
	\item We show that optimal schedulers for minimal expected rewards always exists in \Cref{sec:exrew_schedulers}. 
	For maximal expected rewards, this is not the case.
	\item We demonstrate the usefulness of our least fixed point characterizations by proving the soundness of two weakest preexpectation calculi for probabilistic programs with user-defined reward structures in \Cref{sec:pp_tick}. 
\end{itemize}

\subsubsection*{Origins of this Paper.} Parts of this article are an extension of \cite[Chapter 2.2]{batz_diss}. More specifically, \Cref{lem:mdp:bellman-step} is a generalization of \cite[Lemma 2.8]{batz_diss} and \Cref{thm:mdp:bellman_correct} is a generalization of \cite[Theorem 2.9]{batz_diss} from expected reachability-rewards to total expected rewards. 

%% file: mdps.tex

An MDP can be thought of as a probabilistic extension of a labelled transition system. 
We consider here \emph{finitely branching} processes with \emph{countable} state spaces. 
It goes without saying that we follow the presentation of \cite[Chapter 10]{BK08}.
If the reader is familiar with \cite{BK08}, almost everything in this section should be completely standard, except that we do not designate an initial state.
The latter is motivated by the fact that the probabilistic programs we consider later in \Cref{sec:pp_tick} do not have a fixed initial state either.%
%
%
\begin{definition}[MDPs]
	\label{def:prelim:mdps}
	A \emph{\highlight{Markov decision process (MDP)}}%
	\belowdisplayskip=0pt%
	\begin{align*}
		\highlight{\mdp} \qeq \mdptuple~,
	\end{align*}%
	\normalsize%
	is a structure where:%
	\begin{enumerate}
		\item 
			$\highlight{{\ms}}$ is a countable set of \emph{\highlight{states}}.
			If $\ms$ is finite, then we say that $\mdp$ is \emph{finite-state}. Otherwise, we say that $\mdp$ is \emph{infinite-state}.
		\item 
			$\highlight{{\mact}}$ is a finite set of \emph{\highlight{actions}}.
			\item 
				$\highlight{\mprob} \colon \ms \times \mact \times \ms \to [0,1]$ is a \emph{\highlight{transition probability function}} such that:%
				\begin{enumerate}
					\item 
						For all states $\msa \in \ms$ and actions $\macta \in \mact$, the function $\mprob(\msa,\, \macta,\, {-})$ is either a (full) probability distribution over finitely many successor states, or a null distribution.
						Formally:
				\abovedisplayskip=0pt%
				\begin{align*}
					\gray{
						\underbrace{
							\black{
								\setcomp{\msa' \in \ms}{\mprob(\msa,\macta,\msa') > 0} \textnormal{ is finite}
							}
						}_{
							\textnormal{\tiny probability distribution has finite support}
						}
					}
					\qqand
					\gray{
						\overbrace{
							\black{
								\sum_{\msa'\in\ms} \mprob(\msa,\macta,\msa') \iin \{0,1\}
							}
						}^{
							\mathclap{\substack{\textnormal{\tiny $\macta$ maps to full probability distribution} \\ \substack{\textnormal{\tiny or null distribution}}}}
						}
					}
				\end{align*}%
				\normalsize%
				We say that action $\macta$ is \emph{enabled} in state $\msa$ if $\sum_{\msa'\in\ms} \mprob(\msa,\macta,\msa') = 1$ and denote the set of \highlight{\emph{actions enabled in $\msa$}} by \highlight{$\mact(\msa)$}.
			\item 
				Every state $\msa \in \ms$ has at least one enabled action, i.e.\ $|\mact(\msa)| \geq 1$.
		\end{enumerate}
	\end{enumerate}%
	If $\mprob(\msa,\macta,\msa')>0$, then $\msa'$ is called an \emph{$\macta$-successor of $\msa$} and we denote the finite \highlight{\emph{set of $\macta$-successors of $\msa$}} by \highlight{$\msuccs{\macta}(\msa)$}. 
\end{definition}%
\begin{example}[Markov Decision Processes]
	\Cref{fig:running_example_mdp} depicts an infinite-state MDP with states $\ms = \{s_i,\, s^L_i,\, s^R_i \mid i \in \Nats \} \cup \{\opsink\}$ and actions $\mact = \{\actn,\,\actl,\,\actr\}$.
	The transition probability function $\mprob$ takes, e.g., values $\mprob(s_0,\, \actl,\, s_0^L) = 1$ and $\mprob(s_0,\, \actr,\, s_0^R) = \nicefrac 1 2$, etc.
	\triangleqed%
\end{example}%
\begin{figure}[t]%
	\centering%
	\input{mdp_example}%
	\caption{%
		Example MDP with infinite state space.
		Omitted probabilities are 1.
		Numbers in boxes next to states -- like \fbox{$r\vphantom{2}$} or \fbox{$2$} -- are rewards (see \Cref{sec:expected_rewards_def}).%
	}%
	\label{fig:running_example_mdp}%
\end{figure}
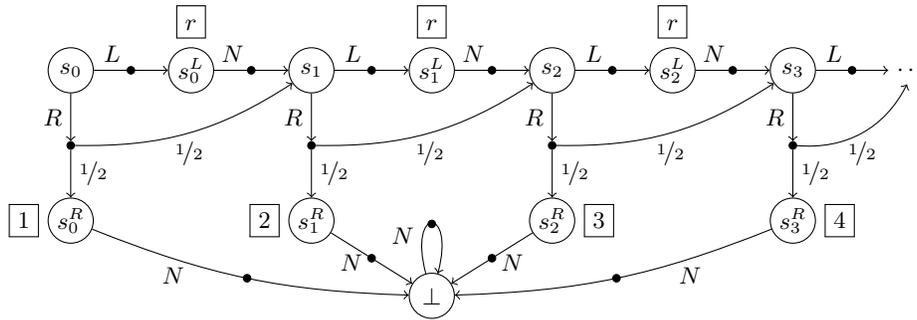%
An MDP moves between states by nondeterministically selecting an action that is enabled in the current state. 
The next state is then chosen probabilistically according to the transition probability function for that action.
To \enquote{execute} an MDP, we move from state to state, thus creating \emph{paths} of visited states.
Towards a formal definition, we denote by $\ms^+$ the set of \emph{non-empty} sequences over $\ms$.
\begin{definition}[Paths]%
	Let $\mdp = \mdptuple$ be an MDP.
	\begin{enumerate}
		\item A \emph{\highlight{path}} in $\mdp$ is a f\underline{inite}, non-empty sequence of states 
		\[
		\highlight{\mpath} \eeq \msa_0\ldots\msa_\nata \iin \ms^+
		\]
		such that there is for each state $\msa_\nati$, with $\nati\in\{0,\, \ldots,\,\nata-1\}$, an enabled action $\macta \in \mact(\msa_\nati)$ such that $\msa_{\nati+1}$ is an $\macta$-successor of $\msa_\nati$, i.e., $\msa_{\nati+1} \in \msuccs{\macta}(\msa_\nati)$. 
        \item We define the set of \emph{\highlight{all paths of length exactly $\nata\in\Nats$ starting in $\msa \in \ms$}} as 
        \[
        \highlight{\mpathseq{\nata} (\msa)} \qeq \setcomp{\mpath}{\textnormal{$\mpath = \msa_0\ldots\msa_\nata$ is a path in $\mdp$ and }\msa_0 = \msa}~.
        \]
		\item We define the set of \emph{\highlight{paths of length at most $\nata\in\Nats$ starting in $\msa \in \ms$}} as 
		\[
		\highlight{\mpathsb{\nata} (\msa)} \qeq \setcomp{\mpath}{\textnormal{$\mpath = \msa_0\ldots\msa_m$ is a path in $\mdp$ and }\msa_0 = \msa~\textnormal{and}~m\leq n}~.
		\]
		\item We define the set of \emph{\highlight{all (finite) paths starting in state $\msa \in \ms$}} as 
		\[
		\highlight{\mpaths(\msa)} \qeq \bigcup_{\nata \in \Nats} \mpathsb{\nata} (\msa)~.
		\]
	\end{enumerate}
\end{definition}%
Note that all of the above sets of paths are countable. 
Our goal is to formalize expected rewards in terms of probabilities and rewards of paths.
However, if multiple actions are enabled for some state $s$, the probability of moving from $s$ to some other state $s'$ -- and the probability of paths including $s s'$  -- depends on the nondeterminstically chosen action.
To assign \emph{unique} probability measures to paths, we must first resolve all nondeterminism. To perform this resolution, we employ \emph{schedulers}, which determine the next action based on the previously visited states.
%
%
\begin{definition}[Schedulers]
	A \emph{\highlight{scheduler}} for an MDP $\mdp = \mdptuple$ is a function%
	\footnote{We consider only \emph{deterministic} schedulers. Some applications also consider \emph{randomized} schedulers.}
	\[
		\highlight{\msched} \colon \ms^+ \to \mact
	\]%
	which for each history of previously visited states $\msa_0\ldots\msa_\nata \in \ms^+$ always chooses an action enabled in $\msa_\nata$, i.e. $\tforall{\msa_0\ldots\msa_\nata \in \ms^+}$, we have $\msched(\msa_0\ldots\msa_\nata ) \in \mact(\msa_\nata)$.
	We denote the \emph{\highlight{set of all schedulers for $\mdp$}} by $\highlight{\mscheds}$ because $\mdp$ will always be clear from context. 
	
	A scheduler $\highlight{\msched}$ is called \emph{\highlight{memoryless}}, if its choices depend solely on the last state of the history, i.e. for all $\msa_0\ldots\msa_\nata \in \ms^+$ and $\msb_0\ldots\msb_\natb \in \ms^+$, we have
	\begin{align*}
		\qquad \msa_\nata = \msb_\natb 
		\qimplies 
		\msched(\msa_0\ldots\msa_\nata) = \msched(\msb_0\ldots\msb_\natb ) ~.
	\end{align*}
	Otherwise, $\msched$ is called \emph{\highlight{history-dependent}}.
    A memoryless scheduler can be identified with a function of type $\ms \to \mact$. We denote the \emph{\highlight{set of all memoryless schedulers for $\mdp$}} by $\highlight{\mmscheds}$ (again $\mdp$ will be clear from context).%
\end{definition}%
A scheduler $\msched$ resolves the nondeterminism in an MDP $\mdp$ as follows: 
Starting in some initial state $\msa_0$, $\msched$ first chooses \emph{deterministically} the action $\msched(\msa_0)$. 
If the thereafter \emph{probabilistically} drawn $\msched(\msa_0)$-successor of $\msa_0$ is $\msa_1$, then $\msched$ next  chooses action $\msched(\msa_0\msa_1)$, and so on. 
If $\msched$ is memoryless, the chosen action only depends on the current state $\msa_1$.

Since a scheduler $\msched$ resolves the nondeterminism in $\mdp$, it makes now sense to speak about \emph{\underline{the} probability of a path in $\mdp$ under $\msched$}.%
\begin{definition}[Path Probabilities]%
	\label{def:prelims:inducedprob}%
	Let $\mdp = \mdptuple$ be an MDP, let \mbox{$\msched \in \mscheds$} be a scheduler, and let $\msa_0\ldots \msa_\nata \in \ms^+$. The \emph{\highlight{probability of $\msa_0\ldots \msa_\nata $ in $\mdp$ under $\msched$}} is
	\abovedisplayskip=0pt%
	\begin{align*}
		\highlight{\mpathprob{\msched}(\msa_0\ldots \msa_\nata)} \qeq \prod_{\nati = 0}^{\nata-1} \mprob(\msa_\nati, \msched(\msa_0\ldots\msa_\nati), \msa_{\nati +1})~,
	\end{align*}%
	\normalsize%
	where the empty product (i.e., in case $n=0$) equals $1$.%
\end{definition}%
%
%
%
%
%
\begin{remark}
    \label{rem:path-probs}
    Somewhat contrary to what the notation suggests, $\mpathprob{\msched}$ as defined above is \emph{not} actually a probability distribution over the (countable) set of all finite paths.
    This is for two reasons:
    First, as our MDPs have no initial state, we have $\mpathprob{\msched}(\msa) = 1$ for \emph{every} $\msa \in \ms$.
    It is thus clear that $\mpathprob{\msched}$ cannot be a sensible distribution if $|\ms| > 1$.
    But even for a fixed initial state $\msa_0 \in \ms$, $\mpathprob{\msched}$ is not a distribution over $\mpaths(\msa_0)$.
    In fact, we have $\sum_{\mpath \in \mpaths(\msa_0)} \mpathprob{\msched}(\mpath) = \infty$.
    However, $\mpathprob{\msched}$ \emph{is} a meaningful  distribution over $\mpathseq{\nata}(\msa_0)$ for every $\nata \in \Nats$.
    The path probability $\mpathprob{\msched}(\msa_0\ldots \msa_\nata)$ can then be understood as the probability that the path $\msa_0\ldots \msa_\nata$ occurs assuming that
    \begin{itemize}
        \item $\msched$ is used to resolve all nondeterminism,
        \item the MDP is started in $\msa_0$, and
        \item the MDP is executed for exactly $\nata$ steps.
    \end{itemize}
    Finally, we remark that constructing a probability measure over (measurable) sets of infinite paths in line with our path probabilities is standard, see~\cite[Chaper 10]{BK08}.
    We omit such constructions to keep the presentation more elementary.
    \triangleqed%
\end{remark}

%% file: mdp_example.tex
\begin{tikzpicture}[on grid, node distance=10mm and 16mm,actstyle/.style={draw,circle,fill=black,inner sep=1pt},every state/.style={inner sep=0pt, minimum size=6mm}]
    \node[state] (0) {$s_0$};
    \node[state,right=of 0,label=90:\fbox{$r\vphantom{2}$}] (1) {$s_0^L$};
    \node[state,right=of 1] (2) {$s_1$};
    \node[state,right=of 2,label=90:\fbox{$r\vphantom{2}$}] (3) {$s_1^L$};
    \node[state,right=of 3] (4) {$s_2$};
    \node[state,right=of 4,label=90:\fbox{$r\vphantom{2}$}] (5) {$s_2^L$};
    \node[state,right=of 5] (6) {$s_3$};
    \node[right=of 6] (dots) {$\cdots$};
    
    \node[actstyle,below=of 0] (c0) {};
    \node[actstyle,below=of 2] (c1) {};
    \node[actstyle,below=of 4] (c2) {};
    \node[actstyle,below=of 6] (c3) {};
    
    \node[state,below=of c0,label=180:\fbox{$1$}] (7) {$s_0^R$};
    \node[state,below=of c1,label=180:\fbox{$2$}] (8) {$s_1^R$};
    \node[state,below=of c2,label=0:\fbox{$3$}] (9) {$s_2^R$};
    \node[state,below=of c3,label=0:\fbox{$4$}] (10) {$s_3^R$};
    
    \node[state,below=30mm of 3] (sink) {$\opsink$};
    
    \draw[->] (0) edgenode[actstyle,pos=0.5] {} node[pos=0.25,above]{$\actl$} node[pos=0.75,above]{} (1);
    \draw[->] (1) edgenode[actstyle,pos=0.5] {} node[pos=0.25,above]{$\actn$} node[pos=0.75,above]{} (2);
    \draw[->] (2) edgenode[actstyle,pos=0.5] {} node[pos=0.25,above]{$\actl$} node[pos=0.75,above]{} (3);
    \draw[->] (3) edgenode[actstyle,pos=0.5] {} node[pos=0.25,above]{$\actn$} node[pos=0.75,above]{} (4);
    \draw[->] (4) edgenode[actstyle,pos=0.5] {} node[pos=0.25,above]{$\actl$} node[pos=0.75,above]{} (5);
    \draw[->] (5) edgenode[actstyle,pos=0.5] {} node[pos=0.25,above]{$\actn$} node[pos=0.75,above]{} (6);
    \draw[->] (6) edgenode[actstyle,pos=0.5] {} node[pos=0.25,above]{$\actl$} node[pos=0.75,above]{} (dots);
    
    \draw[->] (0) edge node[left]{$\actr$} (c0);
    \draw[->] (2) edge node[left]{$\actr$} (c1);
    \draw[->] (4) edge node[left]{$\actr$} (c2);
    \draw[->] (6) edge node[left]{$\actr$} (c3);
    
    \draw[->] (c0) edge node[right]{$\nicefrac 1 2$} (7);
    \draw[->] (c1) edge node[right]{$\nicefrac 1 2$} (8);
    \draw[->] (c2) edge node[right]{$\nicefrac 1 2$} (9);
    \draw[->] (c3) edge node[right]{$\nicefrac 1 2$} (10);
    
    \draw[->] (c0) edge[bend right=17] node[below]{$\nicefrac 1 2$} (2);
    \draw[->] (c1) edge[bend right=17] node[below]{$\nicefrac 1 2$} (4);
    \draw[->] (c2) edge[bend right=17] node[below]{$\nicefrac 1 2$} (6);
    \draw[->] (c3) edge[bend right=34] node[below]{$\nicefrac 1 2$} (dots);
    
    \draw[->,bend right=10] (7) edgenode[actstyle,pos=0.5] {} node[pos=0.25,below]{$\actn$} node[pos=0.75,below]{} (sink);
    \draw[->] (8) edgenode[actstyle,pos=0.5] {} node[pos=0.25,below]{$\actn$} node[pos=0.75,below]{} (sink);
    \draw[->] (9) edgenode[actstyle,pos=0.5] {} node[pos=0.25,below]{$\actn$} node[pos=0.75,below]{} (sink);
    \draw[->,bend left=10] (10) edgenode[actstyle,pos=0.5] {} node[pos=0.25,below]{$\actn$} node[pos=0.75,below]{} (sink);
    
    \path[->]
    (sink) edge[loop above,looseness=15] node[pos=0.5,actstyle,yshift=-2pt]  {} node[pos=0.25,left] {$\actn$} node[pos=0.75, right] {} (sink)
    ;

\end{tikzpicture}

%% file: expected_rewards_def.tex

\newcommand{\newmexprewb}[4]{\mathsf{ER}_{#1}^{= #2} \left( \mdp,#4, #3  \right)}
\newcommand{\newmexprew}[3]{\mathsf{ER}_{#1} \left( \mdp,#3,#2 \right)}
\newcommand{\newmminexprewb}[3]{\mathsf{MinER}^{= #1} \left( \mdp,#3,#2 \right)}
\newcommand{\newmmaxexprewb}[3]{\mathsf{MaxER}^{= #1} \left( \mdp,#3,#2 \right)}
\newcommand{\newmminexprew}[2]{\mathsf{MinER} \left( \mdp,#2,#1 \right)}
\newcommand{\newmmaxexprew}[2]{\mathsf{MaxER} \left( \mdp,#2, #1 \right)}
\newcommand{\newmcexprew}[2]{\mathsf{ER} \left( \mdp,#2,#1 \right)}
\newcommand{\newmcexprewb}[3]{\mathsf{ER}^{= #1} \left( \mdp,#3, #2 \right)}

\newcommand{\newmexprewbs}[4]{\mathsf{ER}_{#1}^{= #2} \left( #4, #3  \right)}
\newcommand{\newmexprews}[3]{\mathsf{ER}_{#1} \left( #3,#2 \right)}
\newcommand{\newmminexprewbs}[3]{\mathsf{MinER}^{= #1} \left( #3,#2 \right)}
\newcommand{\newmmaxexprewbs}[3]{\mathsf{MaxER}^{= #1} \left( #3,#2 \right)}
\newcommand{\newmminexprews}[2]{\mathsf{MinER} \left( #2,#1 \right)}
\newcommand{\newmmaxexprews}[2]{\mathsf{MaxER} \left( #2, #1 \right)}
\newcommand{\newmcexprews}[2]{\mathsf{ER} \left( #2,#1 \right)}
\newcommand{\newmcexprewbs}[3]{\mathsf{ER}^{= #1} \left( #3, #2 \right)}

\newcommand{\newmmexprew}[4]{\mathsf{ER}_{#1} \left( #4,#3,#2 \right)}
\newcommand{\newmmexprewb}[5]{\mathsf{ER}_{#1}^{= #2} \left( #5,#4, #3  \right)}
\newcommand{\newmmminexprewb}[4]{\mathsf{MinER}^{= #1} \left( #4,#3,#2 \right)}
\newcommand{\newmmmaxexprewb}[4]{\mathsf{MaxER}^{= #1} \left( #4,#3,#2 \right)}
\newcommand{\newmmminexprew}[3]{\mathsf{MinER} \left( #3,#2,#1 \right)}
\newcommand{\newmmmaxexprew}[3]{\mathsf{MaxER} \left( #3,#2, #1 \right)}
\newcommand{\newmmcexprew}[3]{\mathsf{ER} \left( #3,#2,#1 \right)}
\newcommand{\newmmcexprewb}[4]{\mathsf{ER}^{= #1} \left( #4,#3, #2 \right)}

We now equip MDPs with a reward structure and formalize \emph{expected rewards} based on the probability of and the reward collected along paths.
In contrast to definitions found in the literature, e.g. \cite{BK08,puterman}, our rewards can be any extended real in $\PosRealsInf$. Apart from that, our definition of expected (total) rewards is standard (cf. \cite[p. 77 - 79]{puterman}).

We first formalize summation over countable index sets.
Throughout this paper, we let $\reala + \infty = \infty + \reala = \infty$ for all $\reala \in \PosRealsInf$, and moreover $0\cdot \infty = \infty \cdot 0 = 0$. 
Following \cite[Section 6.2]{countable_sums}, given a countable set $A$ and $g\colon A \to \PosRealsInf$, we define%
\[
\sum_{\alpha \in A}  g(\alpha) \eeq \sup \setcomp{\vcenter{\hbox{$\displaystyle\sum_{\alpha \in F} g(\alpha)$}}}{\text{$F$ finite and $F\subseteq A$}}~,
\]%
which is a well-defined number in $\PosRealsInf$.
An important feature of the above sum is that even for countably \emph{infinite} $A$, the order of summation is irrelevant: 
For every bijective function $\mathsf{enum} \colon \Nats \to A$ (i.e.\ every permutation of $A$), we have (by \cite[Proposition 6.1]{countable_sums})%
\[
\sum_{\alpha \in A} g(\alpha) \eeq \sum_{\nati=0}^{\omega} g(\mathsf{enum}(\nati))
\textcolor{lightgray}{\eeq \sup_{\nata \in \Nats}\, \sum_{\nati=0}^{\nata} g(\mathsf{enum}(\nati))}~,
\]%
where $\omega$ is the first infinite ordinal number.
With this notion of countable summation at hand, we now define reward functions and expected rewards of MDPs.%
\begin{definition}[Reward Functions]
	A function%
	\[
		\highlight{\mrew} \colon \ms \to \PosRealsInf
	\]%
	for an MDP $\mdp = \mdptuple$ 
	mapping MDP states to non-negative extended reals is called a \emph{\highlight{reward function (for~$\mdp$)}}.
	We call $\mrew(\msa)$ the \emph{reward of state $\msa$}.%
\end{definition}%
Intuitively, we think of $\mrew(\msa)$ as a reward that is collected upon \emph{entering} state $\msa$ as one progresses along a path through the MDP.
We thus extend the definition of $\mrew$ from states to paths of $\mdp$:
For $\msa_0\ldots\msa_\nata \in \ms^{+}$, we define%
\[
\mrew(\msa_0\ldots\msa_\nata) \eeq \sum_{\nati = 0}^\nata \mrew(\msa_\nati) ~.
\]%
The above sum is never empty as every path contains at least one state.

Now that we have assigned a reward to each path and we have earlier assigned a probability to each path (under a given scheduler), we associate each state of an MDP with a minimal and a maximal \emph{expected reward}.\footnote{
Similarly to our definition of path probabilities, we consider it more intuitive to give direct path-based definitions of expected values instead of defining them in the framework of probability theory, i.e.\ via measurable functions and Lebesgue integrals.
}
\begin{definition}[Expected Rewards]%
\label{def:mdp:exrew}%
	Let $\mrew \colon \ms \to \PosRealsInf$ be a reward function for an MDP $\mdp = \mdptuple$ and $\msched \in \mscheds$. 
	We define the following:%
	\begin{enumerate}
		\item
		\label{def:mdp:exrew1} 
			The \emph{\highlight{expected reward from $\msa \in \ms$ after $\nata \in \Nats$ steps under $\msched$}} is defined as%
			\[ 
				\highlight{\newmexprewb{\msched}{\nata}{\mrew}{\msa}} \qeq \sum_{\msa_0\ldots\msa_\nata \in \mpathseq{\nata}(\msa)}
				\mpathprob{\msched}(\msa_0\ldots\msa_\nata )\cdot \mrew(\msa_0\ldots\msa_\nata ) ~.
			\]%
			Notice that the number of summands is finite.%
        			
		\item
		\label{def:mdp:exrew2} 
			The \emph{\highlight{expected reward from $\msa \in \ms$ under $\msched$}} is defined as
			\[
				\highlight{\newmexprew{\msched}{\mrew}{\msa}} \qeq \sup_{\nata \in \Nats}\, \newmexprewb{\msched}{\nata}{\mrew}{\msa}~.
			\]%

		\item
		\label{def:mdp:exrew3} 
			The \emph{\highlight{minimal expected reward from $\msa \in \ms$ after $\nata \in \Nats$ steps}} is defined as%
			\[
				\highlight{\newmminexprewb{\nata}{\mrew}{\msa}} \qeq \inf_{\msched \in \mscheds}  \newmexprewb{\msched}{\nata}{\mrew}{\msa} ~.
			\]%

		\item
		\label{def:mdp:exrew4}%
			The \emph{\highlight{minimal expected reward from $\msa \in \ms$}} is defined as%
			\[
				\highlight{\newmminexprew{\mrew}{\msa}} \qeq \inf_{\msched \in \mscheds}  \newmexprew{\msched}{\mrew}{\msa} ~.
			\]%

		\item
		\label{def:mdp:exrew5}%
			The \emph{\highlight{maximal expected reward from $\msa \in \ms$ {after $\nata \in \Nats$ ste}ps}} is defined as%
			\[
				\highlight{\newmmaxexprewb{\nata}{\mrew}{\msa}} \qeq \sup_{\msched \in \mscheds}  \newmexprewb{\msched}{\nata}{\mrew}{\msa} ~.
			\]%

		\item
		\label{def:mdp:exrew6}%
			The \emph{\highlight{maximal expected reward from $\msa \in \ms$}} is defined as%
			\[
				\highlight{\newmmaxexprew{\mrew}{\msa}} \qeq \sup_{\msched \in \mscheds}  \newmexprew{\msched}{\mrew}{\msa} ~.
			\]%
	\end{enumerate}%
\end{definition}%
If clear from context, we omit $\mdp$ from the notation above.

We briefly go over the above notions of expected rewards:
Upon reaching a state $\msa' \in \ms$, the MDP collects reward $\mrew(\msa')$.
Then $\newmexprewbs{\msched}{\nata}{\mrew}{\msa}$ is the expected reward collected after exactly $n$ steps if we start the MDP in state $\msa$ and resolve all nondeterminism according to $\msched$.
As the number of steps the process progresses tends to $\omega$, i.e. \enquote{in the limit}, we obtain the (total) expected reward $\newmexprews{\msched}{\mrew}{\msa}$.
$\newmminexprews{\mrew}{\msa}$ is the minimal -- under all possible choices of nondeterminism-resolving schedulers -- expected reward when starting in state $\msa$.
Analogously, $\newmmaxexprews{\mrew}{\msa}$ is the maximal\footnote{For historical reasons, we will speak of maximal expected rewards, even though the supremum in \Cref{def:mdp:exrew}.\ref{def:mdp:exrew6} is not necessarily a maximum (cf.\ \Cref{sec:exrew_schedulers}).}
expected reward. $\newmminexprewbs{\nata}{\mrew}{\msa}$ and $\newmmaxexprewbs{\nata}{\mrew}{\msa}$ are the corresponding step-bounded variants. %
Finally, we call a scheduler \emph{min-optimal} (resp. \emph{max-optimal}) if the infimum in $\newmminexprews{\mrew}{\msa}$ (resp. the supremum in $\newmmaxexprews{\mrew}{\msa}$) is attained for that scheduler. 

\begin{example}[Expected Rewards]
	\label{ex:maxERminER}
	Consider the MDP in \Cref{fig:running_example_mdp} together with the reward function $\mrew(s_i^L) = r$, for some constant $r \in \PosRealsInf$, $\mrew(s_i^R) = i+1$, and $\mrew(s_i) = 0$, for all $i \in \Nats$ (see numbers in boxes next to states in \Cref{fig:running_example_mdp}).
	We have the following:%
	\begin{itemize}
		\item 
			If $r > 0$ then $\newmmaxexprews{\mrew}{s_0} = \infty$ as witnessed by the strategy that always stays in the topmost row (cf.\ \Cref{ex:memoryless_scheduler_induced_mc}).
			On the other hand, to \emph{minimize} the expected reward, it is best to reach the sink state $\opsink$ as early as possible, i.e.\ to follow the scheduler $\msched$ with $\msched(s_i) = \actr$ for all $i \in\Nats$.
			Hence,
			\[
				\newmminexprews{\mrew}{s_0} \eeq \frac 1 2 \cdot 1 + \frac 1 4 \cdot 2 + \frac 1 8 \cdot 3 + \ldots \eeq 2~.
			\]
		
		\item 
			Now, let $r=0$.
			Then the minimal expected reward is $\newmminexprews{\mrew}{s_0} = 0$.
			For the maximal expected reward, we claim that still $\newmmaxexprews{\mrew}{s_0} = \infty$.
			To see this, consider for every $k \in \Nats$ the unique scheduler $\msched_k$ with $\msched_k(s_i) = \actl$ for all $i < k$ and $\msched_k(s_j) = \actr$ for all $j \geq k$.
			Note that $\newmexprewbs{\msched_k}{2(k+1)}{\mrew}{s_0} = \frac 1 2 (k+1)$.
			Thus:%
			\begin{align*}
				\newmmaxexprews{\mrew}{s_0}
				\eeq & \sup_{\msched \in \mscheds}  \newmexprews{\msched}{\mrew}{s_0} \\
				\ggeq & \sup_{k \in \Nats}\,  \newmexprews{\msched_k}{\mrew}{s_0} \\
				\eeq & \sup_{k \in \Nats}\, \sup_{n \in \Nats}\,  \newmexprewbs{\msched_k}{n}{\mrew}{s_0} \\
				\eeq & \sup_{k \in \Nats}\, \sup_{n \in \Nats}\,  \newmexprewbs{\msched_k}{2(n+1)}{\mrew}{s_0} \\
				\ggeq & \sup_{k \in \Nats}\,  \newmexprewbs{\msched_k}{2(k+1)}{\mrew}{s_0}
				\eeq  \sup_{k \in \Nats}\,  \frac 1 2 (k+1) \eeq \infty
				~.
			\end{align*}%
			However, it can be shown that $\newmexprews{\msched_k}{\mrew}{s_0} <\infty$ for all $k \in \Nats$, i.e.\ none of the schedulers $\msched_k$ attains the maximal expected reward.
			In general, max-optimal schedulers are not guaranteed to exist in infinite-state MDPs.
			In \emph{this} example, however, there does exist a max-optimal scheduler, see \Cref{ex:opt-sched}.%
			\triangleqed%
	\end{itemize}
\end{example}%
%
%
%
%

%% file: rewards_properties.tex

Readers with a probabilistic model checking background may perhaps be more familiar with \emph{reachability probabilities} or \emph{expected rewards upon reaching a set of target states} than with the (total) expected rewards we formalized in the previous section.
Those notions can be recovered from expected rewards. 
For example, reachability probabilities can be defined in terms of expected rewards as follows (cf. \cite[Chapter 10]{BK08}):
%
%
%
\begin{definition}[Reachability Probabilities]%
	Let $\mdp = \mdptuple$ be an MDP and let $\mtarget \subseteq \ms$ be a set of \emph{target states}.
	Consider the MDP $\mdp_\mtarget = (\ms \cup \{\opsink\},\, \mact,\, \mprob')$, where%
	\[
		\mprob'(\msa,\macta,\msb) 
		\eeq
		\begin{cases}
			\mprob(\msa,\macta,\msb), 	&\text{if } \msa \in \ms \setminus \mtarget \text{ and } \msb \neq \opsink	\\
			0, 						&\text{if } \msa \in \ms \setminus \mtarget \text{ and } \msb = \opsink		\\
			0,						&\text{if } \msa \in \mtarget \cup \{\opsink\} \text{ and } \msb \neq \opsink	\\
			1, 						&\text{if } \msa \in \mtarget \cup \{\opsink\} \text{ and } \msb = \opsink 	~.
		\end{cases}
	\]%
	Intuitively, $\mdp_\mtarget$ is obtained from $\mdp$ by adding a fresh sink state $\opsink$ and redirecting all outgoing transitions from target states in $\mtarget$ to the sink $\opsink$.
	
	Consider furthermore the reward function $\mrew_\mtarget$ for $\mdp_\mtarget$ given by%
	\[
		\mrew_\mtarget(\msa)
		\eeq
		\begin{cases}
			1, 	&\text{if } \msa \in \mtarget	\\
			0, 	&\text{if } \msa \notin \mtarget	~.
		\end{cases}
	\]%
	Then we define the \emph{\highlight{probability to reach $\mtarget$ from a state $\msa$ under a scheduler $\msched$}} as%
	\[
	\highlight{\mrprob{\msched}{\msa}{\mtarget}} \qeq  \newmmexprew{\msched}{\mrew_\mtarget}{\msa}{\mdp_\mtarget}~.
	\]
	The \emph{\highlight{minimal probability to reach $\mtarget$ from a state $\msa$}} is defined as 
	\[
	\highlight{\mminrprob{\msa}{\mtarget}}  \qeq \newmmminexprew{\mrew_\mtarget}{\msa}{\mdp_\mtarget} ~.
	\]
	Analogously, the \emph{\highlight{maximal probability to reach $\mtarget$ from a state $\msa$}} \mbox{is defined as}
	\[
	\highlight{\mmaxrprob{\msa}{\mtarget}}  \qeq \newmmmaxexprew{\mrew_\mtarget}{\msa}{\mdp_\mtarget}  ~.
	\]%
\end{definition}%
\begin{example}[Reachability Probabilities]
	Reconsider the MDP $\mdp$ from \Cref{fig:running_example_mdp} on page~\pageref{fig:running_example_mdp}.
	Let $\mtarget = \{s_i^R \mid i \in \Nats\}$.
	Note that $\mdp$ is already identical to $\mdp_\mtarget$.\footnote{Very pedantically speaking, we would have to add to $\mdp$ a \emph{fresh} sink state, say $\bot_\bot$, have all the $\msa_i^R$'s redirect to $\bot_\bot$, and have $\bot$ also redirect to $\bot_\bot$, but have $\bot$ itself be unreachable. We do not do this for clarity of exposition.}
	For the (memoryless) scheduler $\msched$ that always stays in the topmost row (cf.\ \Cref{ex:memoryless_scheduler_induced_mc}), we have $\mrprob{\msched}{s_0}{\mtarget} = 0$ and thus $\mminrprob{s_0}{\mtarget} = 0$.
	On the other hand, we have $\mmaxrprob{s_0}{\mtarget} = 1$ because some state from $\mtarget$ can be eventually reached with probability $1$ by, for example, always choosing action $\actr$ in the $s_i$ states.%
	\triangleqed%
\end{example}

%% file: prelims.tex

The things we are interested in -- expected rewards -- live in a world of probability measures and expected values over paths through an MDP (under a scheduler).
%
This world is not very compatible with the way one gives denotational semantics to programming languages, which is usually in the language of ordered spaces and fixed points.
In that world, finite approximations of a process progression are captured as a finite amount of iterations of some appropriate function $\domainop$.
One application of $\domainop$ will correspond to one progression step.
The limiting behavior of the process is then captured as iterating $\domainop$ ad infinitum.
Our goal is to connect these two views.\footnote{Yet another time -- this has been done often before, but always with restrictions that are not really suitable for an application to probabilistic program semantics.}
In some sense, we will later describe how to obtain $\newmminexprewb{\nata + 1}{\mrew}{\msa}$ from $\newmminexprewb{\nata}{\mrew}{\msa}$ by applying a suitable function~$\domainop$ to $\newmminexprewb{\nata}{\mrew}{\msa}$, and how $\newmminexprew{\mrew}{\msa}$ emerges from iterating $\domainop$ ad infinitum to a suitable initial value.
This fixed point iteration will moreover happen to correspond to the \emph{least} fixed point of the ambient ordered space that we are working in.

Let us recap some basics of these ordered spaces.
Imagine that we want to assign an expected reward to each (of possibly infinitely many) MDP states.
Let us gather such an assignment in a reward vector.
Consider now an MDP with two states and the reward vectors $u = (\tfrac{1}{2},\, \tfrac{1}{2})$, $v = (\tfrac{1}{2},\, \tfrac{3}{4})$, and $w = (\tfrac{1}{3},\, 1)$.
Clearly, $v$ is \enquote{larger} than $u$, in that all entries of $u$ are component-wise larger than $v$.
On the other hand, $w$ is not comparable to $u$ nor $v$: while the second component of $w$ is indeed larger, the first component is not.
In that sense, reward vectors do \emph{not} form a \emph{total} order but only a \emph{partial order}.
\begin{definition}[Partial Orders]
	Let $\domainset$ be a set and ${\domainleq} \subseteq \domainset\times\domainset$. The structure%
	\[\highlight{\domaintuple}\]%
	is called a \emph{\highlight{\keyword{partial order}}}, if the following conditions hold:
	\begin{enumerate}
		\item $\domainleq$ is \emph{\keyword{reflexive}}, i.e.\ $\tforall{\domaina \in \domainset}$,
		\[
		\domaina\domainleq\domaina~.
		\]
		\item $\domainleq$ is \emph{\keyword{transitive}}, i.e.\ $\tforall{\domaina,\domainb,\domainc \in \domainset}$,
		\[
		\domaina \domainleq \domainb \qand \domainb\domainleq\domainc 
		\qqimplies \domaina \domainleq \domainc~.
		\]
		\item $\domainleq$ is \emph{\keyword{antisymmetric}}, i.e.\ $\tforall{\domaina,\domainb \in \domainset}$,
		\[
		\domaina \domainleq \domainb \qand \domainb\domainleq\domaina 
		\qqimplies \domaina = \domainb~.
		\]
	\end{enumerate}
	We will usually say \enquote{$\domaina$ is \emph{smaller than} $\domainb$} for $\domaina \domainleq \domainb$ instead of the more precise formulation \enquote{$\domaina$ is \emph{smaller than or equal to} $\domainb$}; and similarly for \enquote{larger than}.
\end{definition}%
We now want to speak about some notion of limits and convergence in these ordered spaces called partial orders.
Fix a partial order $\domaintuple$ and take any subset $\domainsubset \subseteq \domainset$. 
An element~$\domaina \in \domainset$ is called an \emph{\highlight{upper bound of $\domainsubset$}}, if $\domaina$ is larger than all elements of $\domainsubset$, i.e.\ if
\[
\tforall{\domainb \in \domainsubset},\quad \domainb \domainleq \domaina~.
\]
Moreover, $\domaina$ is called the \emph{least upper bound} or \emph{supremum} of $\domainsubset$, if%
\begin{enumerate}
	\item $\domaina$ is an upper bound of $\domainsubset$, and
	\item every upper bound $\domainb$ of $\domainsubset$ is larger than $d$, i.e.\ $\domaina \domainleq \domainb$.
\end{enumerate}
Antisymmetry of $\domainleq$ implies that if the supremum of $\domainsubset$ exists, then it is unique and we denote it by \highlight{$\supremum \domainsubset$}.

Dually to upper bounds, $\domaina \in \domainset$ is called a \emph{\highlight{lower bound of $\domainsubset$}}, if $\domaina$ is smaller than all elements of $\domainsubset$, i.e.\ if
\[
\tforall{\domainb \in \domainsubset},\quad \domaina \domainleq \domainb~.
\]
%
Dually to suprema, one can define a (unique) infimum $\infimum \domainsubset$ of $\domainsubset$.

We usually restrict to partially ordered spaces in which all suprema and infima are guaranteed to exist. 
Such spaces are called \emph{complete lattices}:
\begin{definition}[Complete Lattices]
	A partial order \highlight{$\domaintuple$} is called a \emph{\highlight{\keyword{complete lattice}}}, if 
	%
	every $\domainsubset \subseteq \domainset$ has a supremum $\supremum \domainsubset \in \domainset$ and an infimum $\infimum \domainsubset \in \domainset$.
\end{definition}
In particular, every complete lattice has a \emph{\highlight{\keyword{least element} $\domainleast$}} given by $\domainleast = \supremum \emptyset$, and a \emph{\highlight{\keyword{greatest element} $\domaingreatest$}} given by $\domaingreatest = \infimum \emptyset$. 

So far we have described the ambient spaces in which we would like to operate -- complete lattices.
We now come to the functions $\domainop$ which are supposed to capture the one-step progression of a process.
Fix a complete lattice $\domaintuple$.
We now consider endomaps on $\domaintuple$, i.e.\ functions%
\[
	\highlight{\domainop} \colon \domainset \to \domainset
\]%
mapping from the carrier set $\domainset$ into itself.
We already hinted above to the fact that we will be interested in least fixed points of these functions; \emph{least} in terms of the ambient partial order $\domainleq$.
Formally, we have:%
\begin{definition}[Least Fixed Points]%
	\label{def:prelim:domain_theory:fp}
	Let $\domaintuple$ be a complete lattice and let $\domainop\colon \domainset \to \domainset$ be an endomap on $\domainset$.
	\begin{enumerate}
		\item
		\label{def:prelim:domain_theory:fp1} 
			An element $\domaina \in \domainset$ is called a \emph{\highlight{\keyword{fixed point} of $\domainop$}}, if%
			\begin{align*}
				\domainop(\domaina) \eeq \domaina~.
			\end{align*}
		
		\item
		\label{def:prelim:domain_theory:fp2} 
			A fixed point $\domaina\in\domainset$ of $\domainop$ is called the \emph{\highlight{\keyword{least} fixed point of $\domainop$}}, if 
		\[
		\tforall{\text{fixed points $\domaina'$ of \,$\domainop$}}\colon\quad \domaina \ddomainleq\domaina'~.
		\]
		If it exists, we denote  the least fixed point of $\domainop$ by \highlight{$\lfp \domainop$}.
		%
		%
	\end{enumerate}
\end{definition}
\noindent
Antisymmetry of $\domainleq$ implies that if the least fixed point of $\domainop$ exists, then it is unique, which justifies the above notation and to speak of \emph{the} least fixed point.

Next, we impose a very basic property on such endomaps: monotonicity.
\begin{definition}[Monotone Functions]
	An endomap \[\highlight{\domainop} \colon \domainset \to \domainset\] on a complete lattice $\domaintuple$ is called \emph{\highlight{\keywordpl{monotone}{monotone function} (w.r.t.\ $\domaintuple$)}}, if 
	\[
	\tforall{\domaina,\domainb\in\domainset}\quad \domaina\domainleq\domainb \quad\text{implies}\quad \domainop(\domaina)\domainleq\domainop(\domainb)~.
	\]
\end{definition}%
Intuitively, monotonicity of $\domainop$ states: The more we pour into $\domainop$, the more we get out of~$\domainop$.
In our context, it will mean: 
The more reward we have already collected after $n$ steps, the more reward we will collect after $n+1$ steps.
Somewhat astonishingly, monotonicity of an endomap $\domainop$ on a complete lattice is already a sufficient criterion for~$\domainop$~to~have a least fixed point -- a theorem famously attributed to Knaster \cite{knaster1928theoreme} and Tarski \cite{tarski1955lattice}.
However, the Knaster-Tarski theorem is non-constructive and it is not apparent how to arrive at or even approach the least fixed point.
A stricter notion than monotonicity -- \emph{continuity} -- will yield such approachability.%
\begin{definition}[$\boldsymbol{\omega}$-Continuity~\textnormal{\cite{winskel}}]
	Let $\domaintuple$ be a complete lattice and let \mbox{$\domainop \colon \domainset \to \domainset$} be an endomap on $\domainset$.
	%
	We say that $\highlight{\domainop}$ is \emph{$\omega$-continuous} (or simply \emph{continuous}), if for all increasing $\omega$-chains\footnote{I.e.\ ordinary infinite sequences of elements which are consecutively ordered by $\domainleq$.} $\domainsubset = \{\domaina_0 \domainleq \domaina_1 \domainleq \ldots\}\subseteq \domainset$, we have
	\[
		\domainop\left( \supremum \domainsubset \right) 
		\qeq 
		\supremum \domainop\left( \domainsubset \right) 
		\lightgray{\qeq 
		\supremum \big\{ \domainop(\domaina) ~{}\mid{}~ \domaina \in \domainsubset \big\}}~.
	\]
	%
	%
%
\end{definition}
In other words: $\domainop$ is continuous if taking suprema of and applying $\domainop$ to an $\omega$-chain $\domainsubset$ commutes, i.e.\ it does not matter whether we first determine the supremum of the chain and then apply $\domainop$ to it, or whether we first take the image of $\domainsubset$ under $\domainop$ and then determine the supremum of the so-obtained set.

For continuous functions, we will now see an also very well-known fixed point iteration method that approaches the least fixed point in the limit.
Given a natural number $\nata$, let $\domainop^\nata(\domaina)$ denote $\nata$-fold iteration of $\domainop$ to $\domaina$, e.g.\ $\domainop^3(\domaina) = \domainop(\domainop(\domainop(\domaina)))$ (we also speak of the \emph{$n$-th iteration of $\domainop$  on $\domaina$}).
Then we can characterize the least fixed point of a continuous function as follows:%
\begin{theorem}[Kleene Fixed Point Theorem \textnormal{\cite{winskel}}]
	\label{thm:prelim:domain_theory:kleene}
	Let $\domaintuple$ be a complete lattice and let $\domainop \colon \domainset \to \domainset$ be a \emph{continuous} endomap on $\domainset$.
	Then
	\[
		\lfp \domainop \qeq \supremum \setcomp{\domainop^\nata(\domainleast)}{\nata \in \Nats}~.
	\]
\end{theorem}%
The Kleene fixed point theorem is a very powerful general principle.
For instance, the much more prominent Banach fixed point theorem is a consequence of it~\cite{baranga1991contraction}.

Moreover, because continuity implies monotonicity and hence $\domainop$ is monotonic, iterating $\domainop$ on $\domainleast$ actually forms an increasing chain%
\begin{align*}
	\domainleast 
	\ddomainleq \domainop(\domainleast) 
	\ddomainleq \domainop^2(\domainleast) 
	\ddomainleq {\ldots}
	\qquad \ddomainleq \lfp \domainop
\end{align*}
whose supremum is the least fixed point of $\domainop$.
This means that not only does the fixed point iteration actually converge to $\lfp \domainop$, but moreover any iterate $\domainop^n(\domainleast)$ is smaller than (or equal to) $\lfp \domainop$ and hence the fixed point iteration approaches $\lfp \domainop$ faithfully \emph{from below}.
Ultimately, this is (one argument) why value iteration in probabilistic model checking gives sound lower bounds.\footnote{Although without additional ado, one has no idea how tight these lower bounds are~\cite{DBLP:journals/tcs/HaddadM18,DBLP:conf/cav/QuatmannK18,DBLP:conf/cav/HartmannsK20}.}

While the Kleene fixed point theorem gives us a handle at iteratively approaching the sought-after least fixed point, that least fixed point remains in some sense an intrinsically infinitary object, namely in that it is the supremum of an $\omega$-chain that will generally (and even for simple examples) indeed only stabilize at iteration $\omega$ and \emph{not} already after some finite number of iterations.
Yet, even for non-continuous but monotonic maps, the Knaster-Tarski theorem yields as an immediate consequence a powerful induction principle:%
\begin{lemma}[\keywordpl{Park Induction}{Park induction}~\textnormal{\cite{park}}]
	\label{lem:prelim:domain_theory:park}
	Let $\domaintuple$ be a complete lattice and $\domainop\colon \domainset \to \domainset$ be a \emph{monotonic} endomap on $\domainset$. 
	Then $\tforall{\domaina\in\domainop}$,
	\[
	\domainop(\domaina)\ddomainleq \domaina
	\qqimplies
	\lfp \domainop \ddomainleq \domaina~.
	\]
\end{lemma}%
Intuitively, Park induction provides a way to check if some candidate $\domaina$ is an upper bound on the least fixed point of $\domainop$.
For that, we push $\domaina$ through $\domainop$ just \emph{once} and check whether this will take us \emph{down} in the partial order $\domainleq$.
If yes, then we have immediately verified that $\domaina$ is an upper bound on $\lfp \domainop$ -- the supremum of a potentially infinite chain.

Park induction is yet another powerful principle that is ubiquitous in computer science, but especially in the realm of verification and formal methods. 
For example, the invariant rule of Hoare logic \cite{DBLP:journals/cacm/Hoare69} can be entirely deduced from Park induction.
For a detailed explanation, see \cite{kaminski_diss}.

%% file: rewards_lfp.tex

In \Cref{sec:expected_rewards_def}, we defined expected rewards in terms of their paths. 
While we believe this notion is the most intuitive one -- it is straightforward to comprehend the expected reward of small (especially acyclic) MDPs by considering their paths -- it is also desirable to characterize expected rewards as least fixed points (see \Cref{sec:prelim:domain_theory}), particularly when reasoning symbolically about infinite-state MDPs (see \Cref{sec:pp_tick}).

Least fixed point characterizations of expected rewards are typically derived from \emph{Bellman's optimality equations}~\cite{Bellman1957MDPs}.
More specifically, for an MDP $\mdp = \mdptuple$ and a reward function $\mrew\colon\ms \to\PosRealsInf$, the set of Bellmann equations is given by
\begin{align}
	\label{eqn:mdp:bellman}
	X \qeq \mylambda{\msa} \mrew(\msa) + \mopt_{\macta \in \mact(\msa)} \sum_{\msa'\in\msuccs{\macta}(\msa)} \mprob(\msa, \macta, \msa') \cdot X(\msa') ~,
\end{align}%
where $X \colon \ms \to \PosRealsInf$ is a reward vector and the choice of $\mopt \in \{\min,\max\}$ depends on whether we are considering minimal or maximal expected rewards. 

For various classes of MDPs with finitely-valued reward functions, i.e. $\mrew(s) < \infty$ for all states $\msa \in \ms$, it is well-established, that 
$\newmminexprews{\mrew}{\msa}$ is the least solution of \Cref{eqn:mdp:bellman} for $\mopt = \min$ (cf.~\cite[Theorem 7.2.3a]{puterman}) and $\newmmaxexprews{\mrew}{\msa}$ is the least solution of \Cref{eqn:mdp:bellman} for $\mopt = \max$ (cf.~\cite[Theorem 7.3.3a]{puterman}).

In this section, we will show that the same result holds for \emph{unbounded} reward functions $\mrew\colon\ms \to\PosRealsInf$ and expected rewards of MDPs as introduced in \Cref{sec:prelim:mdp,sec:expected_rewards_def}.
That is, the least solution of \Cref{eqn:mdp:bellman} exists (\Cref{thm:mdp:bellman_cont}) and is equal to the sought-after expected rewards (\Cref{thm:mdp:bellman_correct}).
The solution domain of \Cref{eqn:mdp:bellman}, i.e. reward vectors, is the complete lattice of \emph{value functions}.%
\begin{definition}[Value Functions]\label{def:value-functions}
	Let $\mdp = \mdptuple$ be an MDP. The complete lattice of \emph{\highlight{value functions (for $\mdp$)}} is defined as 
	\begin{align*}
	\highlight{(\mvals,\, \mleq)}~, \qquad\textnormal{where:}
	\end{align*}%
	%
	%
	\begin{enumerate}
		\item $\highlight{\mvals} = \ms \to \PosRealsInf$ is the \emph{\highlight{set of value functions}}, and 
		\item  $\highlight{\mleq}$ is the pointwise lifted ordinary order on $\PosRealsInf$, i.e.\ $\tforall{\mval,\mval' \in \mvals},$
		\[
		\mval \mmleq \mval' \qqiff \tforall{\msa \in \ms} \colon \mval(\msa) \leq \mval'(\msa)~.
		\]
	\end{enumerate}
\end{definition}%
The least element of $(\mvals, \, \mleq)$ is the constant-zero-function $\mylambda{\msa} 0$, which we \mbox{--- by} slight abuse of notation --- denote by $0$. Moreover, suprema and infima in \mbox{$(\mvals, \, {\mleq})$} are the pointwise liftings of suprema and infima in $(\PosRealsInf, \, \leq)$, i.e.\ $\tforall{\mvalsset \subseteq \mvals}$,
\[
\supremum \mvalsset \eeq \mylambda{\msa} \sup\setcomp{\mval(\msa)}{\mval \in \mvalsset}
\qqand
\infimum \mvalsset \eeq \mylambda{\msa} \inf\setcomp{\mval(\msa)}{\mval \in \mvalsset}~.
\]%
Equipped with a complete lattice, we next translate the Bellman equations in \Cref{eqn:mdp:bellman} into an endomap $\mvals \to \mvals$:
%
%
%
\begin{definition}[Bellman Operators]
	\label{def:mdp:bellman}
	Let $\mdp = \mdptuple$ be an MDP and \mbox{$\mrew \colon \ms \to \PosRealsInf$} be a reward function for $\mdp$.
	\begin{enumerate}
		\item\label{def:mdp:bellman1} The function $\highlight{\mopmin{\mdp}{\mrew}} \colon \mvals \to \mvals$, defined as
		\[
		\highlight{\mopmins}(\mval) \eeq \mylambda{\msa} \mrew(\msa) + \min_{\macta \in \mact(\msa)} \sum_{\msa'\in\msuccs{\macta}(\msa)} \mprob(\msa, \macta, \msa') \cdot \mval(\msa') ~,
		\]
		is called the \emph{\highlight{min-Bellman operator of $\mdp$ with respect to $\mrew$}}.
		\item\label{def:mdp:bellman2} Analogously, the function $\highlight{\mopmax{\mdp}{\mrew}} \colon \mvals \to \mvals$, defined as
		\[
		\highlight{\mopmaxs}(\mval) \eeq \mylambda{\msa} \mrew(\msa) + \max_{\macta \in \mact(\msa)} \sum_{\msa'\in\msuccs{\macta}(\msa)} \mprob(\msa, \macta, \msa') \cdot \mval(\msa') ~,
		\]
		is called the \emph{\highlight{max-Bellman operator of $\mdp$ with respect to $\mrew$}}.
	\end{enumerate}
\end{definition}%
%
%
%
%
\begin{example}[Bellman Operators]
\label{ex:bellman-operators}
	The min-Bellman operator of the MDP $\mdp$ from \Cref{fig:running_example_mdp} with respect to the reward function indicated by the boxed numbers next to the states is given by%
	\abovedisplayskip=0pt%
	\begin{align*}
		\mopmins(\mval) \eeq \mylambda{\msa}
		\begin{cases}
			\min \left\{ \mval(s_i^L), \frac 1 2 \mval(s_{i+1}) + \frac 1 2 \mval(s_i^R)\right\} & \text{ if } \msa = s_i, i \in \Nats \\
			r + \mval(s_{i+1}) & \text{ if } \msa = s_i^L, i \in \Nats \\
			i+1  + \mval(\opsink) & \text{ if } \msa = s_i^R, i \in \Nats \\
			\mval(\opsink) & \text{ if } \msa = \opsink.
		\end{cases}
	\end{align*}%
	The max-Bellman operator is identical except that $\min$ is replaced by $\max$.%
	\triangleqed%
\end{example}
Next, we will show the existence of the least fixed point of the Bellman operators:%
\begin{theorem}[Continuity of Bellman Operators]
	\label{thm:mdp:bellman_cont}
	Let $\mdp = \mdptuple$ be an MDP and let $\mrew \colon \ms \to \PosRealsInf$.  
	\begin{enumerate}
		\item\label{thm:mdp:bellman_cont1} Both the min-Bellman operator $\mopmins$ and the max-Bellman operator $\mopmaxs$ of $\mdp$ w.r.t.\ $\mrew$ are continuous w.r.t.\ $(\mvals, \, \mleq)$.
		\item\label{thm:mdp:bellman_cont2} Both $\mopmins$ and $\mopmaxs$ have a least fixed point given by
		\begin{align*}
			\lfp \mopmins &\qeq \supremum \setcomp{\mopminsiter{n}(0)}{\nata \in \Nats} \\[.5em]
			\text{and}\qquad
			\lfp \mopmaxs &\qeq \supremum \setcomp{\mopmaxsiter{n}(0)}{\nata \in \Nats}~.
		\end{align*}
	\end{enumerate}
\end{theorem}
\begin{proof}
	\Cref{thm:mdp:bellman_cont}.\ref{thm:mdp:bellman_cont1} is a consequence of the fact that addition, multiplication, minima, and maxima are continuous w.r.t.\ $(\PosRealsInf, \, \leq)$ \cite[Chapter 1]{weighted_automata}. \Cref{thm:mdp:bellman_cont}.\ref{thm:mdp:bellman_cont2} is then an instance of Kleene's fixed point theorem, see \Cref{thm:prelim:domain_theory:kleene}.\qed
\end{proof}%
\noindent%
Now that we have established that the least fixed points exist, we will show that they evaluate to the sought-after expected rewards.
To this end, we first show that performing $n+1$ fixed point iterations yields the $n$-step bound expected reward (see~\Cref{def:mdp:exrew}).

%
%
%
\begin{restatable}[Step-Bounded Exp.\ Rew.\ via Fixed Point Iteration]{lemma}{bellmanStep}
	\label{lem:mdp:bellman-step}
	{
		Let $\mdp = \mdptuple$ be an MDP and let $\mrew \colon \ms \to \PosRealsInf$.  Then for all $\nata \in \Nats$:
		\begin{enumerate}
			\item\label{lem:mdp:bellman-step1}
			$\mopminsiter{\nata+1}(0) \eeq \mylambda{\msa} 
			\newmminexprewb{\nata}{\mrew}{\msa}$
			\bigskip
			\item\label{lem:mdp:bellman-step2}
			$\mopmaxsiter{\nata+1}(0) \eeq \mylambda{\msa} \newmmaxexprewb{\nata}{\mrew}{\msa}$
		\end{enumerate}
		%
	}
\end{restatable}%
\begin{proof}
	By induction on $\nata$. See \Cref{proof:mdp:bellman-step} for details.
	\qed
\end{proof}%
%
%
%
%
%
Our main result states that, in the limit, the least fixed points of the Bellmann operators is equal to the minimal and maximal expected reward, respectively.
%
\begin{restatable}[Expected Rewards via Least Fixed Points]{theorem}{bellmanCorrect}
	\label{thm:mdp:bellman_correct}
	{Let $\mdp = \mdptuple$ be an MDP and $\mrew \colon \ms \to \PosRealsInf$. Then:
		\begin{enumerate}
			\item\label{thm:mdp:bellman_correct1}
			$
			\lfp \mopmins \qeq \mylambda{\msa} \newmminexprew{\mrew}{\msa} 
			$
			\bigskip
			\item\label{thm:mdp:bellman_correct2}
			$
			\lfp \mopmaxs \qeq \mylambda{\msa} \newmmaxexprew{\mrew}{\msa} 
			$
		\end{enumerate}
		%
	}
\end{restatable}
\begin{proof}[sketch]
	For maximal expected rewards, the claim follows from \Cref{thm:prelim:domain_theory:kleene} and \Cref{lem:mdp:bellman-step}.\ref{lem:mdp:bellman-step2}. The proof for minimal expected rewards is more involved and relies on the fact $\lfp \mopmins$ gives rise to an optimal memoryless scheduler (cf.\ \Cref{lem:mdp:min_opt_sched} in the appendix). The full proof is given in \Cref{proof:mdp:bellman_correct}.
	\qed
\end{proof}%
\begin{example}[Expected Rewards as Least Fixed Points in Action]
	\label{ex:bellman-fixpoint}
	In general, the Bellmann operators of MDPs may have many fixed points.
	In particular, for every MDP $\mdp$ with reward function $\mrew$, the constant value function $\infty \mathrel{\gray{({=}}} \gray{\mylambda{\msa}\infty)}$ is a trivial fixed point of both $\mopmins$ and $\mopmaxs$.
	For many interesting examples, however, the \emph{least} fixed point is strictly smaller than $\infty$.
	For instance, if $\mdp$ is the MDP from \Cref{fig:running_example_mdp}, then the value function $\mval$ with 
	\[
		\mval(s_i) = \infty,
		\quad
		\mval(s_i^L) = \infty,
		\quad
		\mval(s_i^R) = i+1,
		\quad\text{and}\quad
		\mval(\opsink) = 0~,
	\]
	for all $i \in \Nats$, is a non-trivial fixed point of both $\mopmaxs$ and $\mopmins$.
    This holds for all possible concrete values of the reward $r$.
	The value function $\mval$ is a fixed point of $\mopmaxs$ because, for all $i \in \Nats$, we have (cf.\ \Cref{ex:bellman-operators})
	\begin{align*}
		&{\max} \left\{ \mval(s_i^L), \frac 1 2 \mval(s_{i+1}) + \frac 1 2 \mval(s_i^R)\right\} 
		= 
		\max \left\{ \infty, \frac 1 2 \cdot \infty + \frac 1 2  (i + 1)\right\}
		=
		\infty 
		=
		\mval(s_i) ~,
		\\
		&r + \mval(s_{i+1})
		\eeq
		r + \infty 
		\eeq
		\infty
		\eeq
		\mval(s_i^L) ~,
		\\
		&i+1 + \mval(\opsink)
		\eeq
		i+1 + 0
		\eeq
		\mval(s_i^R) ~, \qand \\
		&\mval(\opsink)
		\eeq
		0
		\eeq
		\mval(\opsink)
		~.
	\end{align*}
    The fact that $\mval$ is also a fixed point of $\mopmins$ follows similarly by noticing that%
    \begin{align*}
        &{\min} \left\{ \mval(s_i^L), \frac 1 2 \mval(s_{i+1}) + \frac 1 2 \mval(s_i^R)\right\} 
        = 
        \min \left\{ \infty, \frac 1 2 \cdot \infty + \frac 1 2  (i + 1)\right\}
        =
        \infty 
        =
        \mval(s_i) ~.
    \end{align*}%
	We claim that $\mval$ as defined above is indeed the \emph{least} fixed point of $\mopmaxs$ (it is, however, \emph{not} the least fixed point of $\mopmins$, see below).
	This already follows with \Cref{thm:mdp:bellman_correct} and \Cref{ex:maxERminER}, but it is instructive to give a direct proof:
	
	Towards a contradiction, assume that $\mval' \mleq \mval \neq \mval'$ is another fixed point of $\mopmaxs$.
	It is immediate that $\mval'(s_i^R) = \mval(s_i^R)$ for all $i \in \Nats$, and hence
	$\mval'(s_i) < \infty$ or $\mval'(s_i^L) < \infty$ for some $i \in \Nats$.
	Since $\mval'$ is a fixed point, this already implies that $\mval'(s_i) < \infty$ and $\mval'(s_i^L) < \infty$ for \emph{all} $i \in \Nats$.
	Now, consider the following: 
	\begin{align*}
	   \mval'(s_0) \ggeq \frac 1 2 \mval'(s_0^R) \eeq \frac 1 2& \\
		\mval'(s_0) \ggeq \mval'(s_0^L) \geq \mval'(s_1) \ggeq \frac 1 2 \mval'(s_1^R) \eeq 1&  \\
		\mval'(s_0) \ggeq \mval'(s_0^L) \geq \mval'(s_1) \ggeq \mval'(s_1^L) \ggeq \mval'(s_2) \ggeq \frac 1 2 \mval'(s_2^R) \eeq \frac 3 2& ~, \text{ etc.}
	\end{align*}
	The above argument shows that $\mval'(s_0)$ is greater than \emph{any} real number, and so it must hold that $\mval'(s_0) = \infty$, contradiction.
	
	We now consider the least fixed point of $\mopmins$.
	If $r = 0$, then it is easy to see that the least fixed point is given by $\mval(s_i) = \mval(s_i^L) = \mval(\opsink) = 0$ and $\mval(s_i^R) = i+1$, for all $i \in \Nats$.
	On the other hand, if $r > 0$, then one can verify that the least fixed point is $\mval$ with $\mval(s_i) = i+2$, $\mval(s_i^L)= r + i + 3$, $\mval(s_i^R) = i+1$ and $ = \mval(\opsink) = 0$ for all $i \in \Nats$.%
	\triangleqed%
\end{example}

%% file: exrew_schedulers.tex

We saw in \Cref{sec:expected_rewards_def} that minimal and maximal expected rewards are defined via the infimum and supremum, respectively, over all schedulers for an MDP.
However, even though the maximal expected reward is a well-defined number, that does not necessarily mean that there also \emph{exists} a scheduler that actually attains this supremum.
It may well be that there is only a sequence of schedulers that converges to the maximal expected reward.
By contrast, for minimal expected rewards, there always exists a single scheduler.

For restricted classes of MDPs, these results are known (cf.~\cite[Sections 7.23 and 7.32]{puterman}). In this section, we show that they still hold for our generalized expected rewards.


Let $\mdp = \mdptuple$ be an MDP and let $\mrew\colon \ms \to \PosRealsInf$ be a reward function.
Recall from \Cref{sec:expected_rewards_def} that a scheduler $\msched$ is \emph{\highlight{min-optimal for $\mdp$ with respect to $\mrew$ at state $\msa$}}, if $\msched$ attains the minimal expected reward, i.e.\ if%
\[
\newmexprew{\msched}{\mrew}{\msa} \eeq \newmminexprew{\mrew}{\msa}~.
\]%
We say that $\msched$ is \emph{\highlight{uniformly} min-optimal for $\mdp$ w.r.t.\ $\mrew$}, if $\msched$ is min-optimal at all states $\msa\in\ms$, i.e. if 
\[
\tforall{\msa\in\ms}\colon\quad  \newmexprew{\msched}{\mrew}{\msa} \eeq \newmminexprew{\mrew}{\msa}~.
\]
For maximal expected rewards, both notions are defined analogously.

\begin{figure}[t]
    \centering
    \input{no-opt-sched}
    \caption{An MDP where no max-optimal scheduler exist at any state $s_i$, $i \in \Nats$.}
    \label{fig:no-opt-sched}
\end{figure}
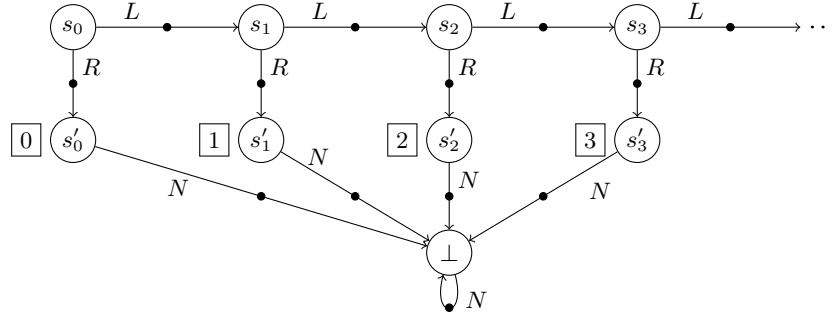

We obtain as a by-catch of \Cref{lem:mdp:min_opt_sched} and \Cref{thm:mdp:bellman_correct}.\ref{thm:mdp:bellman_correct1} that there is always a uniformly \emph{min}-optimal scheduler $\msched$, which is moreover memoryless.%
%
%
\begin{theorem}[Optimal Schedulers for Min.~Expected Rewards]
	\label{thm:mdp:opt_min_strats}
	Let $\mdp = \mdptuple$ be an MDP and let $\mrew$ be a reward function for $\mdp$. 
	There is a memoryless uniformly min-optimal scheduler $\msched \in \mmscheds$ for $\mdp$.
\end{theorem}
Since such a memoryless scheduler actually exists, the infimum from \Cref{def:mdp:exrew}.\ref{def:mdp:exrew4} can be expressed as a \emph{minimum}, namely%
\begin{align}
	\label{eqn:mdp:minexprew_min}
	\newmminexprews{\mrew}{\msa} \qeq  \min_{\msched \in \mmscheds} \newmexprews{\msched}{\mrew}{\msa} ~.
\end{align}
%
Regarding \emph{maximal} expected rewards, it is known that \emph{neither} uniform \emph{nor} non-uniform max-optimal schedulers necessarily exist, see, e.g.\ \cite[Example~7.2.4]{puterman}.
Indeed, for the MDP depicted in \Cref{fig:no-opt-sched}, there does not exist a max-optimal scheduler at any state $\msa_i$, $i \in \Nats$.
This is because for every $n \in \Nats$ there is a scheduler $\msched_n$ achieving reward $\newmexprews{\msched_n}{\mrew}{\msa} = n + i$ from $s_i$ by playing the $\actl$-action $n-1$ times and then the $\actr$-action to collect reward $n+i$.
Hence, $\newmmaxexprews{\mrew}{\msa_i} = \infty$, but there does not exist a scheduler that actually achieves infinite expected reward (note that the scheduler that always chooses $\actl$ obtains zero reward).
In particular, the supremum from \Cref{def:mdp:exrew}.\ref{def:mdp:exrew6} is \emph{not necessarily attained}.
Somewhat analogous to \Cref{eqn:mdp:minexprew_min}, one could, however, ask whether the supremum from \Cref{def:mdp:exrew}.\ref{def:mdp:exrew6} over \emph{all} schedulers can at least be replaced by a supremum over \emph{memoryless} schedulers, i.e. whether for all states $\msa$,
\[
\newmmaxexprews{\mrew}{\msa} ~{}\stackrel{?}{=}{}~  \sup_{\msched \in \mmscheds} \newmexprews{\msched}{\mrew}{\msa} ~.
\]
This is the case if and only if, for all states $\msa$ and all $\epsilon>0$, there is an $\epsilon$-max-optimal memoryless scheduler $\msched$, i.e. a memoryless scheduler $\msched$ with 
\[
\newmmaxexprews{\mrew}{\msa} - \epsilon ~{}\leq{}~ \newmexprews{\msched}{\mrew}{\msa} ~.
\]
While the above is known for positive models if the expected reward is strictly less than $\infty$ for every state~\cite[Corollary 7.2.8]{puterman}, the answer to this question is, to the best of our knowledge, \emph{open} when admitting infinite (expected) rewards.
An in-depth treatment of scheduler existence is outside the scope of this paper.%
\begin{example}[Optimal Schedulers for MDPs]
\label{ex:opt-sched}
	Reconsider the MDP from \Cref{fig:running_example_mdp} on page~\pageref{fig:running_example_mdp} together with the reward function indicated by the boxed numbers next to the states.
	\Cref{thm:mdp:opt_min_strats} asserts existence of a uniformly min-optimal memoryless scheduler.
	Indeed, a memoryless scheduler $\msched$ is uniformly min-optimal iff,
	\begin{itemize}
		\item assuming $r=0$, $\msched(s_i) = \actl$ for all $ i \in \Nats$;
		\item assuming $r>0$, $\msched(s_i) = \actr$ for all $ i \in \Nats$.
	\end{itemize}
	We now consider max-optimal schedulers.
	First, if $r>0$, then there exists a uniformly max-optimal memoryless scheduler $\msched$ with $\msched(s_i) = \actl$ for all $ i \in \Nats$.
	This scheduler attains the least fixpoint of $\mopmaxs$ described in \Cref{ex:bellman-fixpoint}.
	For the case $r=0$ we claim that the scheduler $\msched$ which is uniquely determined by requiring that $\msched(s_i) = \actr$ iff $i = 2^k$ for some $k \in \Nats$ is uniformly max-optimal.
	Indeed, we have%
	\begin{align*}
		\newmexprews{\msched}{\mrew}{s_i}
		\eeq
		&\frac 1 2 (i+1) + \frac 1 4 (i+2) + \frac 1 8 (i+4 ) + \ldots \\
		\eeq
		&i + \frac 1 2 \left(\frac 1 2 \cdot 2 + \frac 1 4 \cdot 4 + \frac 1 8 \cdot 8 + \ldots\right)
		\eeq\infty
		~.
	\end{align*}%
	We recall that, in general, max-optimal schedulers are not guaranteed to exist in infinite-state MDPs, see \Cref{fig:no-opt-sched}.%
	\triangleqed%
\end{example}

%% file: no-opt-sched.tex
\begin{tikzpicture}[on grid, node distance=15mm and 25mm,actstyle/.style={draw,circle,fill=black,inner sep=1pt},every state/.style={inner sep=0pt, minimum size=6mm}]
    \node[state] (0) {$s_0$};
    \node[state,right=of 0] (2) {$s_1$};
    \node[state,right=of 2] (4) {$s_2$};
    \node[state,right=of 4] (6) {$s_3$};
    \node[right=of 6] (dots) {$\cdots$};
    
    \node[state,below=of 0, label=180:\fbox{$0\vphantom{2}$}] (1) {$s_0'$};
    \node[state,right=of 1,label=180:\fbox{$1\vphantom{2}$}] (3) {$s_1'$};
    \node[state,right=of 3,label=180:\fbox{$2\vphantom{2}$}] (5) {$s_2'$};
    \node[state,right=of 5,label=180:\fbox{$3\vphantom{2}$}] (7) {$s_3'$};
    \node[right=of 6] (dots') {$\cdots$};

    \node[state,below=of 5] (sink) {$\opsink$};
    
    \draw[->] (0) edgenode[actstyle,pos=0.5] {} node[pos=0.25,above]{$\actl$} node[pos=0.75,above]{} (2);
    \draw[->] (2) edgenode[actstyle,pos=0.5] {} node[pos=0.25,above]{$\actl$} node[pos=0.75,above]{} (4);
    \draw[->] (4) edgenode[actstyle,pos=0.5] {} node[pos=0.25,above]{$\actl$} node[pos=0.75,above]{} (6);
    \draw[->] (6) edgenode[actstyle,pos=0.5] {} node[pos=0.25,above]{$\actl$} node[pos=0.75,above]{} (dots);
    
    \draw[->] (0) edgenode[actstyle,pos=0.5] {} node[pos=0.25,right]{$\actr$} node[pos=0.75,above]{} (1);
    \draw[->] (2) edgenode[actstyle,pos=0.5] {} node[pos=0.25,right]{$\actr$} node[pos=0.75,above]{} (3);
    \draw[->] (4) edgenode[actstyle,pos=0.5] {} node[pos=0.25,right]{$\actr$} node[pos=0.75,above]{} (5);
    \draw[->] (6) edgenode[actstyle,pos=0.5] {} node[pos=0.25,right]{$\actr$} node[pos=0.75,above]{} (7);
    
    \draw[->] (1) edgenode[actstyle,pos=0.5] {} node[pos=0.25,below]{$\actn$} node[pos=0.75,above]{} (sink);
    \draw[->] (3) edgenode[actstyle,pos=0.5] {} node[pos=0.25,above]{$\actn$} node[pos=0.75,above]{} (sink);
    \draw[->] (5) edgenode[actstyle,pos=0.5] {} node[pos=0.25,right]{$\actn$} node[pos=0.75,above]{} (sink);
    \draw[->] (7) edgenode[actstyle,pos=0.5] {} node[pos=0.25,below right]{$\actn$} node[pos=0.75,above]{} (sink);

    \path[->]
    (sink) edge[loop below,looseness=10] node[pos=0.5,actstyle,yshift=2pt]  {} node[pos=0.25,right] {$\actn$} node[pos=0.75, right] {} (sink)
    ;
    
\end{tikzpicture}

%% file: probabilistic_programs_with_tick.tex

We will now turn towards \enquote{an application to operational semantics of probabilistic semantics}.
For that, we will introduce a simple probabilistic programming language \`a la McIver \& Morgan \cite{DBLP:series/mcs/McIverM05}.
Besides standard constructs such as assignments, while loops as well as conditional, probabilistic, and nondeterministic branching, our language has a designated statement for modelling the collection of rewards during execution.

In this section, we discuss syntax and operational semantics of our language, i.e.\ a single infinite-state MDP describing all program executions.
In \Cref{sec:prelim:wp:calculus}, we will present two weakest preexpectation calculi for the same language and -- using \Cref{thm:mdp:bellman_correct} -- prove them sound with respect to the operational semantics.



\subsection{A Probabilistic Guarded Command Language with Rewards}
\label{sec:pgcl-syntax}
Let $\highlight{\Vars} = \{\vara,\varb,\varc,\ldots\}$ be a countably infinite \emph{\highlight{set of program variables}}.
A \emph{\highlight{(program) state}} is a function of type $\Vars \to \PosRats$. Program states are denoted by $\statea,\stateb$, and variations thereof. 
We restrict to \emph{non-negative} rationals for the sake of convenience (this is further discussed in \Cref{sec:prelim:wp:expectations}). 
%
Since program states will be part of the state space of an MDP, which, by \Cref{def:prelim:mdps}, must be countable, we must ensure that there are only countably many program states. Hence, we define the  \emph{\highlight{set of program states}} as
\[
\highlight{\States} \eeq \setcomp{\statea \colon \Vars \to \PosRats}{ \text{the set $\setcomp{\vara \in \Vars}{\statea(\vara)>0}$ is finite}}~.
\]%
In other words, we assume that all but finitely many variables evaluate to zero for every program state.
This is a fair assumption because only finitely many variables will appear in any given program.

A \emph{\highlight{predicate (over \States)}} is a set $\highlight{\preda} \subseteq \States$ and the \emph{\highlight{set of predicates}} is \highlight{$\Preds$}. We usually write $\statea \models \preda$ instead of $\statea \in \preda$ and say that \emph{$\statea$ satisfies $\preda$}. Moreover, we employ the usual logical connectives between predicates. For instance, we write%
\[
\preda \wedge \predb \qquad\text{instead of}\qquad  \preda \cap \predb~.
\]%
With these notions at hand, we introduce our probabilistic programming language $\pgcl$.%
\begin{definition}[The Probabilistic Guarded Command Language]
	\label{def:prelim:wp:pgcl}
	The set $\pgcl$ of programs in the \highlight{\emph{probabilistic guarded command language}} is given \mbox{by the grammar}
	\begin{align*}
		\cc \quad \longrightarrow \quad& \phantom{{}\mid{}~} \SKIP  \tag{effectless program}\\
		&{}\mid{}~ \ASSIGN{\vara}{\Axa}  \tag{assignment}\\
		&{}\mid{}~ \TICK{\tickrew}  \tag{collect reward}\\
		&{}\mid{}~ \COMPOSE{\cc}{\cc}  \tag{sequential composition}\\
		&{}\mid{}~ \PCHOICE{\cc}{\proba}{\cc} \tag{probabilistic choice}\\
		&{}\mid{}~ \NDCHOICE{\cc}{\cc} \tag{nondeterministic choice}\\
		&{}\mid{}~ \ITE{\Bxa}{\cc}{\cc} \tag{conditional choice}\\
		&{}\mid{}~ \WHILEDO{\Bxa}{\cc}~, \tag{while loop}
	\end{align*}
	where:
	\begin{enumerate}
		\item $\highlight{\Axa} \colon \States \to \PosRats$ is an \emph{\highlight{arithmetic expression}}, 
		\item $\highlight{\tickrew} \in \PosRatsInf$ is a \emph{\highlight{rational reward}} (or $\infty$ ),
		\item $\highlight{\proba} \in [0,1] \cap \Rats$ is a \emph{\highlight{rational probability}}, and
		\item $\highlight{\Bxa} \in \Preds$ is a predicate also referred to as a \emph{\highlight{guard}}.
	\end{enumerate}
	
\end{definition}
The command $\TICK{\tickrew}$ collects reward $\tickrew$ and then immediately terminates. 
All other commands are fairly standard: $\SKIP$ does nothing, i.e.\ it behaves like $\TICK{0}$. The assignment $\ASSIGN{\vara}{\Axa}$ assigns the value of expression $\Axa$ (evaluated in the current program state) to variable $\vara$. The sequential composition $\COMPOSE{\cc_1}{\cc_2}$ executes $\cc_1$ followed by $\cc_2$. The probabilistic choice $\PCHOICE{\cc_1}{\proba}{\cc_2}$ executes $\cc_1$ with probability $\proba$ and $\cc_2$ with probability $(1-\proba)$. By constrast, $\NDCHOICE{\cc_1}{\cc_2}$ nondeterministically executes either $\cc_1$ or $\cc_2$.
$\ITE{\Bxa}{\cc_1}{\cc_2}$ executes $\cc_1$ if the guard $\Bxa$ holds in the current program state; otherwise, it executes $\cc_2$.
Similarly, the loop $\WHILEDO{\Bxa}{\cc}$ checks if the guard $\Bxa$ holds. If yes, it executes $\cc$ and repeats. Otherwise, it stops immediately like $\SKIP$.

To simplify the presentation, we abstract from a concrete syntax for arithmetic or Boolean expressions.
Instead, we admit any function $\States \to \Vals$ and any predicate in $\Preds$ as expressions.
However, to ensure that the set $\pgcl$ of programs is countable, we assume that (i) both arithmetic and Boolean expressions are taken from some countable universe of expressions and (ii) rewards are taken from the countable set of non-negative rationals or infinity. 

\begin{figure}[t]
	\abovedisplayskip=0pt%
	\belowdisplayskip=0pt%
	\begin{align*}
		& \WHILE{x = 0} \\
		& \qquad \{\, \TICK{\tickrew} \,\} \\
		& \qquad \ndchoicesymb \\
		& \qquad \{\, \PCHOICE{\SKIP}{\nicefrac 1 2}{\ASSIGN{x}{1}} \,\} \,; \\
		& \qquad \ASSIGN{y}{y + 1} \quad \}
	\end{align*}%
	\normalsize%
	\caption{An example $\pgcl$ program. We assume that $\tickrew \in \PosRatsInf$ is a constant.}%
	\label{fig:running-example-pgcl}%
\end{figure}

\subsection{Operational MDP Semantics of $\textbf{\textup{\textsf{pGCL}}}$}
\label{sec:prelim:pgcl:opsemtick}
\begin{figure}[t]
	\centering
	\input{opsem_example}
	\caption{A reachable fragment of $\pgcl$'s operational MDP $\opmdp$ for the loop from \Cref{fig:running-example-pgcl}. Program states are denoted by conjunctions of equalities, i.e.\ $\vara = \nata \wedge \varb=\natb$ indicates that the current state $\statea$ satisfies $\statea(\vara)=\nata$ and $\statea(\varb)=\natb$.
        }
	\label{fig:opsem-example}
\end{figure}
Towards defining the operational MDP semantics of $\pgcl$, we define a small-step execution relation à la Plotkin \cite{plotkin_op}. We mainly follow the presentation from \cite{qsl_popl} with adaptions from \cite{programmatic}.  We define the (countable) \emph{\highlight{set of configurations}} as
\[
\highlight{\Confs} \eeq \left( \pgcl \cup \{\term\} \right) \times \States ~\cup~ \{\opsink\}\, ~.
\]
Configurations are denoted by $\confa$ and variations thereof.
A configuration of the form $(\cc,\statea)$ indicates that program $\cc$ is to be executed on program state $\statea$.
A configuration of the form \highlight{$(\term, \stateb)$} is called \emph{\highlight{final}} and indicates termination in the final state $\stateb$.
$\opsink$ is a special configuration reached from every final configuration and is needed to properly encode the semantics in terms of expected rewards as will be explained in \Cref{sec:wp_tick}.
We then define the \emph{\highlight{small-step execution relation}}
\[
\highlight{\xrightarrow{}} ~{}\subseteq{}~ \Confs \times \{\actn,\actl,\actr \} \times \left([0,1] \cap \Rats\right) \times \Confs
\]
as the smallest relation satisfying the rules given in \Cref{fig:prelim:op_rules_tick}, where we write $\execrel{\confa}{\macta}{\proba}{\confa'}$ instead of $(\confa,\macta,\proba,\confa') \in\, \xrightarrow{}$. This yields $\pgcl$'s \emph{operational MDP}:
\begin{definition}[Operational MDP]
	\label{def:prelim:opmdp}
	The \emph{\highlight{operational MDP of $\pgcl$}} is%
	\begin{align*}
	\highlight{\opmdp} \eeq (\Confs, \, \mact, \, \mprob)~,\qquad\textnormal{where:}
	\end{align*}%
	%
	%
	\begin{enumerate}
		\item $\mact = \{\actn,\actl,\actr\}$ is the set of actions,
		\item $\mprob\colon \Confs \times \mact \times \Confs \to [0,1]$ is the transition probability function with
		\[
		\mprob(\confa,\macta,\confa') \eeq 
		\begin{cases}
			\proba &\text{if $\execrel{\confa}{\macta}{\proba}{\confa'}$} \\
			0 &\text{otherwise}~.
		\end{cases}
		\]
	\end{enumerate}
\end{definition}
\input{op_rules_tick}

\noindent%
The transition probability function $\mprob$ of the operational MDP $\opmdp$ formalizes the behavior of $\pgcl$ programs informally explained at the end of \Cref{sec:pgcl-syntax}:
If 
\[
\mprob((\cc,\statea),\macta,\confa') \eeq \proba 
\]
then executing $\cc$ for \emph{one} step on initial state $\statea$ when taking action $\macta$ yields the new configuration $\confa'$ with probability $\proba$. We have $\proba \in (0,1)$ only if executing $\cc$ for one step involves a probabilistic choice. For all other statements, we have $\proba \in \{0,1\}$. Moreover, we have $\size{\mact((\cc,\statea))} =  \{\actl,\actr\}$ iff executing $\cc$ for one step involves a nondeterministic choice, where $\actl$ is the action for the left branch and $\actr$ is the action for the right branch. If executing $\cc$ for one step involves one of the remaining statements, then $\mact((\cc,\statea)) = \{\actn\}$, indicating that no nondeterministic choice is possible.

In the above discussion, the $\TICK{\tickrew}$ statements were treated in exactly the same way as the effectless program $\SKIP$.
Indeed, the only influence of $\TICK{\tickrew}$ on the operational semantics of $\pgcl$ is that they induce the following reward function on $\opmdp$:%
%
\begin{definition}[Reward Function of Operational MDP]
	\label{def:oprew}
	The \emph{reward function $\oprew$} of the operational MDP $\opmdp = (\Confs, \, \mact, \, \mprob)$ for $\pgcl$ is the function $\oprew \colon \Confs \to \PosRealsInf$ with
	$\oprew(\opsink) = 0$ and
	\[
	\oprew\big((\cc,\statea)\big) = \begin{cases}
		\tickrew & \text{ if } \cc = \TICK{\tickrew} \\
        & \text{ or }\, \cc = \COMPOSE{\TICK{\tickrew}}{\cc'} \text{ for some } \cc' \in \pgcl\\
		0 & \text{ else.}
	\end{cases}
	\]
\end{definition}%

\begin{example}
  \Cref{fig:running-example-pgcl} shows a $\pgcl$ program $\cc$ that models our running example, i.e.\ the infinite-state MDP in \Cref{fig:running_example_mdp} on page~\pageref{fig:running_example_mdp}.
  A fragment of the operational MDP $\opmdp$ reachable from an initial configuration $(C,\statea)$ with $\statea(\vara) = \statea(\varb) = 0$ is sketched in \Cref{fig:opsem-example}.
  The reward that $\oprew$ (\Cref{def:oprew}) assigns to the configuration starting with $\langle \TICK{r};\ldots \rangle$ is $r$.
  All other states are assigned zero reward by $\oprew$.
  
  Note that $\opmdp$ has some additional \enquote{intermediate states} not present in the MDP from \Cref{fig:running_example_mdp}.
  Moreover, the latter MDP has several additional rewards attached to the states of the form $s_i^R$ for $i \in \Nats$.
  These rewards, which are special in the sense that they are collected right before reaching the terminal sink state $\opsink$, are not modeled (yet) by the reward function $\oprew$.
  In \Cref{sec:wp_tick} we extend $\oprew$ to account for those special rewards as well.
\end{example}

%% file: opsem_example.tex
\scriptsize
\begin{tikzpicture}[actstyle/.style={draw,circle,fill=black,inner sep=1pt},scale=0.8]

	\node[rectangle, rounded corners,draw=black] at (0,0) (s0) {$(\WHILEDO{\vara=0}{\ldots}, \, \vara=0 \wedge \varb=0)$};
	
	\node[rectangle, rounded corners,draw=black] at (0,-2) (s1) {$(\COMPOSE{\NDCHOICE{\ldots}{\ldots}}{\ldots}, \, \vara=0 \wedge\varb = 0)$};
	
	\node[rectangle, rounded corners,draw=black,label=165:\fbox{$r\vphantom{2}$}] at (0,-4) (s2) {$(\COMPOSE{\TICK{\tickrew}}{\ldots}, \, \vara=0 \wedge\varb = 0)$};
	\node[rectangle, rounded corners,draw=black] at (6,-4) (s3) {$(\COMPOSE{\PCHOICE{\ldots}{\nicefrac 1 2}{\ldots}}{\ldots}, \, \vara=0 \wedge\varb = 0)$};
	
	\node[rectangle, rounded corners,draw=black] at (0,-6) (s4) {$(\COMPOSE{\ASSIGN{\varb}{\varb+1}}{\ldots}, \, \vara=0 \wedge\varb = 0)$};
   
	\node[rectangle, rounded corners,draw=black] at (4.75,-6) (s5) {$(\COMPOSE{\SKIP}{\ldots}, \, \vara=0 \wedge\varb = 0)$};
    \node[rectangle, rounded corners,draw=black] at (3.75,-8) (s8) {$(\COMPOSE{\ASSIGN{\varb}{\varb+1}}{\ldots}, \, \vara=0 \wedge\varb = 0)$};
    
    \node[rectangle, rounded corners,draw=black] at (9.25,-6) (s7) {$(\COMPOSE{\ASSIGN{\vara}{1}}{\ldots}, \, \vara=0 \wedge\varb = 0)$};
    \node[rectangle, rounded corners,draw=black] at (9,-8) (s9) {$(\COMPOSE{\ASSIGN{\varb}{\varb+1}}{\ldots}, \, \vara=1 \wedge\varb = 0)$};
    
    \node[rectangle, rounded corners,draw=black] at (9,-10) (s10) {$(\WHILEDO{\vara=0}{\ldots}, \, \vara=1 \wedge\varb = 1)$};
	
	\node[rectangle, rounded corners,draw=black] at (0,-10) (s6) {$(\WHILEDO{x=0}{\ldots}, \, \vara=0 \wedge\varb = 1)$};
    
   	\node[rectangle, rounded corners,draw=black] at (9,-12) (st) {$(\term, \, \vara=1 \wedge\varb = 1)$};
   	\node[rectangle, rounded corners,draw=black] at (11.5,-12) (sink) {$\opsink$};
	
	\node at (0, -12) (sf1) {$\vdots$};
	
	\path[->]
		(s0) edgenode[pos=0.5,actstyle]  {} node[pos=0.25, left] {\scriptsize $\actn$} node[pos=0.75, right] {\scriptsize $1$} (s1);
	;
	
	\path[->]
	(s1) edgenode[pos=0.5,actstyle]  {} node[pos=0.25, left,yshift=2pt] {\scriptsize $\actl$} node[pos=0.75, right,yshift=-2pt] {\scriptsize $1$} (s2);
	;
	
	\path[->]
	(s1) edgenode[pos=0.5,actstyle] {} node[pos=0.25, right,yshift=2pt] {\scriptsize $\actr$} node[pos=0.75, left,yshift=-2pt] {\scriptsize $1$} (s3);
	;
	
	\path[->]
	(s2) edgenode[pos=0.5,actstyle]  {} node[pos=0.25, left] {\scriptsize $\actn$} node[pos=0.75, right] {\scriptsize $1$} (s4);
	;
	
	\path[->]
	(s4) edgenode[pos=0.5,actstyle]  {} node[pos=0.25, left] {\scriptsize $\actn$} node[pos=0.75, right] {\scriptsize $1$} (s6);
	;

	\path[->]
	(s3) edgenode[pos=0.5,actstyle] (choice) {} node[pos=0.25, left] {\scriptsize $\actn$} node[pos=0.75, right] {\scriptsize $\nicefrac 1 2$} (s5);
	;
    
    \draw[->]
    (choice) -- node[auto] {$\nicefrac 1 2$} (s7);
    ;
    
    \path[->]
    (s8) edgenode[pos=0.5,actstyle]  {} node[pos=0.25, left] {\scriptsize $\actn$} node[pos=0.75, right] {\scriptsize $1$} (s6);
    ;
	
	\path[->]
	(s5) edgenode[pos=0.5,actstyle]  {} node[pos=0.25, left] {\scriptsize $\actn$} node[pos=0.75, right] {\scriptsize $1$} (s8);
	;
    
    \path[->]
    (s7) edgenode[pos=0.5,actstyle]  {} node[pos=0.25, left] {\scriptsize $\actn$} node[pos=0.75, right] {\scriptsize $1$} (s9);
    ;
	
		\path[->]
	(s6) edgenode[pos=0.5,actstyle]  {} node[pos=0.25, left,,yshift=1pt] {\scriptsize $\actn$} node[pos=0.75, right,yshift=-1pt] {\scriptsize $1$} (sf1);
	;

    \path[->]
    (s9) edgenode[pos=0.5,actstyle]  {} node[pos=0.25, left] {\scriptsize $\actn$} node[pos=0.75, right] {\scriptsize $1$} (s10);
    ;
    
    \path[->]
    (s10) edgenode[pos=0.5,actstyle]  {} node[pos=0.25, left] {\scriptsize $\actn$} node[pos=0.75, right] {\scriptsize $1$} (st);
    ;
    
    \path[->]
    (st) edgenode[pos=0.5,actstyle]  {} node[pos=0.25,above] {\scriptsize $\actn$} node[pos=0.75,below] {\scriptsize $1$} (sink);
    ;
    
    	\path[->]
    (sink) edge[loop right,looseness=15] node[pos=0.5,actstyle,xshift=-2pt]  {} node[pos=0.25,above] {\scriptsize $\actn$} node[pos=0.75, below] {\scriptsize $1$} (sink)
    ;
    
\end{tikzpicture}%
\normalsize

%% file: op_rules_tick.tex
\begin{figure}[p]
    \scriptsize
        \begin{enumerate}
            \item Final configurations: \[\infer{
                \execrel{(\term,\statea)}{\actn}{1}{\opsink}
            }{
            }
            \qquad
            \infer{
                \execrel{\opsink}{\actn}{1}{\opsink}
            }{
            }
            \]

            \item Effectless Program: \[\infer{
                \execrel{(\SKIP,\statea)}{\actn}{1}{(\term,\statea)}
            }{}\]
            \item Assignment:\[
            \infer{
                \execrel{(\ASSIGN{x}{\Axa},\statea)}{\actn}{1}{(\term,\statea\statesubst{x}{\Axa(\statea)})}
            }{
            }
            \]%
            \item Reward:\[
            \infer{
                \execrel{(\TICK{\tickrew},\statea)}{\actn}{1}{(\term,\statea)}
            }{
            }
            \]%
            \item Sequential Composition:
            \[
            \infer{
                \execrel{(\COMPOSE{\cc_1}{\cc_2},\statea)}{\macta}{\proba}{(\cc_2,\statea')}
            }{
                \execrel{(\cc_1,\statea)}{\macta}{\proba}{(\term,\statea')}
            }
            \qquad
            %
            \infer{
                \execrel{(\COMPOSE{\cc_1}{\cc_2},\statea)}{\macta}{\proba}{(\COMPOSE{\cc_1'}{\cc_2},\statea')}
            }{
                \execrel{(\cc_1,\statea)}{\macta}{\proba}{(\cc_1',\statea')}
            }
            \]
            \item Probabilistic Choice:
            \[
            \infer{
                \execrel{(\PCHOICE{\cc_1}{\proba}{\cc_2},\statea)}{\actn}{\proba}{(\cc_1,\statea)}
            }{
                \cc_1 \neq \cc_2
            }
            \qquad
            %
            \infer{
                \execrel{(\PCHOICE{\cc_1}{\proba}{\cc_2},\statea)}{\actn}{1-\proba}{(\cc_2,\statea)}
            }{
                \cc_1 \neq \cc_2
            }
            \]
            \phantom{a}
            \[
            \infer{
                \execrel{(\PCHOICE{\cc}{\proba}{\cc},\statea)}{\actn}{1}{(\cc,\statea)}
            }{
            }
            \]
            \item Nondeterministic Choice:
            \[
            \infer{
                \execrel{(\NDCHOICE{\cc_1}{\cc_2},\statea)}{\actl}{1}{(\cc_1,\statea)}
            }{
            }
            \qquad
            %
            \infer{
                \execrel{(\NDCHOICE{\cc_1}{\cc_2},\statea)}{\actr}{1}{(\cc_2,\statea)}
            }{
            }
            \]
            \item Conditional Choice:
            \[
            \infer{
                \execrel{(\ITE{\Bxa}{\cc_1}{\cc_2},\statea)}{\actn}{1}{(\cc_1,\statea)}
            }{
                \statea \models \Bxa
            }
            \qquad
            %
            \infer{
                \execrel{(\ITE{\Bxa}{\cc_1}{\cc_2},\statea)}{\actn}{1}{(\cc_2,\statea)}
            }{
                \statea \models \neg\Bxa
            }
            \]
            \item While Loop:
            \[
            \infer{
                \execrel{(\WHILEDO{\Bxa}{\cc},\statea)}{\actn}{1}{
                    (\COMPOSE{\cc}{\WHILEDO{\Bxa}{\cc}},\statea)
                }
            }{
                \statea \models \Bxa
            }
            \qquad
            %
            \infer{
                \execrel{(\WHILEDO{\Bxa}{\cc},\statea)}{\actn}{1}{(\term,\statea)}
            }{
                \statea \models \neg\Bxa
            }
            \]
        \end{enumerate}
    \caption{Rules defining the small-step execution relation $\xrightarrow{}$.}
    \label{fig:prelim:op_rules_tick}
\end{figure}

%% file: wp_tick.tex

We now present two weakest preexpectation calculi that extend McIver \& Morgan's \cite{DBLP:series/mcs/McIverM05} weakest preexpectations by support for reasoning about expected rewards of $\pgcl$ programs with user-defined reward structures.

Those calculi only differ in their treatment of nondeterministic choice, which may be either demonic (i.e., minimizing) or angelic (i.e. maximizing).
The angelic calculus is similar to the expected runtime calculus~\cite{DBLP:conf/esop/KaminskiKMO16,DBLP:journals/jacm/KaminskiKMO18} if rewards are interpreted as required execution time. To specify a runtime model, which is fixed a-priori in~\cite{DBLP:conf/esop/KaminskiKMO16,DBLP:journals/jacm/KaminskiKMO18}, it then sufficies to place $\TICK{\cdot}$ commands.
%

In \Cref{sec:prelim:wp:expectations}, we define \emph{expectations} -- the objects our calculi operate on.
We present the calculi in \Cref{sec:prelim:wp:calculus}.
Finally, in \Cref{sec:prelim:wp:wp_vs_op}, we leverage \Cref{thm:mdp:bellman_correct} to prove both calculi sound with respect to the MDP semantics from \Cref{sec:pgcl-syntax}.

\subsection{Expectations}
\label{sec:prelim:wp:expectations}
Expectations (think: random variables) are for verification of probabilistic programs what predicates are for verification of ordinary programs.
Technically, they are similar to the complete lattice of value functions introduced in \Cref{def:value-functions}.
The main difference is that the domain of expectations consists of \emph{program} states, whereas the domain of value functions consists of \emph{MDP} states.
We briefly recall the definition of expectations and some notation for conveniently describing them.
\begin{definition}[Expectations]
    The complete lattice of \emph{\highlight{expectations}} is
    \begin{align*}
    \highlight{(\E, \, \eleq)}~, \qquad\textnormal{where:}
    \end{align*}%
    \normalsize%
    \begin{enumerate}
        \item $\highlight{\E} = \States \to \PosRealsInf$ is the \emph{\highlight{set of expectations}},
        \item $\highlight{\eleq}$ is the pointwise lifted order in $\PosRealsInf$, i.e., $\tforall{\FF,\FG\in \E}$
        \[
        \FF \eleq \FG \qiff \tforall{\statea\in\States} \colon \FF(\statea) \leq \FG(\statea)~.
        \]
    \end{enumerate}
\end{definition}
The least element of $(\E, \, \eleq)$ is the constant-zero-expectation $\mylambda{\statea} 0$. The greatest element is $\mylambda{\statea} \infty$. Moreover, suprema and infima are given as the pointwise liftings of suprema and infima in $(\PosRealsInf, \, \leq)$, i.e., for $\Eset \subseteq \E$, we have 
\[
\supremum \Eset \eeq \mylambda{\statea} \sup \setcomp{\FF(\statea)}{\FF \in \Eset}
\qquad\text{and}\qquad
\infimum \Eset \eeq \mylambda{\statea} \inf \setcomp{\FF(\statea)}{\FF \in \Eset}~.
\]
Given two expectations $\FF,\FG \in \E$, we often write
\[
\FF \sqcup \FG ~\text{instead of}~\supremum\{\FF,\FG\}
\qquad\text{and similarly}\qquad
\FF \sqcap \FG ~\text{instead of}~\infimum\{\FF,\FG\}~.
\]
We define standard arithmetic operations on expectations by pointwise application, e.g.
\[
\FF\cdot \FG \eeq \mylambda{\statea} \FF(\statea)\cdot\FG(\statea)
\qquad \text{and}\qquad
\FF+ \FG \eeq \mylambda{\statea} \FF(\statea)+\FG(\statea)~,
\]
where we recall that $0\cdot \infty = \infty \cdot 0 = 0$. The order of precedence is standard, i.e., $\cdot$ binds stronger than $+$, and we use parentheses to resolve ambiguities. Moreover, we often identify constants $\reala \in \PosRealsInf$ and variables $\vara \in \Vars$ with expectations. That is, we write $\vara$ for the expectation $\mylambda{\statea} \statea(\vara)$ and $\reala$ for the expectation $\mylambda{\statea} \reala$.

We use \emph{\highlight{Iverson brackets}} $\highlight{\iverson{\preda}}$ to turn predicates $\preda \in \Preds$ into expectations, i.e. 
\[
\highlight{\iverson{\preda}} \eeq \mylambda{\statea}
\begin{cases}
    1 &\text{if $\statea \models \preda$} \\
    0 & \text{if $\statea \not\models \preda$}~.
\end{cases}
\]%
Finally, given $\FF \in \E$, a variable $\vara$, and an arithmetic expression $\Axa \colon \States \to \Vals$, the \emph{\highlight{expectation $\FF\expsubst{\vara}{\Axa}$ obtained from substituting $\vara$ in $\FF$ by $\Axa$}} is given by
\[
\highlight{\FF\expsubst{\vara}{\Axa}} \eeq \mylambda{\statea} \FF(\statea\statesubst{\vara}{\Axa(\statea)})~.
\]

\subsection{Weakest Preexpectations}
\label{sec:prelim:wp:calculus}
We now define two weakest preexpectation calculi by induction on the structure of $\pgcl$ programs.
Moreover, we discuss some properties of these calculi that are required for justifying that they are well-defined and sound.

\input{table_wp_tick}
\begin{definition}[Weakest Preexpectation Transformers]
    \label{def:prelim:wp:wp}
    Let $\cc \in \pgcl$ and $\FF \in \E$. 
    \begin{enumerate}
        \item The \emph{\highlight{demonic weakest preexpectation transformer}}  \emph{\highlight{of $\cc$}}
        \[
        \highlight{\dwptrans{\cc}} \colon \E \to \E
        \]
        is defined by induction on $\cc$ in \Cref{tab:prelim:wp}. We call
        $
        \highlight{\dwp{\cc}{\FF}}
        $
        the \emph{\highlight{demonic weakest preexpectation of $\cc$ w.r.t.\ postexpectation $\FF$}}.
        \item The \emph{\highlight{angelic weakest preexpectation transformer}} \emph{\highlight{of $\cc$}}
        \[
        \highlight{\awptrans{\cc}} \colon \E \to \E
        \]
        is defined by induction on $\cc$ in \Cref{tab:prelim:wp} after replacing every occurrence of $\dwpsymbol$ by $\awpsymbol$ and every occurrence of $\sqcap$ by $\sqcup$. We call
        $
        \highlight{\awp{\cc}{\FF}}
        $
        the \emph{\highlight{angelic weakest preexpectation of $\cc$ w.r.t.\ postexpectation $\FF$}}.
    \end{enumerate}
\end{definition}
\noindent
The transformers $\dwptrans{\cc}$ and $\awptrans{\cc}$ only differ in how they resolve nondeterminism: \emph{d}emonically (minimizing using $\sqcap$) versus \emph{a}ngelically (maximizing using $\sqcup$). 
Our main goal for this section is to formalize (and prove in \Cref{sec:prelim:wp:wp_vs_op}) that
\begin{align*}
	& \dwpsymbol/\awp{\cc}{\FF}(\statea) \eeq
	\hspace{-7pt}\substack{\text{\normalsize{minimal/maximal expected reward obtained from} } \\
		 \text{\normalsize (i) executing $\cc$ on \emph{initial} state $\statea$ and} 
			\\
			\text{\normalsize (ii) collecting reward $\FF(\stateb)$ upon terminating in some $\stateb$.}}
\end{align*}
In particular, the postexpectation $\FF$ can be thought of as a continuation which models the expected reward obtained \emph{after} executing $\cc$. Notice that if $\cc$ does not contain a reward $\TICK{\cdot}$ statement, then both calculi coincide with McIver \& Morgan's classical weakest preexpectation calculi \cite{DBLP:series/mcs/McIverM05} for reasoning about expected outcomes. On the other hand, for programs $\cc$ (possibly) containing $\TICK{\cdot}$ and postexpectation $\FF = 0$, the two calculi yield the minimal/maximal expected reward obtained from executing $\cc$ --- similarly to Kaminski et al.'s expected runtime calculus \cite{DBLP:conf/esop/KaminskiKMO16,DBLP:journals/jacm/KaminskiKMO18}.
%
%

Apart from the $\TICK{r}$, which adds the just collected reward $r$ to the given postexpectation $\FF$, the weakest preexpectation of each $\pgcl$ command is standard (cf.~\cite{kaminski_diss}). 
For assignments, we substitute the assigned variable $\vara$ by expression $\Axa$ in postexpectation $\FF$.
The rule for sequential composition $\COMPOSE{\cc_1}{\cc_2}$ demonstrates that weakest preexpectation calculi perform backward reasoning: We first determine the weakest preexpectation of $\cc_2$ wrt. postexpectation $\FF$ and then use this intermediate result as the postexpectation of $\cc_1$. The weakest preexpectation wrt. that intermediate result then yields the weakest preexpectation of $\COMPOSE{\cc_1}{\cc_2}$ wrt. $\FF$.
For nondeterministic, probabilistic, and conditional choices between $\cc_1$ and $\cc_2$, we first determine the weakest preexpectations of the branches $\cc_1$ and $\cc_2$ wrt. postexpectation $\FF$. 
For the nondeterministic choice $\NDCHOICE{\cc_1}{\cc_2}$, we then take the minimum -- for $\dwpsymbol$ -- or the maximum -- for $\awpsymbol$ -- of the previously determined expectations. 
For the probabilistic choice $\PCHOICE{\cc_1}{\proba}{\cc_2}$, we sum those expectations after weighting each branch with its probability.
The conditional choice $\ITE{\Bxa}{\cc_1}{\cc_2}$ is similar except that the weights are always zero or one, depending on whether the guard $\Bxa$ holds or not.

Finally, the weakest preexpectation of loops is defined using a least fixed point construction. To get an intuition on this construction, we first observe that the weakest preexpectations of the two programs
\[
\WHILEDO{\Bxa}{\cc'} \quad \text{and}\quad
\ITE{\Bxa}{\COMPOSE{\cc'}{\WHILEDO{\Bxa}{\cc'}}}{\SKIP}
\]
must coincide, since these two programs are semantically equivalent, see also \Cref{fig:prelim:op_rules_tick} (9).
 Hence, by the rules for $\SKIP$, sequential composition, and conditional choice from \Cref{tab:prelim:wp}, $\dwp{\WHILEDO{\Bxa}{\cc'}}{\FF}$ (and analogously $\awpsymbol$) must satisfy the equation
\begin{align}
    &\dwp{\WHILEDO{\Bxa}{\cc'}}{\FF} \notag\\
    \eeq& \iverson{\Bxa} \cdot \dwp{\cc'}{\dwp{\WHILEDO{\Bxa}{\cc'}}{\FF}} + \iverson{\neg \Bxa} \cdot \FF~,
    \label{eqn:prelim:wp:recursive_loops}
\end{align}
and the \emph{least} solution of this equation in the complete lattice $(\E, \, \eleq)$ precisely captures (minimal or maximal) expected outcomes of loops. 
More formally, we associate each loop and each postexpectation with an endomap $\E \to \E$ that captures the effect of perform one (more) loop iteration.

\begin{definition}[Characteristic Functions]
    \label{def:prelim:wp:charfuns}
    Let $\cc =\WHILEDO{\Bxa}{\cc'} \in\pgcl$ and $\FF\in\E$. 
    \begin{enumerate}
        \item The \emph{\highlight{$\dwpsymbol$-characteristic function of $\cc$ w.r.t.\ $\FF$}} is defined as
        %
        %
        \[
        \highlight{\dcharfun{\cc}{\FF}}\colon \E \to \E~,
        \qquad
        \dcharfun{\cc}{\FF}(\FG)\eeq \iverson{\Bxa}\cdot\dwp{\cc'}{\FG} + \iverson{\neg\Bxa}\cdot \FF~.
        \]
        \item The \emph{\highlight{$\awpsymbol$-characteristic function of $\cc$ w.r.t.\ $\FF$}} is defined as
        %
        %
        \[
        \highlight{\acharfun{\cc}{\FF}}\colon \E \to \E~,
        \qquad
        \acharfun{\cc}{\FF}(\FG)\eeq \iverson{\Bxa}\cdot\awp{\cc'}{\FG} + \iverson{\neg\Bxa}\cdot \FF~.
        \]
    \end{enumerate}
\end{definition}
The least solution of \Cref{eqn:prelim:wp:recursive_loops} is then given by the least fixed point of the respective characteristic function, i.e.,
\[
\dwp{\cc}{\FF} \eeq \lfp \dcharfun{\cc}{\FF} \qquad \text{and}\qquad \awp{\cc}{\FF} \eeq \lfp \acharfun{\cc}{\FF}~.
\]
As discussed in \Cref{sec:prelim:domain_theory}, these least fixed points exist if characteristic functions are monotonic. Moreover, they can be approximated iteratively via fixed point iteration (cf. \Cref{thm:prelim:domain_theory:kleene}) if the characteristic function is also continuous.%
\begin{theorem}[Healthiness Properties~\cite{DBLP:journals/jacm/KaminskiKMO18}]
    \label{thm:prelim:wp:cont}
    Let $\cc \in \pgcl$ and $\FF\in\E$.
    \begin{enumerate}
        \item\label{thm:prelim:wp:cont:mono} Both $\dwptrans{\cc}$ and $\awptrans{\cc}$ are monotonic w.r.t.\ $(\E, \, \eleq)$. 
        \item\label{thm:prelim:wp:cont1} Both $\dwptrans{\cc}$ and $\awptrans{\cc}$ are continuous w.r.t.\ $(\E, \, \eleq)$.
        \item\label{thm:prelim:wp:cont2} For loop $\cc = \WHILEDO{\Bxa}{\cc'}$, both $\dcharfun{\cc}{\FF}$ and $\acharfun{\cc}{\FF}$ are continuous w.r.t.\ $(\E, \, \eleq)$. Hence, by Kleene's \Cref{thm:prelim:domain_theory:kleene}, we have
        \[
        \lfp \dcharfun{\cc}{\FF} \eeq \supremum_{\nata \in \Nats} \dcharfuniter{\cc}{\FF}{\nata}(0)
        \qquad\text{and}\qquad
        \lfp \acharfun{\cc}{\FF} \eeq  \supremum_{\nata \in \Nats} \acharfuniter{\cc}{\FF}{\nata}(0)~.
        \]%
    \end{enumerate}%
\end{theorem}%
As we will see next, proving $\dwpsymbol$ and $\awpsymbol$ sound with respect to our operational semantics introduced in \Cref{sec:prelim:pgcl:opsemtick} involves showing that the above fixed points are equal to the expected rewards obtained from unfolding loops on demand as in \Cref{fig:prelim:op_rules_tick} (9).

\begin{example}
    \label{ex:lfpsOfExampleLoop}
    Recall our running example $\pgcl$ program $\cc$ from \Cref{fig:running-example-pgcl} on page~\pageref{fig:running-example-pgcl}.
    We fix postexpectation $\FF = y ~ \gray{ = \mylambda{\statea}\statea(y)}$.
    Then, for all concrete values of the constant $r \in \PosRatsInf$ in $\cc$, we have
    \begin{align*}
        &\dwp{\cc}{y} \eeq\lfp \dcharfun{\cc}{y}
        \qeq
        \iverson{x\neq0} \cdot y + \iverson{x=0} \cdot (y + 2) \qquad \text{and}\\[1ex]
        &\awp{\cc}{y} \eeq \lfp \acharfun{\cc}{y}
        \qeq
        \iverson{x\neq0} \cdot y + \iverson{x=0} \cdot \infty
        ~.
    \end{align*}%
    Intuitively, under \emph{demonic} nondeterminism, we try to avoid collecting additional reward as much as possible, i.e.\ we never execute $\TICK{r}$.
    Hence, $\dwp{\cc}{y}$ is equal to the expected value of $y$ after termination of $\cc$, which is the original value of $y$ if $x \neq 0$ holds initially as we never enter the loop. 
    Otherwise, i.e.\ for $x = 0$, the expected value  of $y$ is $y + 2$ as the loop body is, on average, executed twice and $y$ is incremented in every execution.
    Under \emph{|angelic} nondeterminism, on the other hand, we try to collect as much reward as possible. If we can enter the loop at all, i.e. if $x \neq 0$ holds initially, we thus collect an infinite amount of reward.
    Otherwise, i.e. if $x = 0$, we collect no additional reward beyond the already collected $y$.
    
    To confirm the above intuitive explanations of the two fixed points, we will connect them to expected rewards of operational MDPs in the next section.
\end{example}

\subsection{Soundness of the Weakest Preexpectation Calculus}
\label{sec:prelim:wp:wp_vs_op}

We now leverage our generalized fixed point characterization of expected rewards, i.e. \Cref{thm:mdp:bellman_correct}, to prove $\dwpsymbol$ and $\awpsymbol$ are sound with respect to our operational MDP semantics.
More precisely, we show that, for all $\cc\in\pgcl$ and all $\FF\in\E$,
$\dwp{\cc}{\FF}$ and $\awp{\cc}{\FF}$ evaluate to appropriate minimal and maximal expected rewards in $\pgcl$'s operational MDP $\opmdp$. 
An analogous result has been proven in \cite{gretz_op} for \emph{bounded} expectations\footnote{An expectation $\FF$ is \emph{bounded}, if there exists $\reala \in \Reals$ such that $\FF(\statea)\leq \reala$ for all $\statea\in\States$.} and $\pgcl$ programs not containing $\TICK{\cdot}$. Our proofs very closely follow the lines of \cite[Theorem 4.5]{qsl_popl} (for $\dwpsymbol$) and \cite[Theorem 4.6]{aert} (for $\awpsymbol$), unifying these results with minor simplifications.

We first associate with each postexpectation a reward function for $\opmdp$:
\begin{definition}[Rew.\ Function induced by Postexpectation]
    \label{def:to_reach}
    Let $\FF \in \E$.
    The reward function $\highlight{\torew{\FF}} \colon \Confs \to \PosRealsInf$ for the operational MDP $\opmdp$ is \mbox{defined as}
    \label{def:torew}
    \[
    \torew{\FF}(\confa) \eeq \begin{cases}
        \FF(\stateb) & \text{ if } \confa = (\term,\stateb) \\
        \oprew(\confa) & \text{ else}~,
    \end{cases}
    \]
    where $\oprew$ is the reward function of \Cref{def:oprew} (which is \mbox{independent of $\FF$)}.
\end{definition}
\begin{example}
    Recall once again the running example $\pgcl$ program from \Cref{fig:running-example-pgcl} on page~\pageref{fig:running-example-pgcl} that models the MDP in \Cref{fig:running_example_mdp} on page~\pageref{fig:running_example_mdp}.
    With postexpectation $\FF = y$, the reward function $\torew{y}$ (\Cref{def:to_reach}) on the operational MDP $\opmdp$ faithfully models the rewards depicted in boxes in \Cref{fig:running_example_mdp}.\triangleqed
\end{example}%
%
%
It is convenient to introduce expectation transformers capturing the expected rewards obtained from executing $\pgcl$ programs as modeled by the MDP $\opmdp$:%
\begin{definition}[Operational Weakest Preexpectations]
    \label{def:prelim:wp:op}
    Let $\cc \in \pgcl$ and $\FF \in \E$. 
    \begin{enumerate}
        \item We define the \emph{\highlight{demonic operational weakest preexpectation of $\cc$ w.r.t.\ $\FF$}} as
        \[
        \highlight{\dop{\cc}{\FF}} \eeq \mylambda{\statea} \newmmminexprew{\torew{\FF}}{(\cc,\statea)}{\opmdp}~.
        \]
        \item We define the \emph{\highlight{angelic operational weakest preexpectation of $\cc$ w.r.t.\ $\FF$}} as
        \[
        \highlight{\aop{\cc}{\FF}} \eeq \mylambda{\statea} \newmmmaxexprew{\torew{\FF}}{(\cc,\statea)}{\opmdp}~.
        \]
    \end{enumerate}
\end{definition}
Proving soundness then amounts to showing that, for all $\cc\in\pgcl$ and $\FF\in\E$, we have
\[
\dwp{\cc}{\FF} \eeq \dop{\cc}{\FF} \quad\text{and}\quad \awp{\cc}{\FF} \eeq \aop{\cc}{\FF}~.
\]
To prove equality, we show the inequalities $\eleq$ and $\sqsupseteq$ seperately.
%
\begin{lemma}
    \label{lem:prelim:wp:sound_leq}
    We have for all $\cc \in \pgcl$ and all $\FF \in \E$:
    \begin{enumerate}
        \item $\dop{\cc}{\FF} \eleq \dwp{\cc}{\FF}$
        \item $\aop{\cc}{\FF} \eleq \awp{\cc}{\FF}$
    \end{enumerate}
\end{lemma}
\begin{proof}
    We prove the claim for $\dwpsymbol$. The proof for $\awpsymbol$ is completely analogous. The key idea is to apply Park induction (\Cref{lem:prelim:domain_theory:park}) to the min-Bellman operator of $\pgcl$'s operational MDP $\opmdp$ w.r.t.\ $\torew{\FF}$ (\Cref{def:mdp:bellman}.\ref{def:mdp:bellman1}). For that, we first observe that for all $\cc\in\pgcl$ and all $\FF \in \E$, we have
    \begin{align}
        &\dwp{\cc}{\FF} \eeq \mylambda{\statea} \torew{\FF}\big((\cc,\statea)\big) \notag\\
        &\qquad ~+~ \min_{\macta \in \mact(\cc,\statea)} \sum_{\execrel{(\cc,\statea)}{\macta}{\proba}{\confa'}} \proba \cdot
        \begin{cases}
            \FF(\stateb) & \text{if $\confa' = (\term,\stateb)$} \\
            \dwp{\cc'}{\FF}(\statea') & \text{if $\confa' = (\cc',\statea')$}~,
        \end{cases}
        \label{eqn:prelim:wp:opproof1}
    \end{align}
    which follows by induction on the structure of $\cc$ (cf.\ \Cref{lem:app:prelim:wp:bellman_compliant_tick}).
    Now define the value function $\mval$ for $\pgcl$'s operational MDP $\opmdp$ as
    \[
    \mval \eeq \mylambda{\confa} 
    \begin{cases}
        \dwp{\cc}{\FF}(\statea) & \text{if $\confa = (\cc,\statea)$} \\
        \FF(\stateb) &\text{if $\confa = (\term,\stateb)$} \\
        0 & \text{if $\confa = \opsink$}~.
    \end{cases}
    \]
    It is an immediate consequence of \Cref{eqn:prelim:wp:opproof1} that%
    \[
    \mopmin{\opmdp}{\torew{\FF}}(\mval) \eeq \mval \mmleq \mval~,
    \]%
    which, by Park induction (\Cref{lem:prelim:domain_theory:park}), yields%
    \[
    \lfp \mopmin{\opmdp}{\torew{\FF}} \mmleq \mval~.
    \]%
    By \Cref{thm:mdp:bellman_correct}.\ref{thm:mdp:bellman_correct1}, we thus get%
    \[
    \lfp \mopmin{\opmdp}{\torew{\FF}}  \eeq \mylambda{\confa} \newmmminexprew{\torew{\FF}}{\confa}{\opmdp} \mmleq \mval
    \]%
    so, in particular, we have for every $\cc\in\pgcl, \FF\in\E$, and $\statea\in\States$,%
    \begin{align*}
        \dop{\cc}{\FF}(\statea)
        \eeq& \newmmminexprew{\torew{\FF}}{(\cc,\statea)}{\opmdp} \\
        \lleq& \mval(\cc,\statea) \\
        \eeq& \dwp{\cc}{\FF}(\statea)~,
    \end{align*}%
    which is what we had to show. 
\end{proof}
Next, we prove the converse direction.
\begin{lemma}
    \label{lem:prelim:wp:sound_geq}
    We have for all $\cc \in \pgcl$ and all $\FF \in \E$:
    \begin{enumerate}
        \item $\dwp{\cc}{\FF} \eeleq \dop{\cc}{\FF}$
        \item $\awp{\cc}{\FF} \eeleq \aop{\cc}{\FF}$.
    \end{enumerate}
\end{lemma}
\begin{proof}
    We prove the claim for $\dwpsymbol$. The proof for $\awpsymbol$ is completely analogous. The key observation is that $\dopsymbol$ satisfies the following big-step decomposition w.r.t.\ sequential composition: For all $\cc_1,\cc_2\in\pgcl$ and all $\FF\in \E$,
    \begin{align*}
        \dop{\cc_1}{\dop{\cc_2}{\FF}} \eeleq \dop{\COMPOSE{\cc_1}{\cc_2}}{\FF} ~.
        \tag{see \Cref{app:lem:op_decomp}}
    \end{align*}
    The claim then follows by induction on $\cc$. The most interesting case is $\cc = \WHILEDO{\Bxa}{\cc'}$: We prove that $\dcharfun{\cc}{\FF}(\dop{\cc}{\FF}) \eleq \dop{\cc}{\FF}$, which implies $\dwp{\cc}{\FF} \eeleq \dop{\cc}{\FF}$ by \Cref{lem:prelim:domain_theory:park}.
	%
    %
     For that, consider the following:%
    \begin{align*}
        \dcharfun{\cc}{\FF}(\dop{\cc}{\FF})
        \eeq& \iverson{\Bxa} \cdot \dwp{\cc'}{\dop{\cc}{\FF}} + \iverson{\neg\Bxa} \cdot \FF
        \tag{\Cref{def:prelim:wp:charfuns}} \\
        \eeleq & \iverson{\Bxa} \cdot \dop{\cc'}{\dop{\cc}{\FF}} + \iverson{\neg\Bxa} \cdot \FF
        \tag{I.H.} \\
        \eeleq & \iverson{\Bxa} \cdot \dop{\COMPOSE{\cc'}{\cc}}{\FF} + \iverson{\neg\Bxa} \cdot \FF
        \tag{\Cref{app:lem:op_decomp}} \\
        %
        \eeleq & \mylambda{\statea}
        \begin{cases}
            \newmmminexprew{\torew{\FF}}{(\COMPOSE{\cc'}{\cc},\statea)}{\opmdp} &\text{if $\statea\models \Bxa$} \\
            \FF(\statea)&\text{if $\statea\not\models \Bxa$}
        \end{cases}
        \\
        \eeleq & \mylambda{\statea}
        \newmmminexprew{\torew{\FF}}{(\cc,\statea)}{\opmdp}  
        \tag{\Cref{def:mdp:exrew}.\ref{def:mdp:exrew3} and \Cref{fig:prelim:op_rules_tick}} \\
        \eeq& \dop{\cc}{\FF}~.
        \tag*{(\Cref{def:prelim:wp:op})\qquad\qed}
    \end{align*}%
\end{proof}%
\Cref{lem:prelim:wp:sound_leq,lem:prelim:wp:sound_geq}, and antisymmetry of $\eleq$ now yield the desired claim.%
\begin{theorem}[Soundness of Weakest Preexpectations]
    \label{thm:prelim:wp:sound}
    For all $\cc \in \pgcl$ and $\FF \in \E$:%
    \begin{enumerate}
        \item\label{thm:prelim:wp:sound1} $ \dwp{\cc}{\FF} \eeq \dop{\cc}{\FF}  $
        \item\label{thm:prelim:wp:sound2} $ \awp{\cc}{\FF} \eeq \aop{\cc}{\FF} $
    \end{enumerate}%
\end{theorem}%
\begin{example}
    We consider the running example $\pgcl$ program $\cc$ from \Cref{fig:running-example-pgcl} on page~\pageref{fig:running-example-pgcl} one last time.
    From \Cref{ex:lfpsOfExampleLoop} we already know that%
    \begin{align*}
        &\dwp{\cc}{y}
        \eeq
        \iverson{x\neq0} \cdot y + \iverson{x=0} \cdot (y + 2) \qquad \text{and}\\[1ex]
        &\awp{\cc}{y}
        \eeq
        \iverson{x\neq0} \cdot y + \iverson{x=0} \cdot \infty
        ~.
    \end{align*}%
    Thus, the link between $\dwpsymbol$/$\awpsymbol$ and expected rewards in the operational MDP (see \Cref{thm:prelim:wp:sound}) shows that, for every initial program state $\statea$ with $\statea(x) \neq 0$,%
    \begin{align*}
        &\newmmminexprew{\torew{y}}{(\cc,\statea)}{\opmdp}
        \eeq
        \statea(y) + 2 \qquad \text{and}\\[1ex]
        &\newmmmaxexprew{\torew{y}}{(\cc,\statea)}{\opmdp}
        \eeq
        \infty
        ~.
    \end{align*}%
    Note that this is consistent with our observations from \Cref{ex:maxERminER}.
    
    Finally, we remark that we are now in a position to see the relevance of \emph{unbounded} reward functions and \emph{infinite} expected rewards:
    Establishing the link 
    \[
    \awp{\cc}{\vara} \eeq \mylambda{\statea} \newmmmaxexprew{\torew{y}}{(\cc,\statea)}{\opmdp}
    \]
    from \Cref{thm:prelim:wp:sound} works only if the right-hand side may evaluate to $\infty$. 
\end{example}

%% file: table_wp_tick.tex
\begin{table}[t]
    \begin{center}
        \begin{tabularx}{\textwidth}{X@{\quad}l@{\quad~~}lX}
            $\boldsymbol{\cc}$ & $\mathsf{\mathbf{dwp}}\boldsymbol{\llbracket\cc \rrbracket(\FF)}$  \\[0.5ex]
            \hline 
            $\SKIP$ & $\FF$  \rule{0pt}{3.5ex}\\[1.8ex]
            $\ASSIGN{\vara}{\Axa}$ & $\FF\expsubst{\vara}{\Axa}$   \\[1.8ex]
            $\TICK{\tickrew}$ & $\tickrew + \FF$   \\[1.8ex]
            $\COMPOSE{\cc_1}{\cc_2}$ & $\dwp{\cc_1}{\dwp{\cc_2}{\FF}}$  \\[1.8ex]
            $\NDCHOICE{\cc_1}{\cc_2}$ & $\dwp{\cc_1}{\FF} \sqcap \dwp{\cc_2}{\FF}$ \\[1.8ex]
            $\PCHOICE{\cc_1}{\proba}{\cc_2}$ & $\proba \cdot \dwp{\cc_1}{\FF} + (1-\proba )\cdot \dwp{\cc_2}{\FF}$ \\[1.8ex]
            $\ITE{\Bxa}{\cc_1}{\cc_2}$ & $\iverson{\Bxa}\cdot \dwp{\cc_1}{\FF} + \iverson{\neg\Bxa}\cdot \dwp{\cc_2}{\FF} $   \\[1.8ex]
            %
            $\WHILEDO{\Bxa}{\cc'}$ & $\lfp \FG.\, \iverson{\Bxa}\cdot \dwp{\cc'}{\FG} + \iverson{\neg\Bxa}\cdot\FF$  \\
        \end{tabularx}
    \end{center}%
    \caption{Inductive definition of $\dwptrans{\cc}$.}%
    \label{tab:prelim:wp}%
\end{table}%

%% file: appendix.tex
\section{Omitted Proofs and Results on MDPs (\Cref{sec:expected_rewards_def,sec:rewards_lfp,sec:exrew_schedulers})}
\label{app:mdps}

\subsection{Alternative Characterization of Expected Rewards}
We will use the following alternative characterization of the expected reward in some proofs.
\begin{restatable}{lemma}{exrewCharacterization}
    \label{lem:mdp:exrew_characterization}
    Let $\mdp = \mdptuple$ be an MDP, let $\mrew \colon \ms \to \PosRealsInf$ be a reward function, and let $\msched \in \mscheds$.
    For all $\nata \in \Nats$, it holds that
    \[
    \newmexprewb{\msched}{\nata}{\mrew}{\msa} \qeq \sum_{\msa_0\ldots\msa_\natb \in \mpathsb{\nata}(\msa)} \mpathprob{\msched}(\msa_0\ldots\msa_\natb) \cdot  \mrew(\msa_\natb)~.
    \]
    As a consequence,
    \[
    \newmexprew{\msched}{\mrew}{\msa} \qeq \sum_{\msa_0\ldots\msa_\natb \in \mpaths(\msa)} \mpathprob{\msched}(\msa_0\ldots\msa_\natb) \cdot \mrew(\msa_\natb)~.
    \]
\end{restatable}%
\Cref{lem:mdp:exrew_characterization} may strike as rather unintuitive.
A good approach is perhaps to draw an analogy to expected values of discrete random variables.
There one can characterize the expected value $\textsf{E}(X)$ of a discrete random variable $X$ as%
\begin{align*}
	\textsf{E}(X) \eeq \sum_{i = 0}^\omega \textsf{Pr}(X \geq i)~.
\end{align*}%
The (admittedly vague) analogy is that both for paths as for discrete random variables, it suffices to sum over all \enquote{initial segments} (the $i$'s and $\msa_0\ldots \msa_m$'s, respectively) and look at what these contribute to the whole expectation mass.%

\begin{proof}
    Let $\nata \in \Nats$ be arbitrary.
    \begin{align*}
        & \newmexprewb{\msched}{\nata}{\mrew}{\msa}
        \\
        \eeq & \sum_{\msa_0\ldots\msa_\nata \in \mpathseq{\nata}(\msa)} \mpathprob{\msched}(\msa_0\ldots\msa_\nata) \cdot \mrew(\msa_0\ldots\msa_\nata)
        \tag{\Cref{def:mdp:exrew}.\ref{def:mdp:exrew1}}
        \\
        \eeq & \sum_{\msa_0\ldots\msa_\nata \in \mpathseq{\nata}(\msa)} \mpathprob{\msched}(\msa_0\ldots\msa_\nata) \cdot \sum_{\nati=0}^\nata \mrew(\msa_\nati)
        \tag{Definition of $\mrew(\msa_0\ldots\msa_\nata)$}
        \\
        \eeq & \sum_{\nati=0}^\nata \sum_{\msa_0\ldots\msa_\nata \in \mpathseq{\nata}(\msa)} \mpathprob{\msched}(\msa_0\ldots\msa_\nata) \cdot  \mrew(\msa_\nati)
        \tag{Arithmetic}
        \\
        \eeq & \sum_{\nati=0}^\nata \sum_{\msa_0\ldots\msa_\nati \in \mpathseq{\nati}(\msa)}
        \\
        &\qquad \sum_{\msa_\nati\ldots\msa_n \in \mpathseq{\nata-\nati}(\msa_\nati)} \mpathprob{\msched}(\msa_0\ldots\msa_\nati) \cdot  \mrew(\msa_\nati) \cdot \mpathprob{\msched'}(\msa_\nati\ldots\msa_\nata)
        \tag{Split paths; $\msched'$ behaves like $\msched$ after seeing $\msa_0\ldots\msa_\nati$}
        \\
        \eeq & \sum_{\nati=0}^\nata \sum_{\msa_0\ldots\msa_\nati \in \mpathseq{\nati}(\msa)} \mpathprob{\msched}(\msa_0\ldots\msa_\nati) \cdot  \mrew(\msa_\nati) 
        \\
        & \qquad ~\cdot~ \sum_{\msa_\nati\ldots\msa_n \in \mpathseq{\nata-\nati}(\msa_\nati)}  \mpathprob{\msched'}(\msa_\nati\ldots\msa_\nata)
        \tag{Arithmetic}
        \\
        \eeq & \sum_{\nati=0}^\nata \sum_{\msa_0\ldots\msa_\nati \in \mpathseq{\nati}(\msa)} \mpathprob{\msched}(\msa_0\ldots\msa_\nati) \cdot  \mrew(\msa_\nati)
        \tag{Sum of probabilities of paths with fixed length is $1$}
        \\
        \eeq & \sum_{\msa_0\ldots\msa_\natb \in \mpathsb{\nata}(\msa)} \mpathprob{\msched}(\msa_0\ldots\msa_\natb) \cdot  \mrew(\msa_\natb)
        \tag{$\mpathsb{\nata}(\msa) = \biguplus_{\nati=0}^\nata \mpathseq{\nati}(\msa)$}
    \end{align*}
    The second equation follows by taking suprema on both sides.
\end{proof}

\subsection{Induced Markov Chains}\label{app:induced_mc}

Since memoryless schedulers resolve all nondeterministic choices and their choices depend only on the current state, each memoryless scheduler for an MDP induces a \emph{Markov chain} where all actions are determined.\footnote{Restricting to memoryless schedulers suffices for our purposes. One can more generally define Markov chains induced by history-dependent schedulers \cite[Definition 10.92]{BK08}.}%

\begin{definition}[Induced Markov Chains]%
	Let $\mdp = \mdptuple$ be an MDP and let $\msched \in \mmscheds$ be a memoryless scheduler. The \emph{\highlight{MC of $\mdp$ induced by $\msched$}} is given by%
	\[
		\highlight{\mdp(\msched)} \qeq (\ms, \, \mact, \, \mprob_\msched)~,
	\]%
	where the transition probability function $\mprob_\msched$ is given by%
	\[
		\mprob_\msched(\msa,\macta,\msa')  \eeq 
		\begin{cases}
			\mprob(\msa,\macta,\msa'), &\text{if $\macta=  \msched(\msa)$}, \\
			0, & \text{otherwise}~.
		\end{cases}
	\]%
\end{definition}%
\begin{example}[Induced Markov Chains]
    \label{ex:memoryless_scheduler_induced_mc}
    Reconsider the MDP from \Cref{fig:running_example_mdp}.
    Let $\msched$ be a memoryless scheduler satisfying $\msched(s_i) = \actl$ for all $i \in \Nats$.
    Then in the induced MC, only the topmost row is reachable from $s_0$, i.e., the states $s_0,s_0^L,s_1,s_1^L$, etc.%
    \triangleqed%
\end{example}
By construction (cf.~\cite[p. 845]{BK08}), the expected reward of $\mdp$ for scheduler $\msched$ is the same as the (unique) expected reward of $\mdp(\msched)$:
\begin{lemma}
	\label{def:mdp:induced_mc}
	For every MDP $\mdp$, scheduler $\msched$, $n \in \Nats$, initial state $s$, and reward function $\mrew$, we have
	\[
	   \newmexprewb{\msched}{\nata}{\mrew}{\msa}
	   \eeq 
	   \newmmmaxexprewb{\nata}{\mrew}{\msa}{\mdp(\msched)}
	   \eeq 
	   \newmmminexprewb{\nata}{\mrew}{\msa}{\mdp(\msched)} 
	\]
\end{lemma} 
The above lemma justifies that we simply write 
$\newmmcexprewb{\nata}{\mrew}{\msa}{\mdp(\msched)}$ for induced Markov chains $\mdp(\msched)$ instead of 
$\newmmmaxexprewb{\nata}{\mrew}{\msa}{\mdp(\msched)}$ or 
$\newmmminexprewb{\nata}{\mrew}{\msa}{\mdp(\msched)}$.

\subsection{Auxiliary Results on the Existence of Optimal Schedulers}
We require the following auxiliary result: \kb{cite popl24?} The least fixpoint of the \emph{min-}Bellman operator $\mopmins$ gives rise to a memoryless scheduler $\msched$ such that the least fixpoint of $\mopmins$ and the least fixpoint of the Bellman operator $ \mcop{\mdp(\msched)}{\mrew}$ w.r.t.\ the Markov chain $\mdp(\msched)$ induced by $\msched$ (see \Cref{app:induced_mc}) coincide:
\begin{lemma}
    \label{lem:mdp:min_opt_sched}
    Let $\mdp = \mdptuple$ be an MDP, let $\mrew \colon \ms \to \PosRealsInf$ be a reward function, and let $\preccurlyeq$ be some total order on $\mact$. Moreover, define the memoryless scheduler $\msched$ as
    \[
    \msched(\msa) \eeq 
    \mrew(\msa) + \displaystyle \argmin_{\alpha \in \mact(\msa)} 
    \sum_{\msa'\in\msuccs{\macta}(\msa)} \mprob(\msa, \macta, \msa') \cdot \left( \lfp \mopmins\right)(\msa')
    \]
    where in case the above definition does not yield a unique action, we choose the least action w.r.t.\ $\preccurlyeq$. Then, for the Markov chain $\mdp(\msched)$ induced by $\msched$, 
    \[
    \lfp \mopmins \qeq \lfp \mcop{\mdp(\msched)}{\mrew}~.
    \]
\end{lemma}
\begin{proof}
    We prove the two inequalities
    \begin{align*}
        \lfp \mopmins \mmleq \lfp \mcop{\mdp(\msched)}{\mrew}
        \quad\text{and}\quad
        \lfp \mcop{\mdp(\msched)}{\mrew} \mmleq\lfp \mopmins~. 
    \end{align*}
    The claim then follows from antisymmetry of $\mleq$. The first inequality is straightforward since
    \[
    \tforall{\mval \in \mvals}\quad \mopmins(\mval) \mmleq \mcop{\mdp(\msched)}{\mrew}(\mval)~.
    \]
    For the latter inequality, by Park induction (\Cref{lem:prelim:domain_theory:park}), it suffices to show
    \[
    \mcop{\mdp(\msched)}{\mrew} \left( \lfp \mopmins \right) \mmleq \lfp \mopmins ~.
    \]
    For that, we reason as follows:
    \begin{align*}
        &\mcop{\mdp(\msched)}{\mrew} \left( \lfp \mopmins \right) \\
        \eeq &  \mylambda{\msa}
        \mrew(\msa) + \displaystyle \sum_{\msa'\in\msuccs{\msched(\msa)}(\msa)} \mprob(\msa, \msched(\msa), \msa') \cdot \left(\lfp \mopmins\right)(\msa')
        \tag{\Cref{def:mdp:bellman}.\ref{def:mdp:bellman1}} \\
        \eeq & \mylambda{\msa}
        \mrew(\msa) + \displaystyle\min_{\macta \in \mact(\msa)} \sum_{\msa'\in\msuccs{\macta}(\msa)} \mprob(\msa, \macta, \msa') \cdot \left(\lfp \mopmins\right)(\msa')
        \tag{by construction of $\msched$}\\
        \eeq & \lfp \mopmins
        \tag{\Cref{def:mdp:bellman}.\ref{def:mdp:bellman1} and fixpoint property} \\
        \mmleq & \lfp \mopmins~.
        \tag{reflexivity of $\mleq$}
    \end{align*}
    
\end{proof}

\subsection{Proof of \Cref{lem:mdp:bellman-step}}
\label{proof:mdp:bellman-step}
\bellmanStep*
\begin{proof}
We prove the claim for the min-Bellman operator by induction on $\nata$. The proof for the max-Bellman operator is completely analogous. \\

\noindent
\noindent{\emph{Base case} $\nata = 0$.} We have
\begin{align*}
    & \mopminsiter{1}(0) \\
    \eeq & \mylambda{\msa} \mrew(\msa) + \min_{\macta \in \mact(\msa)} \sum_{\msa'\in\msuccs{\macta}(\msa)} \mprob(\msa, \macta, \msa') \cdot 0
    \tag{\Cref{def:mdp:bellman}.\ref{def:mdp:bellman1}}\\
    \eeq & \mylambda{\msa} \mrew(\msa)
    \\
    \eeq & \mylambda{\msa} \inf_{\msched \in \mscheds}  \sum_{\msa_0 \in \mpathseq{0}(\msa)}
    \underbrace{\mpathprob{\msched}(\msa_0)}_{{}= 1}\cdot \mrew(\msa_0)
    \\
    \eeq & \mylambda{\msa}  \inf_{\msched \in \mscheds}  \newmexprewb{\msched}{0}{\mrew}{\msa} ~.
    \tag{\Cref{def:mdp:exrew}.\ref{def:mdp:exrew3}}
\end{align*}
\noindent
\emph{Induction step.} Let $\msa \in \ms$. Given a scheduler $\msched$ and some action $\macta \in \mact(\msa)$, define the scheduler $(\msa,\macta) \cdot \msched$ as
\[
\left((\msa,\macta) \cdot \msched\right)(\msa_0\ldots \msa_\natb)
\eeq
\begin{cases}
    \macta & \text{if $\msa_0\ldots \msa_\natb = \msa$} \\
    \msched(\msa_0\ldots \msa_\natb) & \text{otherwise}~.
\end{cases}
\]
\begin{align*}
    & \mopminsiter{\nata + 2}(0)(\msa) \\
    \eeq &
    \mrew(\msa) + \displaystyle\min_{\macta \in \mact(\msa)} \sum_{\msa'\in\msuccs{\macta}(\msa)} \mprob(\msa, \macta, \msa') \cdot \mopminsiter{\nata + 1}(0)(\msa') 
    \tag{\Cref{def:mdp:bellman}.\ref{def:mdp:bellman1}}\\
    \eeq & \mrew(\msa) + \displaystyle\min_{\macta \in \mact(\msa)} \sum_{\msa'\in\msuccs{\macta}(\msa)} \mprob(\msa, \macta, \msa')
    \cdot \newmminexprewb{\nata}{\mrew}{\msa'}
    \tag{I.H.} \\
    \eeq& \mrew(\msa) + \displaystyle\min_{\macta \in \mact(\msa)} \sum_{\msa'\in\msuccs{\macta}(\msa)} \mprob(\msa, \macta, \msa') \\
    &\quad {}\cdot \left(
    \inf_{\msched \in \mscheds} 
    \sum_{\msa_0\ldots\msa_\natb \in \mpathsb{\nata}(\msa)}
    \mpathprob{\msched}(\msa_0\ldots\msa_\natb )\cdot \mrew(\msa_\natb )\right)
    \tag{\Cref{lem:mdp:exrew_characterization}} \\
    \eeq& \mrew(\msa) + \displaystyle\min_{\macta \in \mact(\msa)} \sum_{\msa'\in\msuccs{\macta}(\msa)} \\
    &\quad {} 
    \inf_{\msched \in \mscheds} 
    \sum_{\msa_0\ldots\msa_\natb \in \mpathsb{\nata}(\msa)}
    \mprob(\msa, \macta, \msa') \cdot
    \mpathprob{\msched}(\msa_0\ldots\msa_\natb )\cdot \mrew(\msa_\natb ) 
    \tag{$\cdot$ distributes over $\inf$ and finite sums} \\
    \eeq& \mrew(\msa) + \displaystyle\min_{\macta \in \mact(\msa)} \inf_{\msched \in \mscheds}  \sum_{\msa'\in\msuccs{\macta}(\msa)} \\
    &\quad {} 
    \sum_{\msa_0\ldots\msa_\natb \in \mpathsb{\nata}(\msa')}
    \mprob(\msa, \macta, \msa') \cdot
    \mpathprob{\msched}(\msa_0\ldots\msa_\natb )\cdot \mrew(\msa_\natb) 
    \tag{the sets $\mpathsb{\nata}(\msa')$ are pairwise disjoint} \\
    \eeq& \displaystyle\min_{\macta \in \mact(\msa)} \inf_{\msched \in \mscheds} 
    \sum_{\msa_0\ldots\msa_\natb \in \mpathsb{\nata+1}(\msa)}
    \mpathprob{(\msa,\macta) \cdot \msched}(\msa_0\ldots\msa_\natb )\cdot \mrew(\msa_\natb) 
    \tag{de-factorize $\macta$-successors; note that this ``consumes'' the $+\mrew(\msa)$} \\
    \eeq&  \inf_{\msched \in \mscheds} 
    \sum_{\msa_0\ldots\msa_\natb \in \mpathsb{\nata+1}(\msa)}
    \mpathprob{\msched}(\msa_0\ldots\msa_\natb )\cdot \mrew(\msa_\natb) 
    \tag{schedulers cover minimum over $\macta$-successors} \\
    \eeq&\newmminexprewb{\nata+1}{\mrew}{\msa}~.
\end{align*}
\tw{do we need a better argument for the ``consumes +$\mrew(\msa)$'' step?}
\end{proof}

\subsection{Proof of \Cref{thm:mdp:bellman_correct}}
\label{proof:mdp:bellman_correct}

We are now ready to prove the following theorem: 

\bellmanCorrect*
\begin{proof}
%
Since the proof of \Cref{thm:mdp:bellman_correct}.\ref{thm:mdp:bellman_correct1} is more involved, we first prove \Cref{thm:mdp:bellman_correct}.\ref{thm:mdp:bellman_correct2}.
For that, consider the following:
\begin{align*}
    &\lfp \mopmaxs \\
    \eeq & \supremum_{\nata \in \Nats} \mopmaxsiter{\nata}(0)
    \tag{\Cref{thm:mdp:bellman_cont}.\ref{thm:mdp:bellman_cont2}}  \\
    \eeq&\supremum_{\nata \in \Nats} \mylambda{\msa} \newmmaxexprewb{\nata}{\mrew}{\msa} 
    \tag{\Cref{lem:mdp:bellman-step}.\ref{lem:mdp:bellman-step2}} \\
    \eeq&\supremum_{\nata \in \Nats} \mylambda{\msa}  \sup_{\msched \in \mscheds}  \newmexprewb{\msched}{\nata}{\mrew}{\msa}
    \tag{\Cref{def:mdp:exrew}.\ref{def:mdp:exrew5}} \\
    \eeq& \mylambda{\msa} \sup_{\nata \in \Nats}  \sup_{\msched \in \mscheds}  \newmexprewb{\msched}{\nata}{\mrew}{\msa}
    \tag{suprema defined pointwise}\\
    \eeq& \mylambda{\msa}  \sup_{\msched \in \mscheds}  \sup_{\nata \in \Nats}   \newmexprewb{\msched}{\nata}{\mrew}{\msa}
    \tag{suprema commute} \\
    \eeq& \mylambda{\msa}  \sup_{\msched \in \mscheds}  \newmexprew{\msched}{\mrew}{\msa}
    \tag{\Cref{def:mdp:exrew}.\ref{def:mdp:exrew2}} \\
    \eeq& \mylambda{\msa}  \newmmaxexprew{\mrew}{\msa} ~.
    \tag{\Cref{def:mdp:exrew}.\ref{def:mdp:exrew6}}
\end{align*}
%
%
In the above reasoning, we have exploited that suprema commute.
This step is what yields the proof of \Cref{thm:mdp:bellman_correct}.\ref{thm:mdp:bellman_correct1} to be more involved.
In this proof, we encounter an analogous situation but instead of swapping two suprema we have to swap a supremum with an infimum.
Suprema and infima do generally not commute.
We will, however, employ \Cref{lem:mdp:min_opt_sched} to see that the suprema and infima considered in this proof \emph{do} commute.
Now consider the following:
\begin{align*}
    &\lfp \mopmins \\
    \eeq & \supremum_{\nata \in \Nats} \mopminsiter{\nata}(0)
    \tag{\Cref{thm:mdp:bellman_cont}.\ref{thm:mdp:bellman_cont2}} \\
    \eeq&\supremum_{\nata \in \Nats} \mylambda{\msa} \newmminexprewb{\nata}{\mrew}{\msa} 
    \tag{\Cref{lem:mdp:bellman-step}.\ref{lem:mdp:bellman-step1}} \\
    \eeq&\supremum_{\nata \in \Nats} \mylambda{\msa}  \inf_{\msched \in \mscheds}  \newmexprewb{\msched}{\nata}{\mrew}{\msa}
    \tag{\Cref{def:mdp:exrew}.\ref{def:mdp:exrew3}} \\
    \eeq& \mylambda{\msa} \sup_{\nata \in \Nats}  \inf_{\msched \in \mscheds}  \newmexprewb{\msched}{\nata}{\mrew}{\msa}
    \tag{suprema defined pointwise}
    %
    %
\end{align*}
Assume for the moment that for all $\msa \in \ms$, we have
\begin{align}
    \label{eqn:thm:mdp:bellman_correct}
    \sup_{\nata \in \Nats}\inf_{\msched \in \mscheds}  \newmexprewb{\msched}{\nata}{\mrew}{\msa} \eeq \inf_{\msched \in \mscheds} \sup_{\nata \in \Nats} \newmexprewb{\msched}{\nata}{\mrew}{\msa} 
\end{align}
With this assumption and the above reasoning, we obtain the desired claim:
\begin{align*}
    &\lfp \mopmins \\
    \eeq & \mylambda{\msa} \inf_{\msched \in \mscheds} \sup_{\nata \in \Nats} \newmexprewb{\msched}{\nata}{\mrew}{\msa} 
    \tag{above reasoning and assumption} \\
    \eeq& \mylambda{\msa}  \newmminexprew{\mrew}{\msa} ~.
    \tag{\Cref{def:mdp:exrew}.\ref{def:mdp:exrew2} and \Cref{def:mdp:exrew}.\ref{def:mdp:exrew4}}
\end{align*}
It hence remains to prove \Cref{eqn:thm:mdp:bellman_correct}. For that, we prove that both
\begin{align}
    \label{eqn:thm:mdp:bellman_correct2}
    \sup_{\nata \in \Nats}\inf_{\msched \in \mscheds}  \newmexprewb{\msched}{\nata}{\mrew}{\msa} \lleq \inf_{\msched \in \mscheds} \sup_{\nata \in \Nats} \newmexprewb{\msched}{\nata}{\mrew}{\msa} 
\end{align}
and
\begin{align}
    \label{eqn:thm:mdp:bellman_correct3}
    \sup_{\nata \in \Nats}\inf_{\msched \in \mscheds}  \newmexprewb{\msched}{\nata}{\mrew}{\msa} \ggeq \inf_{\msched \in \mscheds} \sup_{\nata \in \Nats} \newmexprewb{\msched}{\nata}{\mrew}{\msa}~.
\end{align}
The claim then follows from antisymmetry of $\leq$. Inequality \ref{eqn:thm:mdp:bellman_correct2} holds since $\sup\inf \ldots \leq \inf \sup \ldots$ holds in every complete lattice. For Inequality \ref{eqn:thm:mdp:bellman_correct3}, observe that for every scheduler $\msched'$,
\[
\sup_{\nata \in \Nats} \newmexprewb{\msched'}{\nata}{\mrew}{\msa}
\ggeq
\inf_{\msched \in \mscheds} \sup_{\nata \in \Nats} \newmexprewb{\msched}{\nata}{\mrew}{\msa}~.
\]
Hence, Inequality \ref{eqn:thm:mdp:bellman_correct3} follows if there is a scheduler $\msched'$ with
\[
\sup_{\nata \in \Nats} \newmexprewb{\msched'}{\nata}{\mrew}{\msa} \eeq \sup_{\nata \in \Nats}\inf_{\msched \in \mscheds}  \newmexprewb{\msched}{\nata}{\mrew}{\msa}~.
\]
We claim that $\msched'$ is given by the scheduler constructed in  \Cref{lem:mdp:min_opt_sched}, since
\begin{align*}
    &\sup_{\nata \in \Nats} \newmexprewb{\msched'}{\nata}{\mrew}{\msa}  \\
    \eeq & \sup_{\nata \in \Nats} \newmmcexprewb{\nata}{\mrew}{\msa}{\mdp(\msched')} 
    \tag{\Cref{def:mdp:induced_mc}} \\
    \eeq & \newmmcexprew{\mrew}{\msa}{\mdp(\msched')}
    \tag{\Cref{def:mdp:exrew}.\ref{def:mdp:exrew2}} \\
    \eeq & \lfp \mcop{\mdp(\msched')}{\mrew}
    \tag{\Cref{thm:mdp:bellman_correct}} \\
    \eeq& \lfp \mopmins
    \tag{\Cref{lem:mdp:min_opt_sched}} \\
    \eeq & \sup_{\nata \in \Nats} \mopminsiter{\nata}(0)
    \tag{\Cref{thm:prelim:domain_theory:kleene} (Kleene)} \\
    \eeq & \sup_{\nata \in \Nats}\inf_{\msched \in \mscheds}  \newmexprewb{\msched}{\nata}{\mrew}{\msa}~.
    \tag{\Cref{lem:mdp:bellman-step}.\ref{lem:mdp:bellman-step1}}
\end{align*}
This completes the proof.
\end{proof}

\section{Omitted Results on the Operational MDP Semantics of $\pgcl$}
\label{app:prelim:wp}

\begin{lemma}
    \label{lem:app:prelim:wp:bellman_compliant_tick}
    Let $\cc\in\pgcl$ and $\FF\in\E$. We have:
    \begin{enumerate}
        \item 	
        \begin{align*}
            &\dwp{\cc}{\FF} \eeq \mylambda{\statea}
            \torew{\FF}\big((\cc,\statea)\big) \\
            &\qquad ~+~ \min_{\macta \in \mact(\cc,\statea)} \sum_{\execrel{(\cc,\statea)}{\macta}{\proba}{\confa'}} \proba \cdot
            \begin{cases}
                \FF(\stateb) & \text{if $\confa' = (\term,\stateb)$} \\
                \dwp{\cc'}{\FF}(\statea') & \text{if $\confa' = (\cc',\statea')$}
            \end{cases}
        \end{align*}
        \item 
        \begin{align*}
            \displaystyle
            &\awp{\cc}{\FF} \eeq
            \mylambda{\statea} \torew{\FF}\big((\cc,\statea)\big) \\
            &\qquad ~+~
            \max_{\macta \in \mact(\cc,\statea)} \sum_{\execrel{(\cc,\statea)}{\macta}{\proba}{\confa'}} \proba \cdot
            \begin{cases}
                \FF(\stateb) & \text{if $\confa' = (\term,\stateb)$} \\
                \awp{\cc'}{\FF}(\statea') & \text{if $\confa' = (\cc',\statea')$}
            \end{cases}
        \end{align*}
    \end{enumerate}
    Note that in the above case distinctions, we do not need to consider the case $\confa' = \opsink$ because $\opsink$ is not reachable from a configuration of the form $(\cc,\statea)$ in one step.
\end{lemma}
\begin{proof}
    By induction on $\cc$. We prove the claim for $\dwpsymbol$. The reasoning for $\awpsymbol$ is completely analogous. Now let $\statea\in\States$. \\

    \noindent
    \emph{The case $\cc = \SKIP$.} We have
    \begin{align*}
        & \dwp{\SKIP}{\FF}(\statea) \\
        \eeq & \FF(\statea)
        \tag{\Cref{tab:prelim:wp}} \\
        \eeq& \min_{\macta \in \mact(\cc,\statea)} \sum_{\execrel{(\cc,\statea)}{\macta}{\proba}{\confa'}} \proba \cdot
        \begin{cases}
            \FF(\stateb) & \text{if $\confa' = (\term,\stateb)$} \\
            \dwp{\cc'}{\FF}(\statea') & \text{if $\confa' = (\cc',\statea')$}~.
        \end{cases}
        \tag{$\mact(\SKIP,\statea) = \{\actn\}$ and $\execrel{(\SKIP,\statea)}{\actn}{1}{(\term,\statea)}$ by \Cref{fig:prelim:op_rules_tick}} \\
        \eeq& \torew{\FF}((\cc,\statea)) + \min_{\macta \in \mact(\cc,\statea)} \sum_{\execrel{(\cc,\statea)}{\macta}{\proba}{\confa'}} \proba \cdot
        \begin{cases}
            \FF(\stateb) & \text{if $\confa' = (\term,\stateb)$} \\
            \dwp{\cc'}{\FF}(\statea') & \text{if $\confa' = (\cc',\statea')$}~.
        \end{cases}
        \tag{Because $\torew{\FF}((\SKIP,\statea)) = 0$ by \Cref{def:oprew,def:torew}}
    \end{align*}
    \noindent
    \emph{The case $\cc = \ASSIGN{\vara}{\Axa}.$} We have 
    \begin{align*}
        & \dwp{\ASSIGN{\vara}{\Axa}}{\FF}(\statea) \\
        \eeq & \FF(\statea\statesubst{\vara}{\Axa(\statea)})
        \tag{\Cref{tab:prelim:wp}} \\
        \eeq& \min_{\macta \in \mact(\cc,\statea)} \sum_{\execrel{(\cc,\statea)}{\macta}{\proba}{\confa'}} \proba \cdot
        \begin{cases}
            \FF(\stateb) & \text{if $\confa' = (\term,\stateb)$} \\
            \dwp{\cc'}{\FF}(\statea') & \text{if $\confa' = (\cc',\statea')$}~.
        \end{cases}
        \tag{$\mact(\ASSIGN{\vara}{\Axa},\statea) = \{\actn\}$ and $	\execrel{(\ASSIGN{x}{\Axa},\statea)}{\actn}{1}{(\term,\statea\statesubst{x}{\Axa(\statea)})}$ by \Cref{fig:prelim:op_rules_tick}} \\
        \eeq& \torew{\FF}((\cc,\statea)) +  \min_{\macta \in \mact(\cc,\statea)} \sum_{\execrel{(\cc,\statea)}{\macta}{\proba}{\confa'}} \proba \cdot
        \begin{cases}
            \FF(\stateb) & \text{if $\confa' = (\term,\stateb)$} \\
            \dwp{\cc'}{\FF}(\statea') & \text{if $\confa' = (\cc',\statea')$}~.
        \end{cases}
        \tag{Because $\torew{\FF}((\SKIP,\statea)) = 0$ by \Cref{def:oprew,def:torew}}
    \end{align*}
    \noindent
    \emph{The case $\cc = \TICK{\tickrew}.$} We have 
    \begin{align*}
        & \dwp{\TICK{\tickrew}}{\FF}(\statea) \\
        \eeq & \tickrew + \FF(\statea) 
        \tag{\Cref{tab:prelim:wp}} \\
        \eeq& \tickrew + \min_{\macta \in \mact(\cc,\statea)} \sum_{\execrel{(\cc,\statea)}{\macta}{\proba}{\confa'}} \proba \cdot
        \begin{cases}
            \FF(\stateb) & \text{if $\confa' = (\term,\stateb)$} \\
            \dwp{\cc'}{\FF}(\statea') & \text{if $\confa' = (\cc',\statea')$}~.
        \end{cases}
        \tag{$\mact(\TICK{\tickrew},\statea) = \{\actn\}$ and $	\execrel{(\TICK{\tickrew},\statea)}{\actn}{1}{(\term,\statea)}$ by \Cref{fig:prelim:op_rules_tick}} \\
        \eeq& \torew{\FF}((\cc,\statea)) +  \min_{\macta \in \mact(\cc,\statea)} \sum_{\execrel{(\cc,\statea)}{\macta}{\proba}{\confa'}} \proba \cdot
        \begin{cases}
            \FF(\stateb) & \text{if $\confa' = (\term,\stateb)$} \\
            \dwp{\cc'}{\FF}(\statea') & \text{if $\confa' = (\cc',\statea')$}~.
        \end{cases}
        \tag{Because $\torew{\FF}((\TICK{\tickrew},\statea)) = \tickrew$ by \Cref{def:oprew,def:torew}}
    \end{align*}
    \noindent
    \emph{The case $\cc = \COMPOSE{\cc_1}{\cc_2}$.} First observe that by the rules in \Cref{fig:prelim:op_rules_tick} we encounter exactly one of the following cases:
    {\normalsize
        \begin{enumerate}
            \item We have $\mact(\cc_1, \statea) = \{\actn\}$ and $\execrel{\cc_1,\statea}{\actn}{1}{\term,\stateb}$. Hence, also $\mact(\COMPOSE{\cc_1}{\cc_2}, \statea) = \{\actn\}$ and $\execrel{(\COMPOSE{\cc_1}{\cc_2},\statea)}{\actn}{1}{(\cc_2,\stateb)}$.
            \item For every $\macta \in \mact(\cc_1, \statea)$, all configurations in $\msuccs{\macta}(\cc_1, \statea)$ are of the form $(\cc_1',\statea')$. Hence, for every $\macta \in \mact(\COMPOSE{\cc_1}{\cc_2}, \statea)$, all configurations in $\msuccs{\macta}(\COMPOSE{\cc_1}{\cc_2}, \statea)$ are of the form $(\COMPOSE{\cc_1'}{\cc_2}, \statea')$.\color{white}
        \end{enumerate}
    }
    We proceed by distinguishing these two cases. For the first case, we have 
    \begin{align*}
        &\dwp{ \COMPOSE{\cc_1}{\cc_2}}{\FF} (\statea) \\
        \eeq & \dwp{\cc_1}{\dwp{\cc_2}{\FF}}(\statea)
        \tag{\Cref{tab:prelim:wp}} \\
        \eeq& \torew{\dwp{\cc_2}{\FF}}\big((\cc_1, \statea)\big) \\
        & {+} \min_{\macta \in \mact(\cc_1,\statea)} \sum_{\execrel{(\cc_1,\statea)}{\macta}{\proba}{\confa'}} \proba \cdot
        \begin{cases}
            \dwp{\cc_2}{\FF}(\stateb) & \text{if $\confa' = (\term,\stateb)$} \\
            \dwp{\cc'}{\dwp{\cc_2}{\FF}}(\statea') & \text{if $\confa' = (\cc',\statea')$}~.
        \end{cases}
        \tag{I.H.} \\
        \eeq& \torew{\dwp{\cc_2}{\FF}}\big((\cc_1, \statea)\big) + \dwp{\cc_2}{\FF}(\stateb) 
        \tag{assumption that $\mact(\cc_1, \statea) = \{\actn\}$ and $\execrel{\cc_1,\statea}{\actn}{1}{\term,\stateb}$} \\
        \eeq& \torew{\FF}\big((\COMPOSE{\cc_1}{\cc_2}, \statea)\big) + \dwp{\cc_2}{\FF}(\stateb) 
        \tag{\Cref{def:oprew,def:torew}} \\
        \eeq& \torew{\FF}\big((\COMPOSE{\cc_1}{\cc_2}, \statea)\big) \\
        & + \min_{\macta \in \mact(\COMPOSE{\cc_1}{\cc_2},\statea)} \sum_{\execrel{(\COMPOSE{\cc_1}{\cc_2},\statea)}{\macta}{\proba}{\confa'}} \proba \cdot
        \begin{cases}
            \FF(\stateb)& \text{if $\confa' = (\term,\stateb)$} \\
            \dwp{\cc'}{\FF}(\statea') & \text{if $\confa' = (\cc',\statea')$}~.
        \end{cases}
        \tag{assumption that $\mact(\COMPOSE{\cc_1}{\cc_2}, \statea) = \{\actn\}$ and $\execrel{(\COMPOSE{\cc_1}{\cc_2},\statea)}{\actn}{1}{(\cc_2,\stateb)}$}
    \end{align*}
    For the second case, we have 
    \begin{align*}
        &\dwp{ \COMPOSE{\cc_1}{\cc_2}}{\FF} (\statea) \\
        \eeq & \dwp{\cc_1}{\dwp{\cc_2}{\FF}}(\statea)
        \tag{\Cref{tab:prelim:wp}} \\
        \eeq& \torew{\dwp{\cc_2}{\FF}}\big((\cc_1, \statea)\big) \\
        &+ \min_{\macta \in \mact(\cc_1,\statea)} \sum_{\execrel{(\cc_1,\statea)}{\macta}{\proba}{\confa'}} \proba \cdot
        \begin{cases}
            \dwp{\cc_2}{\FF}(\stateb) & \text{if $\confa' = (\term,\stateb)$} \\
            \dwp{\cc'}{\FF}(\statea') & \text{if $\confa' = (\cc',\statea')$}~.
        \end{cases}
        \tag{I.H.} \\
        \eeq& \torew{\dwp{\cc_2}{\FF}}\big((\cc_1, \statea)\big) \\
        &+  \min_{\macta \in \mact(\cc_1,\statea)} \sum_{\execrel{(\cc_1,\statea)}{\macta}{\proba}{(\cc_1',\statea')}} \proba \cdot
        \dwp{\cc_1'}{\dwp{\cc_2}{\FF}}(\statea') 
        \tag{assumption} \\
        \eeq & \torew{\FF}\big((\COMPOSE{\cc_1}{\cc_2}, \statea)\big) \\
        &+  \min_{\macta \in \mact(\cc_1,\statea)} \sum_{\execrel{(\cc_1,\statea)}{\macta}{\proba}{(\cc_1',\statea')}} \proba \cdot
        \dwp{\COMPOSE{\cc_1'}{\cc_2}}{\FF}(\statea') 
        \tag{\Cref{def:oprew,def:torew} and \Cref{tab:prelim:wp}}\\
        \eeq &\torew{\FF}\big((\COMPOSE{\cc_1}{\cc_2}, \statea)\big) \\
        &+ \min_{\macta \in \mact(\COMPOSE{\cc_1}{\cc_2},\statea)} \sum_{\execrel{(\COMPOSE{\cc_1}{\cc_2},\statea)}{\macta}{\proba}{\confa'}} \proba \cdot
        \begin{cases}
            \FF(\stateb)& \text{if $\confa' = (\term,\stateb)$} \\
            \dwp{\cc'}{\FF}(\statea') & \text{if $\confa' = (\cc',\statea')$}~.
        \end{cases}
        \tag{assumption}
    \end{align*}
    The cases $\cc = \PCHOICE{\cc_1}{\proba}{\cc_2}$, $\cc = \NDCHOICE{\cc_1}{\cc_2}$, and $\cc = \IFSYMBOL \ldots$ follow immediately from the I.H. \\
    \tw{still ok? Also note that $\torew{\FF}\big((\cc, \statea)\big) = 0$ for these programs.}

    \noindent
    \emph{The case $\cc = \WHILEDO{\Bxa}{\cc'}$.} Since $\dwp{ \WHILEDO{\Bxa}{\cc'}}{\FF}$ is a fixpoint of $\dcharfun{\cc}{\FF}$, we have 
    \begin{align}
        \label{eqn:lem:app:prelim:wp:bellman_compliant1}
        \dwp{\WHILEDO{\Bxa}{\cc'}}{\FF} \eeq \iverson{\Bxa}\cdot \dwp{\COMPOSE{\cc'}{\WHILEDO{\Bxa}{\cc'}}}{\FF} + \iverson{\neg\Bxa}\cdot \FF~.
    \end{align}
    Now distinguish the cases $\statea\models\Bxa$ and $\statea\models\neg\Bxa$. If $\statea\models\neg\Bxa$, then 
    \begin{align*}
        &\dwp{\WHILEDO{\Bxa}{\cc'}}{\FF}(\statea) \\
        \eeq &\FF(\statea)
        \tag{\Cref{eqn:lem:app:prelim:wp:bellman_compliant1}} \\
        \eeq& \min_{\macta \in \mact(\WHILEDO{\Bxa}{\cc'},\statea)} \torew{\FF}\big((\cc,\statea)\big) \\
        &\qquad ~+~ \sum_{\execrel{(\WHILEDO{\Bxa}{\cc'},\statea)}{\macta}{\proba}{\confa'}} \proba \cdot
        \begin{cases}
            \FF(\stateb) & \text{if $\confa' = (\term,\stateb)$} \\
            \dwp{\cc''}{\FF}(\statea') & \text{if $\confa' = (\cc'',\statea')$}~,
        \end{cases}
    \end{align*}
    where the last equation holds since $\torew{\FF}\big((\WHILEDO{\Bxa}{\cc'},\statea)\big) = 0$ and $\statea\not\models\Bxa$ implies $\mact(\WHILEDO{\Bxa}{\cc'},\statea) = \{\actn\}$ and $\execrel{(\WHILEDO{\Bxa}{\cc'},\statea)}{\actn}{1}{(\term,\statea)}$ by \Cref{fig:prelim:op_rules_tick}.
    
    If $\statea\models\Bxa$, then 
    \begin{align*}
        &\min_{\macta \in \mact(\WHILEDO{\Bxa}{\cc'},\statea)} \torew{\FF}((\cc,\statea)) \\
        &\qquad~+~ \sum_{\execrel{(\WHILEDO{\Bxa}{\cc'},\statea)}{\macta}{\proba}{\confa'}} \proba \cdot
        \begin{cases}
            \FF(\stateb) & \text{if $\confa' = (\term,\stateb)$} \\
            \dwp{\cc''}{\FF}(\statea') & \text{if $\confa' = (\cc'',\statea')$}
        \end{cases} \\
        \eeq & \dwp{\COMPOSE{\cc'}{\WHILEDO{\Bxa}{\cc'}}}{\FF}(\statea)
        \tag{see below}\\
        \eeq & \dwp{\WHILEDO{\Bxa}{\cc'}}{\FF}(\statea)~,
        \tag{\Cref{eqn:lem:app:prelim:wp:bellman_compliant1}}
    \end{align*}
    where the one but last step holds since $\torew{\FF}((\WHILEDO{\Bxa}{\cc'},\statea)) = 0$ and $\statea\models\Bxa$ implies $\mact(\WHILEDO{\Bxa}{\cc'},\statea) = \{\actn\}$ $\execrel{(\WHILEDO{\Bxa}{\cc},\statea)}{\actn}{1}{(\COMPOSE{\cc'}{\WHILEDO{\Bxa}{\cc'}},\statea)}$ by \Cref{fig:prelim:op_rules_tick}.
    
    Now distinguish the cases $\statea\models\Bxa$ and $\statea\models\neg\Bxa$. If $\statea\models\neg\Bxa$, then 
    \begin{align*}
        &\dwp{\WHILEDO{\Bxa}{\cc'}}{\FF}(\statea) \\
        \eeq &\FF(\statea)
        \tag{\Cref{eqn:lem:app:prelim:wp:bellman_compliant1}} \\
        \eeq &\torew{\FF}(\WHILEDO{\Bxa}{\cc'}) + \FF(\statea)
        \tag{$\torew{\FF}(\WHILEDO{\Bxa}{\cc'}) = 0$ by \Cref{def:oprew,def:torew}} \\
        \eeq& \torew{\FF}(\WHILEDO{\Bxa}{\cc'}) + 
        \min_{\macta \in \mact(\WHILEDO{\Bxa}{\cc'},\statea)}\\ &\quad \sum_{\execrel{(\WHILEDO{\Bxa}{\cc'},\statea)}{\macta}{\proba}{\confa'}} \proba \cdot
        \begin{cases}
            \FF(\stateb) & \text{if $\confa' = (\term,\stateb)$} \\
            \dwp{\cc'}{\FF}(\statea') & \text{if $\confa' = (\cc',\statea')$}~.
        \end{cases}
        \tag{$\mact(\WHILEDO{\Bxa}{\cc'},\statea) = \{\actn\}$, $\execrel{(\WHILEDO{\Bxa}{\cc},\statea)}{\actn}{1}{(\term,\statea)}$}
    \end{align*}
\end{proof}

\begin{lemma}
	\label{app:lem:op_decomp}
	Let $\someopsymbol \in \{\dopsymbol,\aopsymbol\}$. For all $\cc_1,\cc_2\in\pgcl$ and all $\FF\in \E$, we have
	\begin{align*}
		\someop{\cc_1}{\someop{\cc_2}{\FF}} \eeleq \someop{\COMPOSE{\cc_1}{\cc_2}}{\FF}~.
	\end{align*}
\end{lemma}
\begin{proof}
    \newcommand{\phiabbr}{\Phi}%
    \newcommand{\rewabbr}{\torew{}}%
    To enhance readability, we write $\phiabbr = \mopmin{\opmdp}{\torew{\dop{\cc_2}{\FF}}}$ and $\rewabbr = \torew{\dop{\cc_2}{\FF}}$ throughout this proof.
    We prove the claim for $\dopsymbol$. The reasoning for $\aopsymbol$ is completely analogous.
    First notice that by \Cref{def:prelim:wp:op} on page \pageref{def:prelim:wp:op} and \Cref{thm:mdp:bellman_cont} on page \pageref{thm:mdp:bellman_cont}, we have for all $\statea\in\States$,
    \begin{align*}
        &\dop{\cc_1}{\dop{\cc_2}{\FF}}(\statea)\\
        \eeq & \big( \bigsqcup \setcomp{\phiabbr^n(0)}{n\in\Nats}\big)\big((\cc_1,\statea)\big)~.
    \end{align*}
    It hence suffices to prove that for all $n \in \Nats$, all $\cc_1\in\pgcl$, all $\FF \in \E$, and all $\statea\in\States$,
    \[
    \phiabbr^n(0)\big((\cc_1,\statea)\big) \lleq \dop{\COMPOSE{\cc_1}{\cc_2}}{\FF}(\statea)~.
    \]
    We proceed by induction on $n$. The base case $n=0$ is trivial. For the induction step, consider the following:
    \begin{align*}
        & \phiabbr^{n+1}(0)\big((\cc_1,\statea)\big) \\
        \eeq &\phiabbr \big( \phiabbr^n(0) \big) \big((\cc_1,\statea)\big) \\
        \eeq & \rewabbr\big((\cc_1,\statea)\big) + \min_{\macta \in \mact(\cc_1,\statea)} \Big[ \\ & \sum_{\substack{\execrel{(\cc_1,\statea)}{\macta}{\proba}{(\cc', \statea')} \\ \cc' \neq \term}} \proba \cdot \phiabbr^n(0)\big((\cc',\statea')\big)
        +
        \sum_{\substack{\execrel{(\cc_1,\statea)}{\macta}{\proba}{(\term, \stateb)}}}  \proba \cdot \phiabbr^n(0)\big((\term,\stateb)\big) \Big]
        \tag{\Cref{def:prelim:opmdp} and \Cref{def:mdp:bellman}} \\
        \lleq & \rewabbr\big((\cc_1,\statea)\big) + \min_{\macta \in \mact(\cc_1,\statea)} \Big[
        \\
        &\qquad \sum_{\substack{\execrel{(\cc_1,\statea)}{\macta}{\proba}{(\cc', \statea')}
        \\
        \cc' \neq \term}} \proba \cdot \dop{\COMPOSE{\cc'}{\cc_2}}{\FF}(\statea')
        \\
        & \qquad\qquad + \sum_{\substack{\execrel{(\cc_1,\statea)}{\macta}{\proba}{(\term, \stateb)}}}  \proba \cdot \dop{\cc_2}{\FF}(\stateb) \Big]
        \tag{I.H.\ and using that $\phiabbr^n(0)\big((\term,\stateb)\big) \leq \rewabbr\big((\term,\stateb)\big) = \dop{\cc_2}{\FF}(\stateb)$} \\
        \eeq & \torew{\FF}\big((\COMPOSE{\cc_1}{\cc_2},\statea)\big) + \min_{\macta \in \mact(\COMPOSE{\cc_1}{\cc_2},\statea)} \Big[ \\
        &\qquad \sum_{\substack{\execrel{(\COMPOSE{\cc_1}{\cc_2},\statea)}{\macta}{\proba}{(\COMPOSE{\cc'}{\cc_2}, \statea')} \\ \cc' \neq \term}} \proba \cdot \dop{\COMPOSE{\cc'}{\cc_2}}{\FF}(\statea')
        \\
        &\qquad\qquad +\sum_{\substack{\execrel{(\COMPOSE{\cc_1}{\cc_2},\statea)}{\macta}{\proba}{(\cc_2, \stateb)}}}  \proba \cdot \dop{\cc_2}{\FF}(\stateb) \Big]
        \tag{\Cref{def:prelim:opmdp}} \\
        \eeq & \dop{\COMPOSE{\cc_1}{\cc_2}}{\FF} \tag{\Cref{def:mdp:bellman}}
        ~.
    \end{align*}
\end{proof}